\documentclass[a4paper,oneside,reqno,11pt]{amsart}
\usepackage[margin=1in,includefoot]{geometry}
\usepackage[utf8]{inputenc}
\usepackage[T1]{fontenc}
\usepackage{lmodern,microtype}
\usepackage{amsthm,amsmath,amssymb,colonequals}
\usepackage[american]{babel}
\usepackage{csquotes}
\usepackage{nicefrac}
\usepackage[mathcal]{euscript}
\usepackage{enumitem}
\usepackage{tikz}
\usepackage[unicode=true,pdftex]{hyperref}
\hypersetup{
colorlinks=true,
linkcolor=black,
citecolor=black,
filecolor=black,
urlcolor=black,
pdfauthor={Bartłomiej Bosek, Anna Zych-Pawlewicz},
pdftitle={Recoloring Interval Graphs with Limited Recourse Budget},
pdfsubject={Theoretical Computer Science},
pdfkeywords={algorithm, orientation, online, incremental, dynamical, tree},
pdfproducer={LaTeX},
pdfcreator={pdfLaTeX}
}
\usepackage{bookmark}
\bookmarksetup{open,numbered,depth=5}

\usepackage[square,numbers]{natbib}

\usepackage{amsaddr}

\usepackage[overload]{empheq}
\usepackage{amsmath,stmaryrd,graphicx}
\usepackage[linesnumbered,ruled,vlined]{algorithm2e}
\usepackage{mathtools}
\usepackage{xcolor}

\usepackage{tikz}
\usetikzlibrary{backgrounds}
\usetikzlibrary{arrows}
\usetikzlibrary{shapes,shapes.geometric,shapes.misc}

\tikzstyle{tikzfig}=[baseline=-0.25em,scale=0.5]

\pgfkeys{/tikz/tikzit fill/.initial=0}
\pgfkeys{/tikz/tikzit draw/.initial=0}
\pgfkeys{/tikz/tikzit shape/.initial=0}
\pgfkeys{/tikz/tikzit category/.initial=0}

\pgfdeclarelayer{edgelayer}
\pgfdeclarelayer{nodelayer}
\pgfsetlayers{background,edgelayer,nodelayer,main}

\tikzstyle{none}=[inner sep=0mm]

\newcommand{\tikzfig}[1]{{\tikzstyle{every picture}=[tikzfig]
\IfFileExists{#1.tikz}
  {\input{#1.tikz}}
  {\IfFileExists{./figures/#1.tikz}
      {\input{./figures/#1.tikz}}
      {\tikz[baseline=-0.5em]{\node[draw=red,font=\color{red},fill=red!10!white] {\textit{#1}};}}}}}

\tikzstyle{every loop}=[]
 \usepackage{tikzscale}
\usetikzlibrary{calc,decorations.pathreplacing}
\tikzstyle{every picture}=[tikzfig]
\tikzset{x=1pt,y=1pt}

\definecolor{myviolet}{RGB}{212,17,196}
\definecolor{mynavy}{RGB}{117,15,212}
\definecolor{myorange}{RGB}{196,95,44}
\definecolor{mygreen}{RGB}{10,164,17}

\newcommand{\tikzInterval}{{\begin{tikzpicture}[xscale={50}, yscale=15]
	\begin{pgfonlayer}{nodelayer}
		\node [style=none] (0) at (0, 0) {};
		\node [style=none] (1) at (1, 0) {};
	\end{pgfonlayer}
	\begin{pgfonlayer}{edgelayer}
		\draw [style=interval] (0.center) to (1.center);
	\end{pgfonlayer}
\end{tikzpicture}
 }}
\newcommand{\tikzIntDashedBlue}{{\begin{tikzpicture}[xscale={50}, yscale=15]
	\begin{pgfonlayer}{nodelayer}
		\node [style=none] (0) at (0, 0) {};
		\node [style=none] (1) at (1, 0) {};
	\end{pgfonlayer}
	\begin{pgfonlayer}{edgelayer}
		\draw [style=intDashedBlue] (0.center) to (1.center);
	\end{pgfonlayer}
\end{tikzpicture}
 }}
\newcommand{\tikzIntThick}{{\begin{tikzpicture}[xscale={50}, yscale=15]
	\begin{pgfonlayer}{nodelayer}
		\node [style=none] (0) at (0, 0) {};
		\node [style=none] (1) at (1, 0) {};
	\end{pgfonlayer}
	\begin{pgfonlayer}{edgelayer}
		\draw [style=interval] (0.center) to (1.center);
		\draw [style=intThick] (0.center) to (1.center);
	\end{pgfonlayer}
\end{tikzpicture}
 }}
\newcommand{\tikzIntThickDashedBlue}{{\begin{tikzpicture}[xscale={50}, yscale=15]
	\begin{pgfonlayer}{nodelayer}
		\node [style=none] (0) at (0, 0) {};
		\node [style=none] (1) at (1, 0) {};
	\end{pgfonlayer}
	\begin{pgfonlayer}{edgelayer}
		\draw [style=intThickDashedBlue] (0.center) to (1.center);
		\draw [style=intDashedBlue] (0.center) to (1.center);
	\end{pgfonlayer}
\end{tikzpicture}
 }}
\newcommand{\tikzAreaGreen}{{\begin{tikzpicture}[xscale={50}, yscale=15]
	\begin{pgfonlayer}{nodelayer}
		\node [style=none] (0) at (0.25, -0.5) {};
		\node [style=none] (1) at (0, 0) {};
		\node [style=none] (2) at (0.25, 0.5) {};
		\node [style=none] (3) at (0.75, 0.5) {};
		\node [style=none] (4) at (1, 0) {};
		\node [style=none] (5) at (0.75, -0.5) {};
	\end{pgfonlayer}
	\begin{pgfonlayer}{edgelayer}
		\draw [style=areaGreen] (0.center) to (5.center);
		\draw [style=areaGreen, in=-90, out=0, looseness=1.25] (5.center) to (4.center);
		\draw [style=areaGreen, in=360, out=90, looseness=1.25] (4.center) to (3.center);
		\draw [style=areaGreen] (3.center) to (2.center);
		\draw [style=areaGreen, in=90, out=-180, looseness=1.25] (2.center) to (1.center);
		\draw [style=areaGreen, in=180, out=-90, looseness=1.25] (1.center) to (0.center);
	\end{pgfonlayer}
\end{tikzpicture}
 }}

\tikzstyle{intBrace}=[decorate, decoration={brace,amplitude=10pt,raise=1pt,mirror}, thick]
\tikzstyle{intBraceGreen}=[decorate, decoration={brace,amplitude=10pt,raise=1pt,mirror}, draw={rgb,255: red,10; green,164; blue,17}, thick]
\tikzstyle{textRotated}=[rotate=45]

\tikzstyle{interval}=[{|-|}, fill=none, draw=black, thick]
\tikzstyle{intRed}=[{|-|}, fill=none, draw={rgb,255: red,216; green,29; blue,19}, thick]
\tikzstyle{intGreen}=[{|-|}, fill=none, draw={rgb,255: red,10; green,164; blue,17}, thick]
\tikzstyle{intYellow}=[{|-|}, fill=none, draw={rgb,255: red,245; green,168; blue,12}, thick]
\tikzstyle{intBlue}=[{|-|}, fill=none, draw={rgb,255: red,22; green,39; blue,233}, thick]
\tikzstyle{intViolet}=[{|-|}, fill=none, draw={rgb,255: red,212; green,17; blue,196}, thick]
\tikzstyle{intThick}=[{-}, fill=none, draw=black, very thick]
\tikzstyle{intThickDashedBlue}=[{-}, dashed, fill=none, draw={rgb,255: red,22; green,39; blue,233}, very thick]
\tikzstyle{intDashed}=[{|-|}, dashed, fill=none, draw=black, thick]
\tikzstyle{intDashedViolet}=[{|-|}, dashed, fill=none, draw={rgb,255: red,212; green,17; blue,196}, thick]
\tikzstyle{intDashedGreen}=[{|-|}, dashed, fill=none, draw={rgb,255: red,10; green,164; blue,17}, thick]
\tikzstyle{intDashedBlue}=[{|-|}, dashed, fill=none, draw={rgb,255: red,22; green,39; blue,233}, thick]
\tikzstyle{areaViolet}=[-, fill=none, draw={rgb,255: red,212; green,17; blue,196}, thick]
\tikzstyle{areaNavy}=[-, fill=none, draw={rgb,255: red,117; green,15; blue,212}, thick]
\tikzstyle{areaOrange}=[-, fill=none, draw={rgb,255: red,255; green,124; blue,58}, thick]
\tikzstyle{areaGreen}=[-, fill=none, draw={rgb,255: red,10; green,164; blue,17}, thick]
\tikzstyle{arrDashedRed}=[->, dashed, fill=none, draw={rgb,255: red,216; green,29; blue,19}, thick]
\tikzstyle{arrDashedGreen}=[->, dashed, fill=none, draw={rgb,255: red,10; green,164; blue,17}, thick]
\tikzstyle{arrDashedYellow}=[->, dashed, fill=none, draw={rgb,255: red,245; green,168; blue,12}, thick]
\tikzstyle{arrDashedBlue}=[->, dashed, fill=none, draw={rgb,255: red,22; green,39; blue,233}, thick]
\tikzstyle{arrDashedViolet}=[->, dashed, fill=none, draw={rgb,255: red,212; green,17; blue,196}, thick]
\tikzstyle{fillGray}=[-, fill={rgb,255: red,224; green,224; blue,224}, draw=none]
\tikzstyle{fillBlue}=[-, fill={rgb,255: red,215; green,222; blue,255}, draw=none]

\usepackage{xargs}
\usepackage[colorinlistoftodos,prependcaption,textsize=small]{todonotes}
\presetkeys{todonotes}{inline}{}

\newcommandx{\change}[2][1=]{\todo[linecolor=blue,backgroundcolor=blue!25,bordercolor=blue,#1]{#2}}
\newcommandx{\info}[2][1=]{\todo[linecolor=green,backgroundcolor=green!25,bordercolor=green,#1]{#2}}

\newcommand*{\MyDef}{\mathrm{def}}
\newcommand*{\eqdef}{\ensuremath{\mathop{\overset{\MyDef}{=}}}}

\renewcommand{\leq}{\leqslant}
\renewcommand{\geq}{\geqslant}

\newcommand{\bra}[1]{\ensuremath{({#1})}}
\newcommand{\Bra}[1]{\ensuremath{\left({#1}\right)}}
\newcommand{\size}[1]{\ensuremath{{|}{#1}{|}}}
\newcommand{\Abs}[1]{\ensuremath{\left|{#1}\right|}}
\newcommand{\set}[1]{\ensuremath{\{{#1}\}}}

\newcommand{\tuple}[1]{\ensuremath{\langle\hspace{0.5mm}{#1}\hspace{0.5mm}\rangle}}
\newcommand{\num}[1]{\ensuremath{\left[{#1}\right]}}
\newcommand{\ceil}[1]{\ensuremath{\left\lceil{#1}\right\rceil}}
\newcommand{\floor}[1]{\ensuremath{\left\lfloor{#1}\right\rfloor}}

\newcommand{\NaturalNumbers}{\ensuremath{\mathbb{N}}}
\newcommand{\IntegerNumbers}{\ensuremath{\mathbb{Z}}}
\newcommand{\RealNumbers}{\ensuremath{\mathbb{R}}}

\newcommand{\fO}[1]{\ensuremath{\mathcal{O}\bra{#1}}}

\newcommand{\fOmega}[1]{\ensuremath{\Omega\hspace{0.1mm}\bra{#1}}}

\newcommand{\Graph}{\ensuremath{G}}

\newcommand{\Vertices}{\ensuremath{V}}
\newcommand{\VerticesOf}[1]{\ensuremath{\Vertices\hspace{-0.4mm}\bra{#1}}}

\newcommand{\VerticesTwo}{\ensuremath{U}}

\newcommand{\Edges}{\ensuremath{E}}

\newcommand{\vertex}{\ensuremath{v}}
\newcommand{\vertexOne}{\ensuremath{v}}

\newcommand{\vertexTwo}{\ensuremath{u}}

\newcommand{\coloring}{\ensuremath{c}}
\newcommand{\coloringPrime}{\ensuremath{c\hspace{0.4mm}'}}
\newcommand{\coloringDouble}{\ensuremath{c\hspace{0.4mm}''}}
\newcommand{\coloringArc}{\ensuremath{\breve{c}}}

\newcommand{\fColoring}[1]{\ensuremath{c}\hspace{0.4mm}\bra{#1}}
\newcommand{\fColoringPrime}[1]{\ensuremath{c\hspace{0.4mm}'}\bra{#1}}
\newcommand{\fColoringDouble}[1]{\ensuremath{c\hspace{0.4mm}''}\bra{#1}}
\newcommand{\fColoringArc}[1]{\ensuremath{\breve{c}}\bra{#1}}

\newcommand{\fChromatic}[1]{\ensuremath{\chi\bra{#1}}}
\newcommand{\fClique}[1]{\ensuremath{\omega\bra{#1}}}

\newcommand{\PartB}{\ensuremath{\mathfrak{B}}}
\newcommand{\Block}{\ensuremath{\mathcal{B}}}

\newcommand{\sizesp}{\ensuremath{m}}
\newcommand{\Intervals}{\ensuremath{\mathcal{I}}}

\newcommand{\IntervalsTwo}{\ensuremath{\mathcal{J}}}
\newcommand{\IntervalsThree}{\ensuremath{\mathcal{K}}}
\newcommand{\IntervalsFour}{\ensuremath{\mathcal{L}}}
\newcommand{\IntervalsFive}{\ensuremath{\mathcal{M}}}

\newcommand{\spann}[1]{\ensuremath{\mathsf{span}(#1)}}

\newcommand{\intervals}[2]{\ensuremath{\mathsf{intervals}_{#1}(#2)}}

\newcommand{\distInt}[2]{\ensuremath{\mathsf{dist}(#1,#2)}}

\newcommand{\Pruir}{\textsc{Incremental Unit Interval Recoloring}\;}
\newcommand{\PruirFD}{\textsc{Fully Dynamic Unit Interval Recoloring}\;}
\newcommand{\Recolor}{\textsc{Recolor}\;}
\newcommand{\FRight}{\textsc{FindRight}\;}
\newcommand{\FLeft}{\textsc{FindLeft}\;}
\newcommand{\ColCompPro}{\textsc{Unit Color Completion}\;}
\newcommand{\ColComp}{\textsc{Greedy Color Completion}\;}
\newcommand{\ModColComp}{\textsc{Modulo Color Completion}\;}
\newcommand{\Frog}{\textsc{Frogs}\;}
\newcommand{\SetUnion}{\textsc{Set Union}\;}

\newcommand{\divi}{\textsf{div}}
\newcommand{\opx}[1]{\textsf{x}(#1)}
\newcommand{\opy}[1]{\textsf{y}(#1)}
\newcommand{\prefix}[2]{\textsf{prefix}^{\linearOrder}_{#1}(#2)}
\newcommand{\suffix}[2]{\textsf{suffix}^{\linearOrder}_{#1}(#2)}
\newcommand{\infix}[3]{\textsf{infix}^{\linearOrder}_{#1}(#2,#3)}

\newcommand{\point}{\ensuremath{x}}
\newcommand{\pointOne}{\ensuremath{x}}
\newcommand{\pointTwo}{\ensuremath{y}}
\newcommand{\pointThree}{\ensuremath{z}}

\newcommand{\intDef}[2]{\ensuremath{[{#1},{#2})}}

\newcommand{\interval}{\ensuremath{I}}
\newcommand{\intervalPrim}{\ensuremath{I\hspace{0.1mm}'}}
\newcommand{\intervalOne}{\ensuremath{I}}

\newcommand{\intervalTwo}{\ensuremath{J}}

\newcommand{\intervalThree}{\ensuremath{K}}
\newcommand{\intervalFour}{\ensuremath{L}}
\newcommand{\intervalFive}{\ensuremath{M}}

\newcommand{\Spanset}{\ensuremath{\mathcal{S}}}
\newcommand{\Spanel}{\ensuremath{S}}

\newcommand{\linearOrder}{\ensuremath{\sqsubset}}

\DeclareMathOperator{\load}{load}
\newcommand{\fLoad}[2]{\ensuremath{\load\bra{#1,#2}}}

\DeclareMathOperator{\maxLoad}{max-load}
\newcommand{\fMaxLoad}[1]{\ensuremath{\maxLoad\bra{#1}}}

\newcommand{\circumference}{\ensuremath{\lambda}}

\newcommand{\eqArc}{\ensuremath{\equiv_{\circumference}}}
\newcommand{\eqClass}[1]{\ensuremath{[{#1}]_{\circumference}}}

\newcommand{\CirclePoints}{\ensuremath{\RealNumbers_{/\circumference\IntegerNumbers}}}

\newcommand{\circlePointTwo}{\ensuremath{\beta}}
\newcommand{\circlePointThree}{\ensuremath{\gamma}}
\newcommand{\circlePointFour}{\ensuremath{\delta}}

\newcommand{\Arcs}{\ensuremath{\mathcal{A}}}

\newcommand{\ArcsTwo}{\ensuremath{\mathcal{B}}}

\newcommand{\arcDef}[2]{\ensuremath{[{#1},{#2})}}
\newcommand{\arc}{\ensuremath{A}}
\newcommand{\arcOne}{\ensuremath{A}}

\newcommand{\arcTwo}{\ensuremath{B}}

\newcommand{\arcRepr}{\ensuremath{A}}
\newcommand{\fArcRepr}[1]{\ensuremath{\arcRepr\bra{#1}}}

\newcommand{\HRred}[2]{\ensuremath{{\color{red!60!black}\HR{#1}{#2}}}}
\newcommand{\HRblue}[2]{\ensuremath{{\color{blue!60!black}\HR{#1}{#2}}}}
\newcommand{\HRgreen}[2]{\ensuremath{{\color{green!50!black}\HR{#1}{#2}}}}

\newcommand{\jn}{\ensuremath{\kappa}}
\newcommand{\ti}{\ensuremath{\tau}}
\newcommand{\cost}{\ensuremath{\varsigma}}

\newcommand{\rank}[1]{\ensuremath{r_{#1}}}
\newcommand{\rankTwo}[1]{\ensuremath{p_{#1}}}
\newcommand{\rankThree}[1]{\ensuremath{q_{#1}}}

\newcommand{\rrank}[2]{\ensuremath{r_{#1,#2}}}
\newcommand{\rrankTwo}[2]{\ensuremath{p_{#1,#2}}}
\newcommand{\rrankThree}[2]{\ensuremath{q_{#1,#2}}}

\newcommand{\Ranks}{\ensuremath{R}}
\newcommand{\RanksTwo}{\ensuremath{P}}
\newcommand{\RanksThree}{\ensuremath{Q}}

\newcommand{\RRanks}[1]{\ensuremath{R_{#1}}}

\newcommand{\Rini}{\ensuremath{\delta}}
\newcommand{\RRanksTwo}[1]{\ensuremath{P_{#1}}}
\newcommand{\RRanksThree}[1]{\ensuremath{Q_{#1}}}

\newcommand{\HR}[2]{\ensuremath{H_{\,#1}^{\,#2}}}
\newcommand{\HRP}[2]{\ensuremath{\overline{H}_{\,#1}^{\,\,#2}}}

\newtheorem{theorem}{Theorem}
\numberwithin{theorem}{section}
\newtheorem{observation}[theorem]{Observation}
\newtheorem{lemma}[theorem]{Lemma}
\newtheorem{conjecture}[theorem]{Conjecture}
\newtheorem{corollary}[theorem]{Corollary}

\newtheorem{claim}{Claim}

\theoremstyle{definition}

\numberwithin{definition}{section}

\title[Recoloring Unit Interval Graphs with Logarithmic Recourse Budget]{Recoloring Unit Interval Graphs with Logarithmic Recourse Budget$^{\ddagger}$}

\thanks{$^{\ddagger}$Supported by Narodowe Centrum Nauki -- Grant 2017/26/D/ST6/00264.}

\author{Bart{\l}omiej Bosek}

\address{Theoretical Computer Science Department,\\Faculty of Mathematics and Computer Science,\\Jagiellonian University in Krak\'{o}w}

\email{bosek@tcs.uj.edu.pl}

\author{Anna Zych-Pawlewicz}

\address{University of Warsaw}

\email{anka@mimuw.edu.pl}

\date{\today}

\makeatletter
\begin{document}

\thispagestyle{empty}

\maketitle

\begin{abstract}
In this paper we study the problem of coloring a unit interval graph which changes dynamically. In our model the unit intervals are added or removed one at the time, and have to be colored immediately, so that no two overlapping intervals share the same color. After each update only a limited number of intervals is allowed to be recolored. The limit on the number of recolorings per update is called the recourse budget. In this paper we show, that if the graph remains $k$-colorable at all times, and the updates consist of insertions only, then we can achieve the amortized recourse budget of $\fO{k^7 \log n}$ while maintaining a proper coloring with $k$ colors. This is an exponential improvement over the result in~\cite{bosek_et_al:SWAT} in terms of both $k$ and $n$. We complement this result by showing the lower bound of $\Omega(n)$ on the amortized recourse budget in the fully dynamic setting. Our incremental algorithm can be efficiently implemented.  

As a byproduct of independent interest we include a new result on coloring proper circular arc graphs. Let $L$ be the maximum number of arcs intersecting in one point for some set of unit circular arcs $\Arcs$. We show that if there is a set $\Arcs'$ of non-intersecting unit arcs of size $L^2-1$ such that $\Arcs \cup \Arcs'$ does not contain $L+1$ arcs intersecting in one point, then it is possible to color $\Arcs$ with $L$ colors. This complements the work on unit circular arc coloring \cite{BCh09,Tuc75,Val03}, which specifies sufficient conditions needed to color $\Arcs$ with $L+1$ colors or more.
\end{abstract}
 
\section{Introduction}\label{sec:intro}

One of the most prominent problems within graph theory is coloring. 
The graph coloring problem has
plenty of variants, applications, and deep connections to theoretical computer 
science. For a positive integer $k$, a \emph{proper $k$-coloring} of a graph is 
an assignment of colors in $\{ 1, \ldots, k \}$ to the vertices of the graph in 
such a way that no two adjacent vertices share a color. The chromatic number of 
the graph is the smallest integer $k$ for which a proper $k$-coloring exists. In 
general, it is 
NP-hard~\cite{KhotP06,Zuckerman06} 
to approximate the chromatic number of an $n$-vertex graph within a factor of 
$n^{1-\epsilon}$ for any constant $\epsilon > 0$. The literature offers many 
results for restricted graph classes. 

In this paper, we study the class of unit interval graphs, for which a linear time 
greedy algorithm achieves the optimum coloring~\cite{Olariu91}. Our main interest 
lies within a dynamic setting, where vertices get inserted or removed one at a time, together with adjacent edges to present vertices, and one needs to 
maintain a close to optimal coloring after each vertex update. In this setting we study how many 
vertex recolorings are needed to maintain a good quality coloring.  The number of 
changes one needs to introduce to the maintained solution (in our case vertex 
recolorings) upon an update is referred to as 
\emph{recourse budget} in the literature. A recourse budget of zero and updates restricted to vertex insertions coincide 
with the online setting, where the algorithm's decisions are irrevocable and the graph (typically) only grows. 

Since 1967, when the influential book of Rodgers~\cite{Rodgers} gave rise to online algorithms,
the online model has been applied to a vast variety of optimization problems~\cite{BosekFKKMM12,KiersteadT81,Epstein,DBLP:conf/approx/Chybowska-Sokol20,10.5555/645569.659725,LOVASZ1989319,Halldorsson94lowerbounds,10.1007/978-3-319-13075-0_40,KVV,10.1145/100216.100268,DBLP:conf/soda/Rohatgi20,DBLP:conf/icalp/JiangP020}. Many pessimistic lower bounds for this model have been revealed. In particular, for general $k$-colorable graphs, $\frac{n}{\log^2 n} k$ colors are necessary to maintain the coloring online~\cite{Halldorsson94lowerbounds}. For the case of $k$-colorable interval graphs, $3k-2$ colors are necessary and sufficient in the online model~\cite{KiersteadT81}. The case of unit interval graphs has also been extensively studied, revealing that $2k-1$ colors suffice and $3k/2$ colors are necessary to color a $k$-colorable unit interval graph online~\cite{BosekFKKMM12,Epstein}.
It is natural to ask if the situation improves if one allows a limited recourse budget.  

This question has been recenly receiving increasing attention, and the positive answer has been obtained for a large number of problems~\cite{BernsteinHR18,BosekLSZ14,BosekLSZ18a,
BosekLSZ18b,ImaseW91,MegowSVW12,LackiOPSZ15,BarbaCKLRRV17,SolomonW18,bosek_et_al:SWAT}. In particular Barba et al. applies the recourse budget model to coloring general graphs \cite{BarbaCKLRRV17}.
The authors devise two complementary algorithms (in the regime of adding and removing vertices). For 
any $d>0$, the first (resp.~second) algorithm maintains a 
$k(d+1)$-coloring (resp.~$k(d+1)n^{1/d}$-coloring) of a $k$-colorable graph and recolors at most 
$(d+1)n^{1/d}$ 
(resp.~$d$) vertices per update. The authors also show that the first trade-off is 
tight. The situation improves even further if we restrict the class of graphs to interval graphs and updates to vertex insertions. 
In~\cite{bosek_et_al:SWAT} the authors show that with the recourse budget of $\fO{\log n}$ one can color a $k$-colorable graph with $2k$ colors. The recurse budget of $\fO{k! \sqrt{n}}$ allows maintaining an optimal $k$-coloring. For unit interval graphs recoloring budget of $\fO{k^2}$ is sufficient for maintaining a $(k+1)$-coloring of a $k$-colorable graph, even if both insertions and removals are allowed. It is left open what budget is needed for maintaining an optimal $k$-coloring for unit interval graphs, even if we only allow vertex insertions. The lower bound given in~\cite{bosek_et_al:SWAT} is $\Omega(\log n)$ (for insertions only), while the upper bound is $\fO{k! \sqrt{n}}$ (the same as for general interval graphs). Such a tremendous gap for such elementary graph class calls for further investigation. The main result of this paper is that we close this gap up to the factors polynomial in $k$. To be more precise, we show that an amortized recourse budget of $\fO{k^7 \log n}$ is sufficient to maintain $k$-coloring under vertex insertions. This is an exponential improvement over~\cite{bosek_et_al:SWAT} in terms of both $n$ and $k$. It is fairly easy to see that our algorithm can be efficiently implemented. Our upper bound holds under the assumption that $k$ is known apriori to the algorithm. Without this assumption our bound grows by a factor of $k$. Our result translates directly to online chain partitioning of semi-orders of width $k$. We complement this result by showing that in the fully dynamic setting one must spent an amortized recorse budget of $\Omega(n)$ per update.   

As a byproduct of our approach we offer a new combinatorial result for coloring proper circular arc graphs. The important parameters here are the load $L$ and the cover number $l$ of the input instance. The load $L$ stands for the maximum number of arcs intersecting in one point, whereas the cover number $l$ stands for the minimum number of arcs covering the whole circle. There has been an ongoing line of research to expose the conditions under which a circular arc graph admits a proper coloring using close to $L$ colors. Tucker~\cite{Tuc75} shows, that if $l \geq 4$, then $\floor{3L/2}$ colors suffice. Valencia-Pabon~\cite{Val03} shows that if $l \geq 5$, then $\ceil{\frac{l-1}{l-2}L}$ colors is enough. For $l \geq L+2$ the bound becomes $L+1$. Belkale and Chandran~\cite{BCh09} prove the Hadwiger's conjecture for proper circular arc graphs. Neither of these results exposes a condition sufficient to color the instance with precisely $L$ colors. We show, that if one can add $L^2-1$ nonintersecting arcs to the instance in a way that the load does not increase, then $L$ colors is sufficient to properly color the instance.

Our final results are based on three problems that we solve in the paper. In Section~\ref{sec:colcomp} we address the \ColCompPro problem, which implies the results for coloring proper circular arcs. In Section~\ref{sec:frog} we study the \Frog problem, and finally in Section~\ref{sec:mainres} we combine the results of Section~\ref{sec:colcomp} and Section~\ref{sec:frog} to solve the \Pruir problem, which is the main result of the paper. We defer the fomal statements of these problems to the end of Section~\ref{sec:prelims}, in which we introduce the necessary definitions to state these problems. Also in Section~\ref{sec:prelims} we give a lower bound for the fully dynamic setting.

\section{Preliminaries}\label{sec:prelims}
In this section we formally state the problems that we address in the further sections. To be able to do so, we first introduce the related notation and some folklore facts concerning interval and unit interval graphs.

We consider in this paper half-open intervals $\interval=\intDef{x}{y}$ for some $x,y \in \RealNumbers, x<y$. For an interval $\interval$ we define operators $\opx{\interval}\eqdef x$ and $\opy{\interval}\eqdef y$ for accesing the begin and the end coordinate. A multiset of intervals can be interpreted as a graph: the intervals are interpreted as vertices, which are adjacent if and only if the corresponding intervals intersect. Graphs obtained in this way are called interval graphs. The coloring related definitions stated in the first paragraph of Section~\ref{sec:intro} directly translate to multisets of intervals interpreted as graphs. For a multiset of intervals $\Intervals$, the function $c : \Intervals \mapsto \set{ 1, \ldots, k}$ is a proper $k$-coloring if for any $\interval,\intervalTwo \in \Intervals$ it holds that $c(\interval) \neq c(\intervalTwo)$ iff $\interval$ and $\intervalTwo$ intersect. The chromatic number $\fChromatic{\Intervals}$ is the minimum number $k$ for which $\Intervals$ admits a proper $k$-coloring.  Similarly we can adapt the definition of a clique: 
a multiset of intervals $\IntervalsTwo=\set{\intervalTwo_1,\intervalTwo_2,\ldots,\intervalTwo_m}$ is a clique if and only if $\bigcap\IntervalsTwo \eqdef \intervalTwo_1 \cap \ldots \cap \intervalTwo_m \neq \emptyset,$ or equivalently $\max_{i=1}^{m}\opx{\intervalTwo_i}<\min_{i=1}^{m}\opy{\intervalTwo_i}$.  
In such case it holds that $\bigcap\IntervalsTwo= \intDef{\max_{i=1}^{m}\opx{\intervalTwo_i}}{\min_{i=1}^{m}\opy{\intervalTwo_i}}$. We refer to this intersection as span of $\IntervalsTwo$, and denote it as $\spann{\IntervalsTwo} \eqdef \bigcap\IntervalsTwo$. To emphasize the size of $\IntervalsTwo$ we often refer to it as an $m$-clique. The clique number $\fClique{\Intervals}$ of a multiset of intervals $\Intervals$ is the maximum number $m$ such that $\Intervals$ contains an $m$-clique. It is well known that interval graphs are perfect, that is:
\begin{claim}[~\cite{GOLUMBIC2004171}]
Let $\Intervals$ be a multiset of intervals. Then
$
\fClique{\Intervals} = \fChromatic{\Intervals}.
$
\end{claim}

We introduce two orders $\linearOrder$ and $<$ 
on a multiset of intervals $\Intervals$. 
For any $\intervalOne,\intervalTwo \in \Intervals$ we let
$
\intervalOne \linearOrder \intervalTwo \quad \text{if} \quad \opx{ \intervalOne} < \opx{ \intervalTwo}
$. If $\interval=\intervalTwo$, we solve the tie in an arbitrary consistent manner. Hence, $\linearOrder$ is a linear order on $\Intervals$. We say that $\intervalOne < \intervalTwo$ if and only if $\opy{ \intervalOne} \leq \opx{ \intervalTwo}$. Hence, $<$ is a linear order only on independent sets of intervals (i.e., multisets of intervals that are pairwise non-intersecting). Orders $\linearOrder$ and $<$ extend to multisets of intervals in a natural manner. For two multisets of intervals $\IntervalsTwo$ and $\IntervalsThree$, we say that $\IntervalsTwo \linearOrder \IntervalsThree$ (respectively $\IntervalsTwo < \IntervalsThree$) if for all $\intervalTwo \in \IntervalsTwo, \intervalThree \in \IntervalsThree$ it holds that $\intervalTwo \linearOrder \intervalThree$ (respectively $\intervalTwo < \intervalThree$). For a multiset of intervals $\Intervals$ by $\Intervals=\set{\interval_1 \linearOrder \ldots\linearOrder \interval_m}$ we denote that $\Intervals=\set{\interval_1 , \ldots , \interval_m}$ and $\interval_1 \linearOrder \ldots \linearOrder \interval_m$. For an ordered multiset of intervals $\Intervals=\set{\interval_1 \linearOrder \ldots\linearOrder \interval_m}$ we now define a prefix, a suffix and an infix of $\Intervals$. For $\interval_i,\interval_l \in \Intervals$, $l>i$ we set $\prefix{\Intervals}{\interval_i}=\set{\interval_1,\ldots,\interval_i}$,
$\suffix{\Intervals}{\interval_i}=\set{\interval_i,\ldots,\interval_m}$
and $\infix{\Intervals}{\interval_i}{\interval_l}=\set{\interval_i, \ldots, \interval_l}$.

We now move on to some basic observations that hold for multisets of unit intervals.
An interval $\interval$ is unit if and only if $\opy{\interval}=\opx{\interval}+1$. 
\begin{observation}
\label{obs:extremalIntervals}
Let $\Intervals=\set{\interval_1 \linearOrder \ldots \linearOrder \interval_m}$ be a multiset of unit intervals.
If  $\fClique{\Intervals}<m$, then the extremal intervals are disjoint, i.e., $\interval_1 < \interval_m$.
\end{observation}

\begin{proof}
Suppose contrary that $\interval_1 \cap \interval_n \neq \emptyset$, i.e., $\opx{\interval_n} < \opy{\interval_1}=\opx{\interval_1}+1$.
Then $\max_{i=1}^{n}\opx{\interval_i}<\min_{i=1}^{n}\opy{\interval_i}$
and as a consequence
$
\bigcap \Intervals \neq \emptyset
$
contradicting $\fClique{\Intervals} < n$.
\end{proof}
Our next observation exhibits the fact that coloring modulo $k$ works well on unit intervals ordered by $\linearOrder$. \begin{observation}\label{obs:gcc}
Let $\Intervals=\set{\interval_1 \linearOrder \ldots \linearOrder \interval_m}$ be a multiset of unit intervals such that $\fClique{\Intervals} \leq k$. Let $k \leq l < m$ and let $\coloring:\prefix{\Intervals}{\interval_l} \mapsto [k]$ be a proper $k$-coloring for $\prefix{\Intervals}{\interval_l}$ such that $\coloring$ is a bijection on $\infix{\Intervals}{\interval_{l-k+1}}{\interval_{l}}$. Let $\coloring':\Intervals \mapsto [k]$ be defined as follows: $\coloring'(\interval_i)\eqdef \coloring(\interval_i)$ for $i \in \num{l}$ and $\coloring'(\interval_{l+i}) \eqdef \coloring(\interval_{l-k+(i \mod k)})$ for $i \in \num{n-l}$. Then $\coloring'$ is a proper $k$-coloring on $\Intervals$.
\end{observation}
\begin{proof}
It is clear that $\coloring$ assigns only colors in $[k]$, it remains to prove that it is also a proper coloring.
Let $\interval_i,\interval_j \in \Intervals$, where $i<j$ and $j>l$. If $j<i+k$, then by definition $\fColoring{\interval_i} \neq \fColoring{\interval_j}$.
Otherwise, $\size{ \set{ \interval_i, \interval_{i+1}, \ldots, \interval_j}} \geq k+1$, so by Observation~\ref{obs:extremalIntervals}, $\interval_i \cap \interval_j=\emptyset$. 
\end{proof}
Note, that Observation~\ref{obs:gcc} implies a greedy coloring algorithm for any multiset $\Intervals$ of unit intervals: color the first $k$ intervals into colors from $1$ to $k$ and then use Observation~\ref{obs:gcc} to complete the coloring. Unfortunatelly Observartion~\ref{obs:gcc} is not useful if the coloring given on the prefix of $\Intervals$ is not a bijection on the last $k$ intervals of the prefix. This can be easly fixed using the observation and corollary that follow next. 
\begin{observation}\label{obs:makeBijection}
Let $\Intervals=\set{ \interval_1 \linearOrder \ldots \linearOrder \interval_{2k} }$ be a multiset of unit intervals such that  $\fClique{\Intervals} \leq k$ and let $\coloringPrime : \prefix{\Intervals}{\interval_k} \mapsto \num{k}$ be a proper coloring on $\prefix{\Intervals}{\interval_k}$. Then there is a proper coloring $\coloring : \Intervals \mapsto \num{k}$ such that $\coloring(\interval_i)=\coloringPrime(\interval_i)$ for $i \in \num{k}$ and $\coloring$ is a bijection on $\suffix{\Intervals}{\interval_{k+1}}$.
\end{observation}
\begin{proof}
To construct $\coloring$, we first obviously set $\coloring(\interval_i)=\coloringPrime(\interval_i)$ for all $i \in \num{k}$.
Now we need to color the intervals of $\suffix{\Intervals}{\interval_{k+1}}$ in a way that the statement of Observation~\ref{obs:makeBijection} holds. Let $\IntervalsTwo=\set{ \interval_l, \interval_{l+1}, \ldots \interval_k}$ be the intervals of $\prefix{\Intervals}{\interval_k}$ that intersect $\interval_k$. Clearly, $\IntervalsTwo$ is a $(k-l+1)$-clique.
We first assign the $l-1$ colors not used by $\IntervalsTwo$ to the first $l-1$ intervals of $\suffix{\Intervals}{\interval_{k+1}}$. That is we assign colors $\num{k} \setminus \coloring(\IntervalsTwo)$ to intervals $\infix{\Intervals}{\interval_{k+1}}{\interval_{k+l-1}}$ in an arbitrary order. We get a proper coloring since the intervals of $\infix{\Intervals}{\interval_{k+1}}{\interval_{k+l-1}}$ do no intersect intervals of $\prefix{\Intervals}{\interval_k}$ other than $\IntervalsTwo$. Finally, we set $\coloring(\interval_{i})=\coloring(\interval_{i-k})$ for $i \in \set{k+l, \ldots, 2k}$. By Observarion~\ref{obs:extremalIntervals} intervals $\interval_i$ and $\interval_{i-k}$ do not intersect. This completes the proof.
\end{proof}

\begin{observation}\label{obs:gccNC}
Let $\Intervals=\set{\interval_1 \linearOrder \ldots \linearOrder \interval_m}$ be a multiset of unit intervals such that $\fClique{\Intervals} \leq k$. Let $l < m$ and let $\coloring:\prefix{\Intervals}{ \interval_l } \mapsto [k]$ be a proper $k$-coloring for the prefix $\prefix{\Intervals}{ \interval_l }$. Then there is a proper $k$-coloring $\coloring'$ for $\Intervals$ such that $\coloring'(\interval_i)=\coloring(\interval_i)$ for $i \in \num{l}$. 
\end{observation}
\begin{proof}Direct consequence of Observation~\ref{obs:gcc} and Observation~\ref{obs:extremalIntervals}.\end{proof}
The algorithm of Observation~\ref{obs:gccNC} is further referred to as the \ColComp algorithm while the algorithm of Observation~\ref{obs:gcc} is referred to as the \ModColComp algorithm. All the above basic observations will be used as bulding blocks in further sections. 

We conclude this section with stating all the problems considered in the paper.
\paragraph*{\ColCompPro problem} We are given a $k$-colorable multiset of unit intervals $\Intervals=\set{\interval_1 \linearOrder \ldots \linearOrder \interval_m}$, $m > 2k$. We are also given a proper $k$-coloring $\coloring'$ on $ \prefix{\Intervals}{\interval_k}$ and a proper coloring $\coloring''$ on $\suffix{\Intervals}{\interval_{m-k+1}}$. When is it possible to extend $c'$ and $c''$ to a proper $k$-coloring of $\Intervals$ ? We study this problem in Section~\ref{sec:colcomp}.
\paragraph*{\Frog problem}
The Frogs game is defined by three non-negative integer parameters: $n$, where $n \geq 1$, $\jn$, where $1 \leq \jn \leq n$, and $\Rini$. In addition to that the game is defined by two integer sequences:
\begin{itemize}
\item the intial ranks $\Ranks = \bra{\rank{1},\ldots,\rank{n}}$, where each $\rank{i}=\Rini$ and
\item the jumping sequence $J = \bra{j_1,j_2,\ldots,j_{n-1}}$, where $1 \leqslant j_{\ti} \leqslant n-\ti$.
\end{itemize}
We now define the sequence $\RRanks{1},\ldots,\RRanks{n}$ inductively: $\RRanks{1}\eqdef \Ranks$ and for $\ti\in[n-1]$, based on $\RRanks{\ti}=\bra{\rrank{\ti}{1},\ldots,\rrank{\ti}{n-\ti+1}}$, we define
$$
\RRanks{\ti+1} \eqdef  \bra{\rrank{\ti}{1},\ \ldots,\ \rrank{\ti}{j_{\ti}-1},\ \rrank{\ti}{j_{\ti}}+\rrank{\ti}{j_{\ti}+1},\ \rrank{\ti}{j_{\ti}+2},\ \ldots,\ \rrank{\ti}{n-\ti+1}}.
$$
In other words, new ranks $\RRanks{\ti+1}$ are created by adding two neighboring numbers in $\RRanks{\ti}$ placed at positions $j_{\ti}$ and $j_{\ti}+1$.
The cost of such operation is defined as
$$
\cost_{\ti}\eqdef \min\bra{\rrank{\ti}{j_{\ti-\jn+1}}+\ldots+\rrank{\ti}{j_{\ti}},\ \rrank{\ti}{j_{\ti+1}}+\ldots+\rrank{\ti}{j_{\ti+\jn}}}
$$
where all $\rrank{\ti}{i}$ not occuring in $\RRanks{\ti}$ are defined as $\rrank{\ti}{i}\eqdef \Rini$.
The total cost of the Frog is $\cost \eqdef \cost_1+\ldots+\cost_{n-1}$. How large can the total cost be for the most pessimistic jumping sequence $J$ ? We address this problem in Section~\ref{sec:frog}.
\paragraph*{\Pruir problem}
We are given a parameter $k$ and a sequence of $n$ unit intervals: $\interval_1,\interval_2, \ldots, \interval_n$ such that $\set{\interval_1 , \ldots, \interval_n}$ is $k$-colorable. The sequence is not necessarily ordered by $\linearOrder$ or $<$, as it represents the arrival order of the intervals. This sequence defines $n+1$ multisets of unit intervals: $\Intervals_0=\emptyset$ and $\Intervals_j \eqdef \set{\interval_1,\ldots,\interval_j}$ for $j>0$. Instance $\Intervals_j$ differs from $\Intervals_{j-1}$ by one interval $\interval_j$. The goal is to maintain a proper $k$-coloring on the dynamic instance. To be more precise, after the interval $\interval_j$ is presented, the algorithm needs to compute a 
proper $k$-coloring $\coloring_j$ for $\Intervals_j$. Our objective is to 
minimize the recourse budget, which is the number of intervals with 
different colors in $\coloring_j$ and $\coloring_{j-1}$. We address this problem in Section~\ref{sec:mainres}.
\paragraph*{\PruirFD problem}
Initially, we are given an empty multiset of unit intervals $\Intervals_0=\emptyset$. The instance $\Intervals_{j+1}$ is obtained from $\Intervals_{j}$ by adding a new unit interval to $\Intervals_j$ or removing the existing chosen interval from $\Intervals_j$. Every instance $\Intervals_j$ presented to the algorithm is $k$-colorable. The goal is again to maintain a proper $k$-coloring on the dynamic instance. After each interval insertion and removal, the algorithm needs to compute a 
proper $k$-coloring $\coloring_j$ for $\Intervals_j$. Similarily as before, our objective is to 
minimize the recourse budget, which is the number of intervals with 
different colors in $\coloring_j$ and $\coloring_{j-1}$. This problem is addressed in this section, in Observation~\ref{obs:FD}. The proof of Observation~\ref{obs:FD} is fairly simple and due to space limits it is moved to Appendix~\ref{app:lb}.

\begin{observation}\label{obs:FD}
There is a sequence of $m$ updates for the \PruirFD problem that forces the total number of recolorings of $\Omega(m^2)$. 
\end{observation}
 
\section{Unit Color Completion problem}\label{sec:colcomp}
In this section we consider the \ColCompPro problem defined at the end of Section~\ref{sec:prelims}. We are given a $k$-colorable multiset of unit intervals $\Intervals=\set{\interval_1 \linearOrder \ldots \linearOrder \interval_n}$. We are given the coloring on the first $k$ intervals and the last $k$ intervals. We show here the sufficient condition to extend this coloring to the proper $k$-coloring of the entire instance. To be more precise, we show that it is enough to have enough slack between the first $k$ and the last $k$ intervals, meaning that we can fit there a large enough independent set of unit intervals without increasing the clique number. This result is stated more formally in the following lemma.

\begin{lemma}
\label{lem:partColor}
Let $k>0$ be a positive integer number and let $\Intervals=\set{\interval_1 \linearOrder \ldots \linearOrder \interval_n}$ be a multiset of $n\geqslant 2k$ unit intervals such that $\fClique{\Intervals} \leqslant k$. Let $\coloringPrime: \prefix{\Intervals}{\interval_k} \mapsto \num{k}$ and $\coloringDouble: \suffix{\Intervals}{\interval_{n-k+1}} \mapsto\num{k}$ be bijections. In addition let the following condition (1.) hold:
\begin{enumerate}
\item There is a set of $k^2-1$ pairwise disjoint unit intervals $\IntervalsFour = \set{\intervalFour_1,\ldots,\intervalFour_{k^2-1}}$, which lie between the interval $\interval_k$ and the interval $\interval_{n-k+1}$, whose addition does not increase the clique number of $\Intervals$ above $k$, i.e.,
\begin{enumerate}
\item $\interval_1 \linearOrder \interval_2 \linearOrder \ldots \linearOrder \interval_k \ \linearOrder \ \intervalFour_1 < \intervalFour_2 < \ldots < \intervalFour_{k^2-1} \ \linearOrder \ \interval_{n-k+1} \linearOrder \interval_{n-k+2} \linearOrder \ldots \linearOrder \interval_n$,
\item $\fClique{\Intervals \cup \IntervalsFour} \leqslant k.$
\end{enumerate}

\end{enumerate}
Then there is a proper coloring $\coloring:\Intervals\rightarrow [k]$ which is an extension of coloring $\coloringPrime$ and $\coloringDouble$ to the whole~$\Intervals$, i.e., $\fColoring{\interval_i}=\fColoringPrime{\interval_i}$ for $i=1,\ldots,k$ and $\fColoring{\interval_j}=\fColoringDouble{\interval_j}$ for $j=n-k+1,\ldots,n$.
\end{lemma}

\begin{proof}An example instance for which our assumtions hold is pictured in Figure~\ref{fig:1}.
First of all, we assume that the multiset of unit intervals $\Intervals$ is connected as a graph.
Otherwise Lemma \ref{lem:partColor} follows as one can use the algorithm of Observation~\ref{obs:gcc}. 

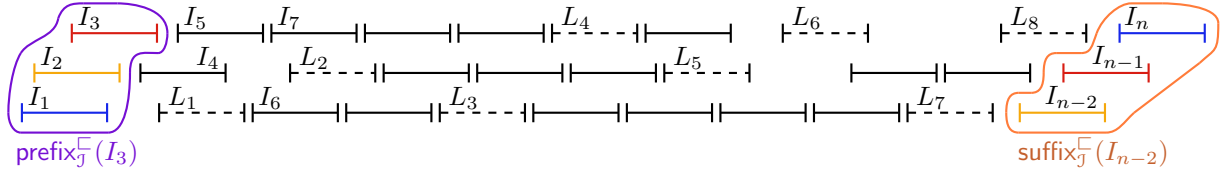
\begin{figure}[hbpt!]
\begin{small}

\centering

\begin{tikzpicture}[xscale={2\textwidth/97}, yscale=15]
	\begin{pgfonlayer}{nodelayer}
		\node [style=none] (0) at (-1, 0) {};
		\node [style=none] (1) at (0, 2) {};
		\node [style=none] (2) at (3, 4) {};
		\node [style=none] (3) at (6, 0) {};
		\node [style=none] (4) at (7, 2) {};
		\node [style=none] (5) at (8.5, 2) {};
		\node [style=none] (6) at (10, 0) {};
		\node [style=none] (7) at (10, 4) {};
		\node [style=none] (8) at (11.5, 4) {};
		\node [style=none] (9) at (15.5, 2) {};
		\node [style=none] (10) at (17, 0) {};
		\node [style=none] (11) at (17.5, 0) {};
		\node [style=none] (12) at (18.5, 4) {};
		\node [style=none] (13) at (19, 4) {};
		\node [style=none] (14) at (20.5, 2) {};
		\node [style=none] (15) at (24.5, 0) {};
		\node [style=none] (16) at (25, 0) {};
		\node [style=none] (17) at (26, 4) {};
		\node [style=none] (18) at (26.5, 4) {};
		\node [style=none] (19) at (27.5, 2) {};
		\node [style=none] (20) at (28, 2) {};
		\node [style=none] (21) at (32, 0) {};
		\node [style=none] (22) at (32.5, 0) {};
		\node [style=none] (23) at (33.5, 4) {};
		\node [style=none] (24) at (34, 4) {};
		\node [style=none] (25) at (35, 2) {};
		\node [style=none] (26) at (35.5, 2) {};
		\node [style=none] (27) at (39.5, 0) {};
		\node [style=none] (28) at (40, 0) {};
		\node [style=none] (29) at (41, 4) {};
		\node [style=none] (30) at (41.5, 4) {};
		\node [style=none] (31) at (42.5, 2) {};
		\node [style=none] (32) at (43, 2) {};
		\node [style=none] (33) at (47, 0) {};
		\node [style=none] (34) at (47.5, 0) {};
		\node [style=none] (35) at (48.5, 4) {};
		\node [style=none] (36) at (49, 4) {};
		\node [style=none] (37) at (50, 2) {};
		\node [style=none] (38) at (54.5, 0) {};
		\node [style=none] (39) at (55, 0) {};
		\node [style=none] (40) at (50.5, 2) {};
		\node [style=none] (41) at (56, 4) {};
		\node [style=none] (42) at (62, 0) {};
		\node [style=none] (43) at (62.5, 0) {};
		\node [style=none] (44) at (60, 4) {};
		\node [style=none] (45) at (57.5, 2) {};
		\node [style=none] (46) at (65.5, 2) {};
		\node [style=none] (47) at (69.5, 0) {};
		\node [style=none] (48) at (70, 0) {};
		\node [style=none] (49) at (67, 4) {};
		\node [style=none] (51) at (0.5, 0.75) {$\interval_1$};
		\node [style=none] (52) at (72.5, 2) {};
		\node [style=none] (53) at (73, 2) {};
		\node [style=none] (54) at (77.5, 4) {};
		\node [style=none] (55) at (77, 0) {};
		\node [style=none] (56) at (79, 0) {};
		\node [style=none] (57) at (80, 2) {};
		\node [style=none] (58) at (82.5, 2) {};
		\node [style=none] (59) at (84.5, 4) {};
		\node [style=none] (60) at (87, 4) {};
		\node [style=none] (61) at (86, 0) {};
		\node [style=none] (62) at (89.5, 2) {};
		\node [style=none] (63) at (94, 4) {};
		\node [style=none] (65) at (1.5, 2.75) {$\interval_2$};
		\node [style=none] (66) at (4.5, 4.75) {$\interval_3$};
		\node [style=none] (67) at (14, 2.75) {$\interval_4$};
		\node [style=none] (68) at (13, 4.75) {$\interval_5$};
		\node [style=none] (69) at (12, 0.75) {$\intervalFour_1$};
		\node [style=none] (70) at (19, 0.75) {$\interval_6$};
		\node [style=none] (71) at (20.5, 4.75) {$\interval_7$};
		\node [style=none] (72) at (22.5, 2.75) {$\intervalFour_2$};
		\node [style=none] (73) at (34.5, 0.75) {$\intervalFour_3$};
		\node [style=none] (74) at (43.5, 4.75) {$\intervalFour_4$};
		\node [style=none] (75) at (52.5, 2.75) {$\intervalFour_5$};
		\node [style=none] (76) at (62, 4.75) {$\intervalFour_6$};
		\node [style=none] (77) at (72, 0.75) {$\intervalFour_7$};
		\node [style=none] (78) at (79.5, 4.75) {$\intervalFour_8$};
		\node [style=none] (81) at (3.5, -2) {\textcolor{mynavy}{$\prefix{\Intervals}{\interval_{3}}$}};
		\node [style=none] (82) at (-2, 0) {};
		\node [style=none] (83) at (-1, 3) {};
		\node [style=none] (84) at (2, 4.75) {};
		\node [style=none] (85) at (-1, -1) {};
		\node [style=none] (86) at (6, -1) {};
		\node [style=none] (87) at (7.75, 2) {};
		\node [style=none] (88) at (10.75, 4) {};
		\node [style=none] (89) at (4, 5.5) {};
		\node [style=none] (90) at (10, 3) {};
		\node [style=none] (91) at (6.75, 5.5) {};
		\node [style=none] (92) at (83.25, 0.75) {$\interval_{n-2}$};
		\node [style=none] (93) at (87, 2.75) {$\interval_{n-1}$};
		\node [style=none] (94) at (88.5, 4.75) {$\interval_{n}$};
		\node [style=none] (95) at (78, 0) {};
		\node [style=none] (96) at (81, 2) {};
		\node [style=none] (97) at (85.25, 4) {};
		\node [style=none] (98) at (84, 3) {};
		\node [style=none] (99) at (82, 3) {};
		\node [style=none] (100) at (80, 1) {};
		\node [style=none] (101) at (79, 1) {};
		\node [style=none] (102) at (79, -1) {};
		\node [style=none] (103) at (86, -1) {};
		\node [style=none] (104) at (88.25, 0) {};
		\node [style=none] (105) at (95, 4.5) {};
		\node [style=none] (106) at (90.5, 5.5) {};
		\node [style=none] (107) at (88, 5.5) {};
		\node [style=none] (110) at (93.5, 2.5) {};
		\node [style=none] (111) at (85, -2) {\textcolor{myorange}{$\suffix{\Intervals}{\interval_{n-2}}$}};
	\end{pgfonlayer}
	\begin{pgfonlayer}{edgelayer}
		\draw [style=intBlue] (0.center) to (3.center);
		\draw [style=intYellow] (1.center) to (4.center);
		\draw [style=intRed] (2.center) to (7.center);
		\draw [style=interval] (5.center) to (9.center);
		\draw [style=intDashed] (6.center) to (10.center);
		\draw [style=interval] (8.center) to (12.center);
		\draw [style=interval] (11.center) to (15.center);
		\draw [style=interval] (13.center) to (17.center);
		\draw [style=intDashed] (14.center) to (19.center);
		\draw [style=interval] (16.center) to (21.center);
		\draw [style=interval] (18.center) to (23.center);
		\draw [style=interval] (20.center) to (25.center);
		\draw [style=intDashed] (22.center) to (27.center);
		\draw [style=interval] (24.center) to (29.center);
		\draw [style=interval] (26.center) to (31.center);
		\draw [style=interval] (28.center) to (33.center);
		\draw [style=intDashed] (30.center) to (35.center);
		\draw [style=interval] (32.center) to (37.center);
		\draw [style=interval] (34.center) to (38.center);
		\draw [style=interval] (36.center) to (41.center);
		\draw [style=interval] (39.center) to (42.center);
		\draw [style=intDashed] (40.center) to (45.center);
		\draw [style=interval] (43.center) to (47.center);
		\draw [style=intDashed] (44.center) to (49.center);
		\draw [style=interval] (46.center) to (52.center);
		\draw [style=intDashed] (48.center) to (55.center);
		\draw [style=interval] (53.center) to (57.center);
		\draw [style=intDashed] (54.center) to (59.center);
		\draw [style=intYellow] (56.center) to (61.center);
		\draw [style=intRed] (58.center) to (62.center);
		\draw [style=intBlue] (60.center) to (63.center);
		\draw [style=areaNavy, in=60, out=-180] (89.center) to (84.center);
		\draw [style=areaNavy, in=45, out=-105, looseness=0.75] (84.center) to (83.center);
		\draw [style=areaNavy, in=90, out=-135, looseness=0.75] (83.center) to (82.center);
		\draw [style=areaNavy, in=-180, out=-90] (82.center) to (85.center);
		\draw [style=areaNavy, in=-180, out=0, looseness=0.50] (85.center) to (86.center);
		\draw [style=areaNavy, in=-90, out=0] (86.center) to (87.center);
		\draw [style=areaNavy, in=180, out=90, looseness=1.25] (87.center) to (90.center);
		\draw [style=areaNavy, in=-90, out=0] (90.center) to (88.center);
		\draw [style=areaNavy, in=0, out=90, looseness=0.75] (88.center) to (91.center);
		\draw [style=areaNavy, in=0, out=180, looseness=0.75] (91.center) to (89.center);
		\draw [style=areaOrange] (102.center) to (103.center);
		\draw [style=areaOrange, in=210, out=0] (103.center) to (104.center);
		\draw [style=areaOrange, in=0, out=90, looseness=0.50] (105.center) to (106.center);
		\draw [style=areaOrange] (106.center) to (107.center);
		\draw [style=areaOrange, in=60, out=-180, looseness=0.75] (107.center) to (97.center);
		\draw [style=areaOrange, in=0, out=-120, looseness=1.25] (97.center) to (98.center);
		\draw [style=areaOrange] (98.center) to (99.center);
		\draw [style=areaOrange, in=90, out=-180, looseness=1.25] (99.center) to (96.center);
		\draw [style=areaOrange, in=0, out=-90, looseness=1.25] (96.center) to (100.center);
		\draw [style=areaOrange] (100.center) to (101.center);
		\draw [style=areaOrange, in=90, out=180, looseness=1.25] (101.center) to (95.center);
		\draw [style=areaOrange, in=180, out=-90, looseness=1.25] (95.center) to (102.center);
		\draw [style=areaOrange, in=30, out=-90, looseness=1.50] (105.center) to (110.center);
		\draw [style=areaOrange] (110.center) to (104.center);
	\end{pgfonlayer}
\end{tikzpicture}
 
\end{small}

\caption{An example illustrating the assumptions of Lemma~\ref{lem:partColor} for $k=3$.}
\label{fig:1}
\end{figure}

Recall that intervals in $\IntervalsFour=\set{\intervalFour_1,\ldots,\intervalFour_{k^2-1}}$ satisfy $\intervalFour_i<\intervalFour_{i+1}$ for $i \in \num{k^2-2}$. We first partition $\IntervalsFour$ into $2k-1$ sets of unit intervals as follows: 
$
\IntervalsFour=\IntervalsFour_1\cup\set{\intervalFive_1}\cup\IntervalsFour_2\cup\ldots\cup\set{\intervalFive_{k-1}}\cup\IntervalsFour_{k},
$
where $\IntervalsFour_i=\{ \intervalFour_{(i-1)k+1}, \ldots,\intervalFour_{ik-1} \}$ for $i \in \num{k}$ and $\intervalFive_i=\intervalFour_{ik}$ for $i \in \num{k-1}$. Observe that $\size{\IntervalsFour_i}=k-1$. Thus, the intervals of the partition are ordered as follows:
$\prefix{\Intervals}{\interval_{k}} \linearOrder \IntervalsFour_1 < \{ \intervalFive_1 \} < \IntervalsFour_2 < \ldots < \{ \intervalFive_{k-1} \} < \IntervalsFour_{k} \linearOrder \suffix{\Intervals}{\interval_{n-k+1}}$.
Next we partition $\Intervals$ into $k$ parts:
$
\Intervals=\Intervals_1\cup\Intervals_2\cup\ldots\cup\Intervals_{k}
$,
in the way that the following holds:
$\prefix{\Intervals}{\interval_{k}} \subseteq \Intervals_1 \linearOrder \{ \intervalFive_1 \} \linearOrder \Intervals_2 \linearOrder \ldots \linearOrder \{ \intervalFive_{k-1} \} \linearOrder \Intervals_{k} \supseteq \suffix{\Intervals}{\interval_{n-k+1}}$.
This partition is pictured in Figure~\ref{fig:2}. 
\begin{figure}[hbpt!]

\begin{small}

\centering

\begin{tikzpicture}[xscale={2\textwidth/94}, yscale=15]
	\begin{pgfonlayer}{nodelayer}
		\node [style=none] (0) at (0, 0) {};
		\node [style=none] (1) at (1, 2) {};
		\node [style=none] (2) at (4, 4) {};
		\node [style=none] (3) at (7, 0) {};
		\node [style=none] (4) at (8, 2) {};
		\node [style=none] (5) at (8.5, 2) {};
		\node [style=none] (6) at (9.5, 0) {};
		\node [style=none] (7) at (11, 4) {};
		\node [style=none] (8) at (11.5, 4) {};
		\node [style=none] (9) at (15.5, 2) {};
		\node [style=none] (10) at (16.5, 0) {};
		\node [style=none] (11) at (17.5, 0) {};
		\node [style=none] (12) at (18.5, 4) {};
		\node [style=none] (13) at (19, 4) {};
		\node [style=none] (14) at (20.5, 2) {};
		\node [style=none] (15) at (24.5, 0) {};
		\node [style=none] (16) at (25, 0) {};
		\node [style=none] (17) at (26, 4) {};
		\node [style=none] (18) at (26.5, 4) {};
		\node [style=none] (19) at (27.5, 2) {};
		\node [style=none] (20) at (28.5, 2) {};
		\node [style=none] (21) at (32, 0) {};
		\node [style=none] (22) at (33, 0) {};
		\node [style=none] (23) at (33.5, 4) {};
		\node [style=none] (24) at (34.5, 4) {};
		\node [style=none] (25) at (35.5, 2) {};
		\node [style=none] (26) at (36.5, 2) {};
		\node [style=none] (27) at (40, 0) {};
		\node [style=none] (28) at (41, 0) {};
		\node [style=none] (29) at (41.5, 4) {};
		\node [style=none] (30) at (42, 4) {};
		\node [style=none] (31) at (43.5, 2) {};
		\node [style=none] (32) at (44, 2) {};
		\node [style=none] (33) at (48, 0) {};
		\node [style=none] (34) at (48.5, 0) {};
		\node [style=none] (35) at (49, 4) {};
		\node [style=none] (36) at (49.5, 4) {};
		\node [style=none] (37) at (51, 2) {};
		\node [style=none] (38) at (55.5, 0) {};
		\node [style=none] (39) at (56, 0) {};
		\node [style=none] (40) at (53, 2) {};
		\node [style=none] (41) at (56.5, 4) {};
		\node [style=none] (42) at (63, 0) {};
		\node [style=none] (43) at (64.5, 0) {};
		\node [style=none] (44) at (61.5, 4) {};
		\node [style=none] (45) at (60, 2) {};
		\node [style=none] (46) at (65.5, 2) {};
		\node [style=none] (47) at (71.5, 0) {};
		\node [style=none] (49) at (68.5, 4) {};
		\node [style=none] (51) at (1.5, 0.75) {$\interval_1$};
		\node [style=none] (52) at (72.5, 2) {};
		\node [style=none] (53) at (73, 2) {};
		\node [style=none] (56) at (78.5, 0) {};
		\node [style=none] (57) at (80, 2) {};
		\node [style=none] (58) at (80.5, 2) {};
		\node [style=none] (60) at (85, 4) {};
		\node [style=none] (61) at (85.5, 0) {};
		\node [style=none] (62) at (87.5, 2) {};
		\node [style=none] (63) at (92, 4) {};
		\node [style=none] (65) at (2.5, 2.75) {$\interval_2$};
		\node [style=none] (66) at (5.5, 4.75) {$\interval_3$};
		\node [style=none] (67) at (14, 2.75) {$\interval_4$};
		\node [style=none] (68) at (13, 4.75) {$\interval_5$};
		\node [style=none] (69) at (11.25, -0.75) {\textcolor{myviolet}{$\intervalFour_1$}};
		\node [style=none] (72) at (22.5, 2.75) {\textcolor{myviolet}{$\intervalFour_2$}};
		\node [style=none] (73) at (36.5, -1.25) {\textcolor{mygreen}{$\intervalFour_3\!=\!\intervalFive_1$}};
		\node [style=none] (74) at (44, 4.75) {\textcolor{myviolet}{$\intervalFour_4$}};
		\node [style=none] (75) at (58.5, 1.25) {\textcolor{myviolet}{$\intervalFour_5$}};
		\node [style=none] (76) at (65, 5.25) {\textcolor{mygreen}{$\intervalFour_6\!=\!\intervalFive_2$}};
		\node [style=none] (81) at (16.75, -1.75) {\textcolor{myviolet}{$\IntervalsFour'_1$}};
		\node [style=none] (82) at (10, 1) {};
		\node [style=none] (83) at (15.5, 1) {};
		\node [style=none] (84) at (22, 3.5) {};
		\node [style=none] (85) at (9, 0) {};
		\node [style=none] (86) at (11.75, -1.5) {};
		\node [style=none] (87) at (17, 0) {};
		\node [style=none] (88) at (27, 1) {};
		\node [style=none] (89) at (23, 3.5) {};
		\node [style=none] (90) at (18, 1) {};
		\node [style=none] (91) at (27.75, 2) {};
		\node [style=none] (92) at (82.75, 0.75) {$\interval_{n-2}$};
		\node [style=none] (93) at (85, 2.75) {$\interval_{n-1}$};
		\node [style=none] (94) at (86.5, 4.75) {$\interval_{n}$};
		\node [style=none] (103) at (85.5, -1) {};
		\node [style=none] (104) at (88, 0) {};
		\node [style=none] (105) at (93, 4) {};
		\node [style=none] (106) at (88, 5.75) {};
		\node [style=none] (107) at (78.5, 5.75) {};
		\node [style=none] (112) at (16, -1) {};
		\node [style=none] (113) at (27, 3) {};
		\node [style=none] (115) at (59.75, 0.5) {};
		\node [style=none] (116) at (60.75, 2) {};
		\node [style=none] (117) at (59.75, 3) {};
		\node [style=none] (119) at (43.5, 5.5) {};
		\node [style=none] (120) at (44.5, 5.5) {};
		\node [style=none] (121) at (41.75, 4) {};
		\node [style=none] (122) at (43.25, 3) {};
		\node [style=none] (123) at (52, 2) {};
		\node [style=none] (124) at (50, 2.75) {};
		\node [style=none] (126) at (49.25, 4) {};
		\node [style=none] (127) at (51, 3) {};
		\node [style=none] (128) at (48.25, 5) {};
		\node [style=none] (129) at (53.25, 1) {};
		\node [style=none] (130) at (58, 0.5) {};
		\node [style=none] (131) at (10.75, -1.5) {};
		\node [style=none] (132) at (58.25, 3.75) {\textcolor{myviolet}{$\IntervalsFour'_2$}};
		\node [style=none] (133) at (78.25, 4) {\textcolor{myviolet}{$\IntervalsFour'_3\!=\!\emptyset$}};
		\node [style=none] (134) at (-1, 0) {};
		\node [style=none] (135) at (-0.25, 2) {};
		\node [style=none] (136) at (2, 4) {};
		\node [style=none] (137) at (5, 5.75) {};
		\node [style=none] (138) at (6, 5.75) {};
		\node [style=none] (139) at (13.5, 5.75) {};
		\node [style=none] (140) at (0, -1) {};
		\node [style=none] (141) at (15, -2.75) {};
		\node [style=none] (142) at (16.25, -2.75) {};
		\node [style=none] (143) at (26, 5.75) {};
		\node [style=none] (144) at (33.75, 4) {};
		\node [style=none] (145) at (35.75, 2) {};
		\node [style=none] (149) at (32.5, 0) {};
		\node [style=none] (150) at (30.5, -1.5) {};
		\node [style=none] (151) at (18, -2.75) {};
		\node [style=none] (152) at (29, -2.75) {\textcolor{mygreen}{$\Intervals_1$}};
		\node [style=none] (153) at (36.25, 2) {};
		\node [style=none] (154) at (34.25, 4) {};
		\node [style=none] (155) at (40.5, 0) {};
		\node [style=none] (157) at (43.5, 5.75) {};
		\node [style=none] (158) at (44.5, 5.75) {};
		\node [style=none] (159) at (41.5, -1) {};
		\node [style=none] (160) at (54.75, -1) {};
		\node [style=none] (161) at (62.25, -1) {};
		\node [style=none] (162) at (63.5, 0) {};
		\node [style=none] (164) at (59.25, 4.75) {};
		\node [style=none] (165) at (52.5, -2) {\textcolor{mygreen}{$\Intervals_2$}};
		\node [style=none] (166) at (64, 0) {};
		\node [style=none] (168) at (65.5, 3) {};
		\node [style=none] (169) at (65, -1) {};
		\node [style=none] (171) at (92, 2.75) {};
		\node [style=none] (172) at (76.5, -2) {\textcolor{mygreen}{$\Intervals_3$}};
	\end{pgfonlayer}
	\begin{pgfonlayer}{edgelayer}
		\draw [style=intBlue] (0.center) to (3.center);
		\draw [style=intYellow] (1.center) to (4.center);
		\draw [style=intRed] (2.center) to (7.center);
		\draw [style=interval] (5.center) to (9.center);
		\draw [style=intDashedViolet] (6.center) to (10.center);
		\draw [style=interval] (8.center) to (12.center);
		\draw [style=interval] (11.center) to (15.center);
		\draw [style=interval] (13.center) to (17.center);
		\draw [style=intDashedViolet] (14.center) to (19.center);
		\draw [style=interval] (16.center) to (21.center);
		\draw [style=interval] (18.center) to (23.center);
		\draw [style=interval] (20.center) to (25.center);
		\draw [style=intDashedGreen] (22.center) to (27.center);
		\draw [style=interval] (24.center) to (29.center);
		\draw [style=interval] (26.center) to (31.center);
		\draw [style=interval] (28.center) to (33.center);
		\draw [style=intDashedViolet] (30.center) to (35.center);
		\draw [style=interval] (32.center) to (37.center);
		\draw [style=interval] (34.center) to (38.center);
		\draw [style=interval] (36.center) to (41.center);
		\draw [style=interval] (39.center) to (42.center);
		\draw [style=intDashedViolet] (40.center) to (45.center);
		\draw [style=interval] (43.center) to (47.center);
		\draw [style=intDashedGreen] (44.center) to (49.center);
		\draw [style=interval] (46.center) to (52.center);
		\draw [style=interval] (53.center) to (57.center);
		\draw [style=intYellow] (56.center) to (61.center);
		\draw [style=intRed] (58.center) to (62.center);
		\draw [style=intBlue] (60.center) to (63.center);
		\draw [style=areaViolet, in=0, out=-180, looseness=0.75] (89.center) to (84.center);
		\draw [style=areaViolet, in=0, out=180] (84.center) to (83.center);
		\draw [style=areaViolet, in=90, out=-180] (82.center) to (85.center);
		\draw [style=areaViolet, in=180, out=90, looseness=1.25] (87.center) to (90.center);
		\draw [style=areaViolet, in=-90, out=0, looseness=1.25] (88.center) to (91.center);
		\draw [style=areaViolet, in=15, out=-90, looseness=1.50] (87.center) to (112.center);
		\draw [style=areaViolet, in=0, out=-165, looseness=0.50] (112.center) to (86.center);
		\draw [style=areaViolet] (88.center) to (90.center);
		\draw [style=areaViolet] (83.center) to (82.center);
		\draw [style=areaViolet, in=-15, out=90, looseness=1.25] (91.center) to (113.center);
		\draw [style=areaViolet, in=360, out=165, looseness=0.50] (113.center) to (89.center);
		\draw [style=areaViolet] (120.center) to (119.center);
		\draw [style=areaViolet, in=90, out=-180, looseness=1.25] (119.center) to (121.center);
		\draw [style=areaViolet, in=-180, out=-90, looseness=1.25] (121.center) to (122.center);
		\draw [style=areaViolet] (122.center) to (124.center);
		\draw [style=areaViolet, in=90, out=0] (124.center) to (123.center);
		\draw [style=areaViolet, in=-90, out=0, looseness=1.25] (115.center) to (116.center);
		\draw [style=areaViolet, in=0, out=90, looseness=1.50] (116.center) to (117.center);
		\draw [style=areaViolet, in=-90, out=180, looseness=1.25] (127.center) to (126.center);
		\draw [style=areaViolet, in=-15, out=90, looseness=1.25] (126.center) to (128.center);
		\draw [style=areaViolet, in=0, out=165] (128.center) to (120.center);
		\draw [style=areaViolet, in=180, out=0] (127.center) to (117.center);
		\draw [style=areaViolet] (130.center) to (115.center);
		\draw [style=areaViolet, in=165, out=-90, looseness=1.25] (123.center) to (129.center);
		\draw [style=areaViolet, in=-15, out=180, looseness=0.75] (130.center) to (129.center);
		\draw [style=areaViolet] (86.center) to (131.center);
		\draw [style=areaViolet, in=-90, out=180, looseness=1.25] (131.center) to (85.center);
		\draw [style=areaGreen, in=-165, out=0, looseness=0.25] (151.center) to (150.center);
		\draw [style=areaGreen, in=-90, out=15] (150.center) to (149.center);
		\draw [style=areaGreen, in=0, out=90, looseness=0.50] (144.center) to (143.center);
		\draw [style=areaGreen, in=360, out=180] (143.center) to (139.center);
		\draw [style=areaGreen] (139.center) to (138.center);
		\draw [style=areaGreen] (138.center) to (137.center);
		\draw [style=areaGreen, in=45, out=-180] (137.center) to (136.center);
		\draw [style=areaGreen, in=45, out=-135] (136.center) to (135.center);
		\draw [style=areaGreen, in=90, out=-135] (135.center) to (134.center);
		\draw [style=areaGreen, in=165, out=-90, looseness=1.25] (134.center) to (140.center);
		\draw [style=areaGreen, in=180, out=-15, looseness=0.50] (140.center) to (141.center);
		\draw [style=areaGreen] (141.center) to (142.center);
		\draw [style=areaGreen] (142.center) to (151.center);
		\draw [style=areaGreen, in=-90, out=90, looseness=1.25] (145.center) to (144.center);
		\draw [style=areaGreen, in=90, out=-90, looseness=1.25] (154.center) to (153.center);
		\draw [style=areaGreen, in=-90, out=90] (155.center) to (153.center);
		\draw [style=areaGreen, in=90, out=-90, looseness=1.25] (145.center) to (149.center);
		\draw [style=areaGreen] (157.center) to (158.center);
		\draw [style=areaGreen, in=180, out=-90] (155.center) to (159.center);
		\draw [style=areaGreen] (160.center) to (159.center);
		\draw [style=areaGreen] (160.center) to (161.center);
		\draw [style=areaGreen, in=-90, out=0, looseness=1.25] (161.center) to (162.center);
		\draw [style=areaGreen, in=360, out=135, looseness=0.25] (164.center) to (158.center);
		\draw [style=areaGreen, in=-45, out=90, looseness=0.50] (162.center) to (164.center);
		\draw [style=areaGreen, in=-90, out=165, looseness=1.25] (169.center) to (166.center);
		\draw [style=areaGreen, in=180, out=0] (168.center) to (107.center);
		\draw [style=areaGreen] (169.center) to (103.center);
		\draw [style=areaGreen, in=210, out=0] (103.center) to (104.center);
		\draw [style=areaGreen, in=-135, out=30, looseness=0.50] (104.center) to (171.center);
		\draw [style=areaGreen, in=-90, out=45, looseness=0.75] (171.center) to (105.center);
		\draw [style=areaGreen] (106.center) to (107.center);
		\draw [style=areaGreen, in=180, out=90, looseness=1.25] (166.center) to (168.center);
		\draw [style=areaGreen, in=-180, out=90, looseness=0.50] (154.center) to (157.center);
		\draw [style=areaGreen, in=0, out=90, looseness=0.75] (105.center) to (106.center);
	\end{pgfonlayer}
\end{tikzpicture}
 
\end{small}
\caption{An example illustrating partitioning the intervals, $k=3$.}
\label{fig:2}
\end{figure}
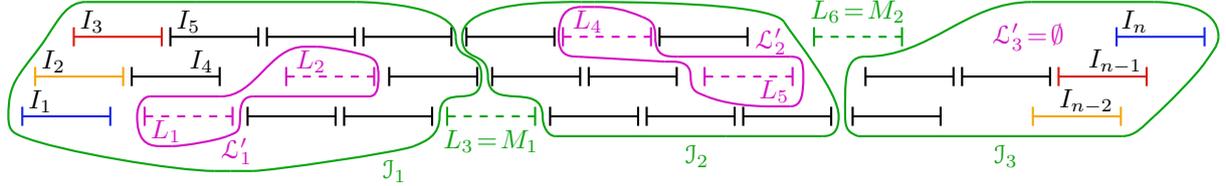
We now enlarge $\Intervals$ by adding to each $\Intervals_i$ some dummy intervals.
More precisely, to each $\Intervals_i \subseteq \Intervals$ we add as many intervals of $\IntervalsFour_i \subseteq \IntervalsFour$  as to make the number of intervals in such extended $\Intervals_i$ a multiple of $k$.
More formally, for each $i \in \num{k}$ let $\IntervalsFour'_i\subseteq\IntervalsFour_i$ be an arbitrary subset of size $\size{\IntervalsFour'_i} = k p_i -\size{\Intervals_i} \in \num{k}$, where $k\bra{p_i-1}<\size{\Intervals_i}\leqslant k p_i$.
Then we define an extension $\IntervalsTwo_i$ of $\Intervals_i$ as $\IntervalsTwo_i\eqdef\Intervals_i\cup\IntervalsFour'_i$ and an extension $\IntervalsTwo$ of $\Intervals$ as
$
\IntervalsTwo\eqdef\IntervalsTwo_1\cup\IntervalsTwo_1\cup\ldots\cup\IntervalsTwo_{k}=\Intervals\cup\IntervalsFour'_1\cup\IntervalsFour'_2\cup\ldots\cup\IntervalsFour'_{k}.
$ The sets $\IntervalsFour'_i$ used for extending $\Intervals_i$ are also shown in Figure~\ref{fig:2}. 
From definition, for $\IntervalsTwo$ and $\IntervalsTwo_1,\IntervalsTwo_2,\ldots,\IntervalsTwo_{k}$ the following properties hold (note that since $\Intervals$ is connected as a graph, each $\Intervals_i$ is nonempty and as a consequence each $\IntervalsTwo_i$ is also nonempty):
\begin{enumerate}\item $\fClique{\IntervalsTwo\cup\set{\intervalFive_1,\ldots,\intervalFive_{k-1}}} \leqslant k$,
\item for each $i\in\num{k}$ we have $\size{\IntervalsTwo_i}= k p_i$ for some $p_i\in\NaturalNumbers\setminus\set{0}$,
\item $\prefix{\Intervals}{\interval_{k}} \subseteq \IntervalsTwo_1 \linearOrder \{ \intervalFive_1 \} \linearOrder \IntervalsTwo_2 \linearOrder \ldots \linearOrder \{ \intervalFive_{k-1} \} \linearOrder \IntervalsTwo_{k} \supseteq \suffix{\Intervals}{\interval_{n-k+1}}.$
\end{enumerate}

Let us now briefly shed some light on the main idea for the proof. First of all, rather than the proper coloring for $\Intervals$, we construct the proper coloring $\coloring$ for $\IntervalsTwo$ (which is obviously also proper for $\Intervals$). The construction of $\coloring$ starts by copying colors of $\coloringPrime$, so that $\coloring$ coincides with $\coloringPrime$ on $\prefix{\Intervals}{\interval_k}$. Let us now consider a block $\IntervalsTwo_i=\set{\intervalTwo^i_0 \linearOrder \ldots \linearOrder \intervalTwo^i_{s_i}}$, $i \in \num{k}$, $s_i=\size{\IntervalsTwo_i}$. Suppose that $\coloring$ is already defined on $\prefix{\IntervalsTwo_i}{\intervalTwo^i_{k-1}}$ (the first $k$ intervals of $\IntervalsTwo_i$). Note that coloring $\coloring$ on $\prefix{\IntervalsTwo_i}{\intervalTwo^i_{k-1}}$ defines a permutation of colors. We use the coloring modulo introduced in Observation~\ref{obs:gcc} to copy this permutation to $\suffix{\IntervalsTwo_i}{\intervalTwo^i_{s_i-k}}$ (the last $k$ intervals of $\IntervalsTwo_i$). This works because $k | s_i$. Suppose that this permutation already coincides with the permutation given by $\coloringDouble$ on the last $i-1$ positions. We now use the slack given by $\intervalFive_i$ to define $\coloring$ on $\prefix{\IntervalsTwo_{i+1}}{\intervalTwo^{i+1}_{k-1}}$ in such a way that the corresponding permutation agrees with the permutation given by $\coloringDouble$ on the last $i$ positions. To be more precise, we apply one step of the insertion sort on the permutation to bring one more color to the appriopriate position. This is possible because of the slack given by $\intervalFive_i$. We now make this idea formal in the following claim.
\begin{claim}
\label{clm:insertionSort}
For each $j\in\num{k}$ there is a proper coloring $\coloring:\IntervalsTwo_1\cup\IntervalsTwo_2\cup\ldots\cup\IntervalsTwo_j\rightarrow [k]$ which satisfies the following invariants:
\begin{enumerate}\item $\fColoring{\intervalTwo^1_i} = \fColoringPrime{\intervalTwo^1_i}$ for each $i \in \set{0, 1, \ldots, k-1}$, \label{clmA:cPrim}
\item $\coloring$ is a bijection on $\suffix{\IntervalsTwo_j}{\intervalTwo^j_{s_j-k}}$ , \label{clmA:bijection}
\item $\fColoring{\intervalTwo^j_{s_j-i+1}} = \fColoringDouble{\intervalTwo^k_{s_k-i+1}}$ for each $i \in \set{2, \ldots, j}$. \label{clmA:cDouble}
\end{enumerate}
\end{claim}
\begin{proof}[Proof of Claim \ref{clm:insertionSort}]
The proof proceedes by induction on $j \in \num{k}$.
For $j=1$ we define $\coloring:\IntervalsTwo_1\rightarrow \num{k}$ in the following way:
$
\fColoring{\intervalTwo^1_{i}} \, \eqdef \, \fColoringPrime{\intervalTwo^1_{\bra{i \!\!\! \mod k}}} \qquad \text{for each} \ i\in\set{0,1,\ldots,s_1-1}.
$
Note that $\coloring$ is a proper $k$-coloring of $\IntervalsTwo_1$ due to Observation~\ref{obs:gcc}. Invariant \ref{clmA:bijection} also follows from the definition of $\coloring$ and from the fact that $\coloringPrime$ is bijection on $\prefix{\IntervalsTwo_1}{\intervalTwo^1_{k-1}}$ (the first $k$ intervals of $\IntervalsTwo_1$).
Invariant \ref{clmA:cDouble} is trivially satisfied.

\smallskip

We move on to proving the inductive step, i.e., we assume that Claim~\ref{clm:insertionSort} is true for $j-1$, where $1 \leqslant {j-1} < k$, and we prove it for $j$.
To this end, we extend the coloring $\coloring:\IntervalsTwo_1\cup\ldots\cup\IntervalsTwo_{j-1}\rightarrow\num{k}$ to the coloring of $\IntervalsTwo_{j}$.
We proceed in two steps.
In the first step we define $\coloring$ on $\prefix{\IntervalsTwo_j}{\intervalTwo^j_{k-1}}$. In the second step we define  $\coloring$ on the remaining suffix of $\IntervalsTwo_j$, i.e., on $\suffix{\IntervalsTwo_j}{\intervalTwo^j_{k}}$.

\medskip

\noindent \emph{Step\,1.}  In this step the goal is to increase the compatibility with the coloring $\coloringDouble$ (see Invariant~\ref{clmA:cDouble}).
To be more precise, we need the color $\gamma \eqdef \fColoringDouble{\intervalTwo^k_{s_k-j+1}}\in\num{k}$ to be assigned to the interval $\intervalTwo^j_{k-j+1}$.
Let $i_0\in\set{0,1,\ldots,k-1}$ be the position where $\gamma$ is assigned on the set $\set{\intervalTwo^{j-1}_{s_{j-1}-k} \linearOrder \ldots \linearOrder \intervalTwo^{j-1}_{s_{j-1}-1}}$, i.e., $\fColoring{\intervalTwo^{j-1}_{\bra{s_{j-1}-k}+i_0}}=\gamma$. Such position extsts since by Invariant~\ref{clmA:bijection} coloring $\coloring$ is a bijection on $\suffix{\IntervalsTwo_{j-1}}{\intervalTwo^{j-1}_{s_{j-1}-k}}$.
Moreover, $i_0 \leqslant k - j + 1$, because  $\coloringDouble$ is a bijection and $\fColoring{\intervalTwo^{j-1}_{s_{j-1}-i+1}} = \fColoringDouble{\intervalTwo^k_{s_k-i+1}}$ for each $i \in \set{2, \ldots, j-1}$ (Invariant \ref{clmA:cDouble} for $j-1$), so $\gamma = \fColoringDouble{\intervalTwo^k_{s_k-j+1}} \notin \fColoringDouble{\set{\intervalTwo^{k}_{s_{k}-j+2}, \ldots, \intervalTwo^{k}_{s_{k}-1}}}= \fColoring{\set{\intervalTwo^{j-1}_{s_{j-1}-j+2}, \ldots, \intervalTwo^{j-1}_{s_{j-1}-1}}}$.
We now define the coloring $\coloring$ on the first $k$ intervals of $\IntervalsTwo_j$, which is supposed to immitate one step of the insertion sort on the current permutation of colors:
\begin{equation}
\label{eq:defCj}
\fColoring{\intervalTwo^{j}_i} \ \eqdef \ 
\begin{cases}
\fColoring{\intervalTwo^{j-1}_{(s_{j-1}-k)+i}} & 0 \leqslant i < i_0, \\
\fColoring{\intervalTwo^{j-1}_{(s_{j-1}-k)+i+1}} & i_0 \leqslant i < k - j + 1, \\
\fColoring{\intervalTwo^{j-1}_{(s_{j-1}-k)+i_0}} & i = k - j + 1, \\
\fColoring{\intervalTwo^{j-1}_{(s_{j-1}-k)+i}} & k - j + 1 < i \leqslant k - 1.
\end{cases}
\end{equation}
An example of one step of the insertion sort defined by Equation (\ref{eq:defCj}) for $j = 2$, $k = 5$, and $i_0=2$ is presented in Figure~\ref{fig:3}.
\begin{figure}[hbpt]

\begin{footnotesize}

\centering

\begin{tikzpicture}[xscale={2\textwidth/55}, yscale=18]
	\begin{pgfonlayer}{nodelayer}
		\node [style=none] (0) at (-3, 0) {};
		\node [style=none] (1) at (5, 0) {};
		\node [style=none] (2) at (-1, 2) {};
		\node [style=none] (3) at (7, 2) {};
		\node [style=none] (4) at (4, 4) {};
		\node [style=none] (5) at (12, 4) {};
		\node [style=none] (6) at (6, 6) {};
		\node [style=none] (7) at (14, 6) {};
		\node [style=none] (8) at (8, 8) {};
		\node [style=none] (9) at (16, 8) {};
		\node [style=none] (10) at (9, -2) {};
		\node [style=none] (11) at (17, -2) {};
		\node [style=none] (12) at (11, 0) {};
		\node [style=none] (13) at (19, 0) {};
		\node [style=none] (14) at (13, 2) {};
		\node [style=none] (15) at (21, 2) {};
		\node [style=none] (16) at (15, 4) {};
		\node [style=none] (17) at (23, 4) {};
		\node [style=none] (18) at (18, 6) {};
		\node [style=none] (19) at (26, 6) {};
		\node [style=none] (20) at (20, 8) {};
		\node [style=none] (21) at (28, 8) {};
		\node [style=none] (22) at (31, 4) {{\huge $\cdots$}};
		\node [style=none] (23) at (34, 0) {};
		\node [style=none] (24) at (42, 0) {};
		\node [style=none] (25) at (37, 2) {};
		\node [style=none] (26) at (45, 2) {};
		\node [style=none] (27) at (38, 4) {};
		\node [style=none] (28) at (46, 4) {};
		\node [style=none] (29) at (43, 6) {};
		\node [style=none] (30) at (51, 6) {};
		\node [style=none] (31) at (44, 8) {};
		\node [style=none] (32) at (52, 8) {};
		\node [style=none] (33) at (13, -1.25) {$\intervalFive_{j-1}$};
		\node [style=none] (34) at (12, 8.75) {$\intervalTwo^{j-1}_{s_{j-1}-1}$};
		\node [style=none] (35) at (10, 6.75) {$\intervalTwo^{j-1}_{s_{j-1}-2}$};
		\node [style=none] (36) at (8, 4.75) {$\intervalTwo^{j-1}_{s_{j-1}-3}$};
		\node [style=none] (37) at (3, 2.75) {$\intervalTwo^{j-1}_{s_{j-1}-4}$};
		\node [style=none] (38) at (1, 0.75) {$\intervalTwo^{j-1}_{s_{j-1}-5}$};
		\node [style=none] (39) at (24, 8.75) {$\intervalTwo^{j}_{4}$};
		\node [style=none] (40) at (22, 6.75) {$\intervalTwo^{j}_{3}$};
		\node [style=none] (41) at (19, 4.75) {$\intervalTwo^{j}_{2}$};
		\node [style=none] (42) at (17, 2.75) {$\intervalTwo^{j}_{1}$};
		\node [style=none] (43) at (15, 0.75) {$\intervalTwo^{j}_{0}$};
		\node [style=none] (44) at (48, 8.75) {$\intervalTwo^{k}_{s_{k}-1}$};
		\node [style=none] (45) at (47, 6.75) {$\intervalTwo^{k}_{s_{k}-2}$};
		\node [style=none] (46) at (42, 4.75) {$\intervalTwo^{k}_{s_{k}-3}$};
		\node [style=none] (47) at (41, 2.75) {$\intervalTwo^{k}_{s_{k}-4}$};
		\node [style=none] (48) at (39, 0.75) {$\intervalTwo^{k}_{s_{k}-5}$};
		\node [style=none] (49) at (17, 8) {};
		\node [style=none] (50) at (19, 8) {};
		\node [style=none] (51) at (15, 6) {};
		\node [style=none] (52) at (17, 6) {};
		\node [style=none] (53) at (12.25, 3.5) {};
		\node [style=none] (54) at (12.75, 2.5) {};
		\node [style=none] (55) at (8, 1.5) {};
		\node [style=none] (56) at (10, 0.5) {};
		\node [style=none] (57) at (6, 0.5) {};
		\node [style=none] (58) at (14, 3.5) {};
	\end{pgfonlayer}
	\begin{pgfonlayer}{edgelayer}
		\draw [style=intViolet] (0.center) to (1.center);
		\draw [style=intYellow] (2.center) to (3.center);
		\draw [style=intGreen] (4.center) to (5.center);
		\draw [style=intBlue] (6.center) to (7.center);
		\draw [style=intRed] (8.center) to (9.center);
		\draw [style=intDashed] (10.center) to (11.center);
		\draw [style=intYellow] (12.center) to (13.center);
		\draw [style=intGreen] (14.center) to (15.center);
		\draw [style=intViolet] (16.center) to (17.center);
		\draw [style=intBlue] (18.center) to (19.center);
		\draw [style=intRed] (20.center) to (21.center);
		\draw [style=intGreen] (23.center) to (24.center);
		\draw [style=intYellow] (25.center) to (26.center);
		\draw [style=intViolet] (27.center) to (28.center);
		\draw [style=intBlue] (29.center) to (30.center);
		\draw [style=intRed] (31.center) to (32.center);
		\draw [style=arrDashedRed] (49.center) to (50.center);
		\draw [style=arrDashedBlue] (51.center) to (52.center);
		\draw [style=arrDashedGreen] (53.center) to (54.center);
		\draw [style=arrDashedYellow] (55.center) to (56.center);
		\draw [style=arrDashedViolet] (57.center) to (58.center);
	\end{pgfonlayer}
\end{tikzpicture}
 
\end{footnotesize}
\caption{Step 1 of constructing coloring $\coloring$ on $\IntervalsTwo_j$ for $j = 2$, $k = 5$, and $i_0=2$.}
\label{fig:3}
\end{figure}

First, we prove that $\coloring$ is a proper coloring on $\IntervalsTwo'$ where
$
\IntervalsTwo' \eqdef \IntervalsTwo_1 \cup \IntervalsTwo_2 \cup \ldots \cup \IntervalsTwo_{j-1} \cup \prefix{\IntervalsTwo_j}{\intervalTwo^j_{k-1}}.
$
By inductive assumption $\coloring$ is proper on $ \IntervalsTwo_1 \cup \IntervalsTwo_2 \cup \ldots \cup \IntervalsTwo_{j-1}$. To complete the proof it suffices to show that for any $\intervalTwo_i^j \in \prefix{\IntervalsTwo_j}{\intervalTwo^j_{k-1}}$ and  $\intervalTwo^p_q \linearOrder \intervalTwo_i^j$ with $\intervalTwo_i^j \cap \intervalTwo^p_q \neq \emptyset$, $\intervalTwo_i^j$ and $\intervalTwo^p_q$ have different colors, i.e., $\fColoring{\intervalTwo^p_q} \neq \fColoring{\intervalTwo_i^j}$.
First, let us assume that $\intervalTwo^p_q \linearOrder \intervalTwo^{j-1}_{(s_{j-1}-k)+i+2}$.
Then, because
$
\intervalTwo^p_q \linearOrder \intervalTwo^{j-1}_{(s_{j-1}-k)+i+2} \linearOrder \ldots \linearOrder \intervalTwo^{j-1}_{s_{j-1}-1} \linearOrder \intervalFive_{j-1} \linearOrder \intervalTwo^{j}_{1} \linearOrder \ldots \linearOrder \intervalTwo^{j}_{i}
$, by Observation~\ref{obs:extremalIntervals} the extreme intervals are disjoint, i.e.,
$
\intervalTwo^p_q < \intervalTwo^{j}_{i}.
$
On the other hand, because $\coloring$ is a bijection on $\suffix{\IntervalsTwo_{j-1}}{\intervalTwo^{j-1}_{s_{j-1}-k}}$  and by definition in Equation~(\ref{eq:defCj}), the remaining intervals $\intervalTwo^{j-1}_{(s_{j-1}-k)+i+2},  \ldots , \intervalTwo^{j-1}_{s_{j-1}-1}, \intervalTwo^{j}_{1}, \ldots, \intervalTwo^{j}_{i-1}$, are assigned different colors than $\intervalTwo^{j}_{i}$.

Now we prove the following counterparts of Invariant~\ref{clmA:bijection} and Invariant~\ref{clmA:cDouble} for $\prefix{\IntervalsTwo_j}{\intervalTwo^j_{k-1}}$:
\begin{enumerate}\setcounter{enumi}{1}
\item $\coloring$ is a bijection on $\prefix{\IntervalsTwo_j}{\intervalTwo^j_{k-1}}$ , and \label{clmAPrim:bijection}
\item $\fColoring{\intervalTwo^j_{i}} = \fColoringDouble{\intervalTwo^k_{\bra{s_k-k}+i}}$ for each $i \in \set{k-j+1, \ldots, k-1}$. \label{clmAPrim:cDouble}
\end{enumerate}
The statement \ref{clmAPrim:bijection} is a consequence of the fact that $\coloring$ is a~bijection on $\suffix{\IntervalsTwo_{j-1}}{\intervalTwo^{j-1}_{s_{j-1}-k}}$, and that the mapping defined by Equation~(\ref{eq:defCj}) is a permutation.
The proof of statement~\ref{clmAPrim:cDouble} breaks down into two cases.
First, $\fColoring{\intervalTwo^j_{i}} = \fColoringDouble{\intervalTwo^k_{\bra{s_k-k}+i}}$ for each $i \in \set{ k-j+2, \ldots, k-1}$ follows from Equation~(\ref{eq:defCj}) and Invariant \ref{clmA:cDouble} for $j-1$.
Finally, $\fColoring{\intervalTwo^j_{k-j+1}} = \fColoringDouble{\intervalTwo^k_{s_k-j+1}}$ follows from the appropriate selection of $i_0$ in Equation (\ref{eq:defCj}).

\medskip

\noindent \emph{Step\,2.} Here, we define the coloring on $\IntervalsTwo_j \setminus \set{\intervalTwo^j_0, \ldots, \intervalTwo^j_{k-1}}$ similarily as for $j=1$.
We set
$
\fColoring{\intervalTwo^{j}_i} \eqdef \fColoring{\intervalTwo^{j}_{\bra{i\!\!\!\! \mod k}}} \quad \text{for} \ i=k,\ldots,s_{j}-1.
$
The coloring $\coloring$ is proper on $\IntervalsTwo_1\cup\ldots\cup\IntervalsTwo_j$ by Observation~\ref{obs:gcc}. Let us argue that Invariant~\ref{clmA:cPrim}, Invariant~\ref{clmA:bijection} and Invariant~\ref{clmA:cDouble} are now satisfied for $j$. Invariant~\ref{clmA:cPrim} for $j$ is an immediate consequence of the inductive assumption.
Because the size $s_{j} = \size{\IntervalsTwo_{j}}$ is a multiple of $k$ we get that $\fColoring{\intervalTwo^{j}_i} = \fColoring{\intervalTwo^{j}_{(s_{j}-k)+i}}$ for $i = 0, \ldots, k-1$, meaning that the coloring $\coloring$ is identical on the first and the last $k$ intervals of $\IntervalsTwo_j$.
As a consequence we obtain that Invariants~\ref{clmA:bijection} and~\ref{clmA:cDouble} follow from Invariants~\ref{clmAPrim:bijection} and~\ref{clmAPrim:cDouble} that we proved for $\prefix{\IntervalsTwo_j}{\intervalTwo^j_{k-1}}$ in Step~1.
This finishes the proof of Claim \ref{clm:insertionSort}.
\end{proof}

For $j = k$ Claim $\ref{clm:insertionSort}$ guarantees the proper coloring $\coloring : \IntervalsTwo \rightarrow \num{k}$ which satisfies the following:
\begin{enumerate}
\item $\fColoring{\intervalTwo^1_i} = \fColoringPrime{\intervalTwo^1_i}$ for each $i \in \set{0, 1, \ldots, k-1}$, \item $\coloring$ is a bijection on $\set{\intervalTwo^k_{s_k-k}, \ldots, \intervalTwo^j_{s_k-1}}$ , \item $\fColoring{\intervalTwo^k_{\bra{s_k-k}+i}} = \fColoringDouble{\intervalTwo^k_{\bra{s_k-k}+i}}$ for each $i \in \set{1, \ldots, k-1}$. \label{cond:cDouble}
\end{enumerate}
Thus, $\coloring$ is the extension of $\fColoringPrime{\cdot}$ and $\fColoringDouble{\cdot}$ on instance $\IntervalsTwo$, where the first and the last $k$ intervals of $\IntervalsTwo$ coincide with the first and the last $k$ intervals of $\Intervals$. In addition to that $\Intervals \subseteq \IntervalsTwo$. Thus, $\coloring$ restricted to $\Intervals$ is the coloring promised in Lemma~\ref{lem:partColor}.
This finishes the proof of Lemma \ref{lem:partColor}.
\end{proof}

As a side effect of Lemma \ref{lem:partColor}, we obtain the following interesting theorem regarding the chromatic number of the unit circular-arc graph.
For the purpose of the argument below, we assume that a point of a circle with circumference $\circumference>0$ is defined as an equivalence claa $\eqClass{\point}\eqdef\set{\point+m\circumference\in\RealNumbers:m\in\IntegerNumbers}$, as to enforce that if we move a distance $\circumference$ from any point on the circle, we end up in the same point (or formally speaking both points are equivallent). 
\begin{theorem}
\label{thm:main-circ-2}
Let $\Arcs$ be a set of unit circular-arcs on the circle such that $\Arcs$ can be extended with $\fClique{\Arcs}^2-1$ not intersecting unit circular-arcs without increasing the clique number. Then $\fChromatic{\Arcs} \leqslant \fClique{\Arcs}$.
\end{theorem}
\begin{proof}Without loss of generality we assume that the unit circular-arcs containing the point $\eqClass{0}$ are an $\fClique{\Arcs}$-clique.
Let us order $\Arcs = \set{\intDef{\point_1}{\point_1+1},\ldots,\intDef{\point_n}{\point_n+1}}$ by the representatives $\point_1, \ldots, \point_n$ of the starting points such that
$
-1 < \point_1 < \point_2 < \ldots < \point_{\omega} \leqslant 0 < \point_{\omega+1} < \point_{\omega+2} < \ldots < \point_n \leqslant \circumference-1
$
where $\omega\eqdef \fClique{\Arcs}$ and $n \eqdef \size{\Arcs}$.
Now we define the set of $n+\omega$ unit intervals on a line (the $\circumference$ distant points are no more equivalent):
\begin{multline*}
\Intervals\eqdef\{\intDef{\point_1}{\point_1+1}, \intDef{\point_2}{\point_2+1}, \ldots, \intDef{\point_n}{\point_n+1},
\\
\intDef{\point_{1}+\circumference}{\point_{1}+\circumference+1},\intDef{\point_{2}+\circumference}{\point_{2}+\circumference+1},\ldots\intDef{\point_{\omega}+\circumference}{\point_{\omega}+\circumference+1}\}
\end{multline*}
and coloring $\coloringPrime$ of intervals $\intDef{\point_1}{\point_1+1}, \ldots, \intDef{\point_{\omega}}{\point_{\omega}+1}$
and coloring $\coloringDouble$ of intervals $\intDef{\point_1+\circumference}{\point_1+\circumference+1}, \ldots, \intDef{\point_{\omega}+\circumference}{\point_{\omega}+\circumference+1}$ such that for each $i=1,2,\ldots,\omega$ there is
$
\fColoringPrime{\intDef{\point_i}{\point_i+1}}\eqdef\fColoringDouble{\intDef{\point_i+\circumference}{\point_i+\circumference+1}}\eqdef i.
$
By the assumptions of the theorem and Lemma~\ref{lem:partColor} we obtain the proper coloring $\coloring$ of the whole $\Intervals$ by colors from $\set{1,2,\ldots,\omega}$ which is the extension of the colorings $\coloringPrime \cup \coloringDouble$, i.e.,
$
\fColoring{\intDef{\point_i}{\point_i+1}}=\fColoring{\intDef{\point_i+\circumference}{\point_i+\circumference+1}}= i
$
for each $i=1,2,\ldots,\omega$.
The final coloring $\coloringArc:\Arcs\rightarrow\set{1,2,\ldots,\omega}$ is obtained by appropriately copying the coloring $\coloring$ in the following way
$
\fColoringArc{\eqClass{\intDef{\point_i}{\point_i+1}}}\eqdef\fColoring{\intDef{\point_i}{\point_i+1}}
$
for each $i=1,2,\ldots,n$.
\end{proof}

\section{The Frogs Game}\label{sec:frog}
In this section we study the \Frog problem introduced in Section~\ref{sec:prelims}. It is a generalisation of a much simpler \SetUnion problem, which can also be seen as a game. In the \SetUnion game we are initially given $n$ sets of size $1$. There are $n-1$ turns in the game. In each turn the adversary picks two sets that get merged, and the cost of merging is the size of the smaller of these two sets. It is easy to see that the total cost is bounded by $n \log n$, since each set element contributes at most $\log n$ times to the merging cost.

In the \Frog game, we start with a sequence of $n$ sets of size $\Rini$. The adversary again in each turn picks two consecutive sets $S$ and $S'$ that get merged. However, the cost we pay is the total size of $\jn$ consecutive sets, either starting from $S'$ and going to the right, or starting from $S$ and going to the left. We only pay the cheaper of these two options. Our result with regarding to \Frog problem is that the total cost encountered in the game is $\fO{\jn \Rini n \log n}$.
The precise statement is given in the following Theorem~\ref{thm:main}. Due to space limitations, we postpone the proof of Theorem~\ref{thm:main} until Appendix~\ref{app:zaba}.
\begin{theorem}
\label{thm:main}
The total cost of the \Frog game is at most $\Rini \bra{2 \jn-1}\bra{n+2 \jn-2}\log_2\frac{n+2 \jn -2}{2 \jn-1}.$
\end{theorem}
 
\section{Unit interval recoloring problem}\label{sec:mainres}
In this section we are interested in the \Pruir problem defined in Section~\ref{sec:prelims}. 
We a $k$-colorable instance of a unit interval graph $\Intervals = \set{\interval_1,\interval_2,\ldots,\interval_n}$.  This instance defines a sequence of $n+1$ instances: $\Intervals_0=\emptyset$ and $\Intervals_j = \set{\interval_1,\ldots,\interval_j}$ for $j>0$. Instance $\Intervals_j$ differs from $\Intervals_{j-1}$ by one interval $\interval_j$.
The parameter $k$ is known to the algorithm from the start. After the interval $\interval_j$ is presented, the algorithm needs to compute a 
proper $k$-coloring $\coloring_j$ for $\Intervals_j$. Our objective is to 
minimize the recourse budget, which is the number of intervals with 
different colors in $\coloring_j$ and $\coloring_{j-1}$. Note that here an index of an interval denotes the time of arrival and has nothing to do with the $\linearOrder$ order. To differentiate with the time in the Frogs game, we refer to index $j$ of interval $\interval_j$ as step, and we say that $\interval_j$ is revealed (or added) in step $j$. The main result of this section is stated in the following theorem.

\begin{theorem}\label{thm:mainres}
There is an algorithm for \Pruir with a total recourse budget of $\fO{k^7 n \log n}$.
\end{theorem}

\begin{proof}
In order to prove Theorem~\ref{thm:mainres}, we first show the \Recolor algorithm for the \Pruir problem. Later we show that \Recolor achieves the claimed recoloring budget. In the algorithm and the analysis we use the $\linearOrder$ order on the arriving intervals, where the ties are broken by the arrival time. That is, if $\interval_j=\interval_{j'}$, and $j'>j$, we set $\interval_j \linearOrder \interval_{j'}$. Note that this is no change to previous definition of $\linearOrder$ order, as we were allowed to break ties in an arbitrary consistent manner. Here, we only use one specific way of breaking the ties, that is the arrival order.  
\subparagraph*{The \Recolor algorithm} The \Recolor algorithm, given instance $\Intervals_{j-1}$, proper coloring $\coloring_{j-1}$ on $\Intervals_{j-1}$, and interval $\interval_j$, finds a proper coloring $\coloring_j$ on $\Intervals_j$. The \Recolor algorithm consists of two phases.
In the first phase \Recolor finds two intervals: $\intervalTwo^{(R)}$ and $\intervalTwo^{(L)}$, such that $\intervalTwo^{(L)} \linearOrder \interval_j \linearOrder \intervalTwo^{(R)}$. The important thing is that there is a proper coloring $\coloring$ for $\Intervals_j$, such that $\coloring$ differs from $\coloring_{j-1}$ only on intervals of $\infix{\Intervals_j}{\interval_j}{\intervalTwo^{(R)}}$. Analogously, there is a proper coloring $\coloring'$ for $\Intervals_j$, such that $\coloring'$ differs from $\coloring_{j-1}$ only on intervals of $\infix{\Intervals_j}{\intervalTwo^{(L)}}{\interval_j}$. In other words, to obtain $\coloring_j$ it is sufficient to recolor intervals of $\infix{\Intervals_j}{\interval_j}{\intervalTwo^{(R)}}$ or intervals of $\infix{\Intervals_j}{\intervalTwo^{(L)}}{\interval_j}$. In the second phase, the \Recolor algorithm recolors the intervals of the smaller set among $\infix{\Intervals_j}{\interval_j}{\intervalTwo^{(R)}}$ and $\infix{\Intervals_j}{\intervalTwo^{(L)}}{\interval_j}$ and sets $\coloring_j$ accordingly.

Next we describe the first phase which consists of finding intervals $\intervalTwo^{(L)}$ and $\intervalTwo^{(R)}$. We only describe the idea for finding $\intervalTwo^{(R)}$, as finding $\intervalTwo^{(L)}$ is analogous. To find $\intervalTwo^{(R)}$ the \Recolor algorithm scans through the intervals of $\Intervals_j$ in $\linearOrder$ order starting in $\interval_j$. Before processing interval $\intervalThree$, the algorithm always checks if $\infix{\Intervals_j}{\interval_j}{\intervalThree}$ is connected as a graph. If this is ever not the case, the algorithm declares $\intervalTwo^{(R)}$ to be the predecessor of $\intervalThree$ and moves on to finding $\intervalTwo^{(L)}$. Thus, it is guaranteed by the algorithm that $\infix{\Intervals_j}{\interval_j}{\intervalTwo^{(R)}}$ is connected.
So the algorithm iterates through $\Intervals_j$ starting from $\interval_j$.
First, the algorithm finds interval $\interval$, such that $\infix{\Intervals_j}{\interval_j}{\interval}=k+1$ (or in other words it skips through the first $k$ intervals).  Then, starting from $\interval$, the algorithm iterates further. It detects
the first interval $\intervalTwo$, for which a set $\IntervalsFour$ of $k^2-1$ disjoint unit intervals, such that $\set{ \interval_j } < \IntervalsFour < \set{ \intervalTwo}$, can be added to $\Intervals_{j}$ in a way that $\fClique{\Intervals_{j} \cup \interval_j \cup \IntervalsFour} \leq k$. The algorithm then returns as $\intervalTwo^{(R)}$ the interval satisfying $\infix{\Intervals_j}{\intervalTwo}{\intervalTwo^{(R)}}=k+1$ (in other words, it skips $k$ intervals from $\intervalTwo$ to find $\intervalTwo^{(R)}$). The procedure \FRight for finding $\intervalTwo^{(R)}$ is given in Algorithm~\ref{alg:findset}. An analogous procedure \FLeft is used to find $\intervalTwo^{(L)}$, with the only difference that we iterate backwards from $\interval_j$. The pseudocode for the recoloring phase is presented in Algorithm~\ref{alg:recolor}. 

At this point let us observe that there is a proper $k$-coloring $\coloring$ for $\Intervals_j$ such that $\coloring$ coincides with $\coloring_{j-1}$ on all intervals but $\infix{\Intervals_j}{\interval_j}{\intervalTwo^{(R)}}$. If the algorithm stopped iterating because it detected that the successor of $\intervalTwo^{(R)}$ is disconnected with $\infix{\Intervals_j}{\interval_j}{\intervalTwo^{(R)}}$, then based on Observation~\ref{obs:gccNC} we can find coloring $\coloring$ using the \ColComp algorithm. Otherwise, coloring $\coloring$ exists based on Observation~\ref{obs:makeBijection} and Lemma~\ref{lem:partColor}. In particular, the algorithm skips $k$ intervals succeding $\interval_j$ and $k$ intervals preceeding $\intervalTwo^{(R)}$ to be able to use Observation~\ref{obs:makeBijection}. Based on this observation, the algorithm generates a proper $k$-coloring $\overline{\coloring}$ which is a bijection on the first and the last $k$ intervals of $\infix{\Intervals_j}{\interval_j}{\intervalTwo^{(R)}}$, and coincides with $\coloring_{j-1}$ outside of $\infix{\Intervals_j}{\interval_j}{\intervalTwo^{(R)}}$. Thus the assumptions of Lemma~\ref{lem:partColor} are satisfied by $\overline{\coloring}$ restricted to $k$ first and last intervals of $\infix{\Intervals_j}{\interval_j}{\intervalTwo^{(R)}}$. Hence, Lemma~\ref{lem:partColor} implies that coloring $\coloring$ exists.
\begin{algorithm}
\SetKwInOut{Input}{Input}
 \SetKwInOut{Output}{Output}
\SetKw{Raise}{raise exception:\ }
 
 \vskip 0.2cm
\Input{A set of unit intervals $\Intervals_j$ ordered by $\linearOrder$, an integer $k$ such that $\fClique{\Intervals_j} \leq k$, an interval $\interval_j \in \Intervals_j$}
\Output{the first interval $\intervalTwo^{(R)} \sqsupset \interval_j$ satisfying one of the following conditions:
\begin{itemize}
 \item $\infix{\Intervals_j}{\interval_j}{\intervalTwo^{(R)}}$ is disjoint from the successor of $\intervalTwo^{(R)}$
 \item for intervals $\interval,\intervalTwo$
 such that $\infix{\Intervals_j}{\interval_j}{\interval}=k+1$ and $\infix{\Intervals_j}{\intervalTwo}{\intervalTwo^{(R)}}=k+1$ there is a set $\IntervalsFour$ of $k^2-1$ disjoint unit intervals such that $\set{\interval} < \IntervalsFour < \set{ \intervalTwo }$ and $\fClique{\Intervals_j \cup \IntervalsFour} \leq k$. 
\end{itemize}
}
$Q \gets \emptyset$;

$current=\interval_j$\;
\For{$i \gets 1 \ldots k$}
{
$add(Q,current)$\;
$prev=current$\;
$current=next(current)$\;
\If{$current \cap prev = \emptyset$} 
{
$\intervalTwo^{(R)}=prev$\;
\Return $\intervalTwo^{(R)}$\;
}
}

$esp=\opy{current}$\;
$\Delta \gets 0$\;

\While{$\Delta+\floor{\opx{current}-esp}) < k^2+1$}
{
$pop(Q)$\;
$add(Q,current)$\;
$prev=current$\;
$current=next(current)$\;
\If{$current \cap prev = \emptyset$} 
{
$\intervalTwo^{(R)}=prev$\;
\Return $\intervalTwo^{(R)}$\;
}

\If{$span(Q) \neq \emptyset$} 
{
$bsp=\opx{span(Q)}$\;
  
$\Delta+=\max(0,\floor{bsp-esp})$\;
$esp=\opy{span(Q)}$\;

}
}
\For{$i \gets 1 \ldots k$}
{
$prev=current$\;
$current=next(current)$\;
\If{$current \cap prev = \emptyset$} 
{
$\intervalTwo^{(R)}=prev$\;
\Return $\intervalTwo^{(R)}$\;
}

}
$\intervalTwo^{(R)}=current$\;
\Return $\intervalTwo^{(R)}$\;
\caption{The procedure \FRight $Q$ is a FIFO queue, $next(\intervalThree)$ returns the next after  $\intervalThree$ interval of $\Intervals_j$ in $\linearOrder$ order or $\emptyset$ if the next interval does not exist, $span(Q)$ returns the span of all intervals in $Q$ or $\emptyset$ if the intervals on $Q$ are not a clique}
\label{alg:findset}
\end{algorithm}

\begin{algorithm}
 \SetKwInOut{Input}{Input}
 \SetKwInOut{Output}{Output}
\SetKw{Raise}{raise exception:\ }
 
 \vskip 0.2cm
 
  \Input{A set of unit intervals $\Intervals_{j-1}$,  $\interval_j$, proper $k$-coloring $c_{j-1}$ on $\Intervals_{j-1}$}
  \Output{A proper $k$-coloring $c_j$ on $\Intervals_j$}
  
  \vskip 0.1cm

  $\intervalTwo^{(R)} \gets \FRight(\Intervals_j,k,\interval_j)$\;

$\intervalTwo^{(L)} \gets \FLeft(\Intervals_j,k,\interval_j)$\;

  find $c_j$ by recoloring the smaller set among $\infix{\Intervals_j}{\intervalTwo^{(L)}}{\interval_j}$ and $\infix{\Intervals_j}{\interval_j}{\intervalTwo^{(R)}}$ using $\ColComp$ (Observation~\ref{obs:gccNC}) or Observation~\ref{obs:makeBijection} combined with Lemma~\ref{lem:partColor}
 
  \caption{The recoloring algorithm \Recolor}
  \label{alg:recolor}
 \end{algorithm}

\subparagraph*{Analysis} The core idea of the analysis is to employ the \Frog game from Section~\ref{sec:frog}. The recoloring cost of \Recolor algorithm in step $j$ is $\min(\size{\infix{\Intervals_j}{\intervalTwo^{(L)}}{\interval_j}},\size{\infix{\Intervals_j}{\interval_j}{\intervalTwo^{(R)}}})$, and we plan to charge it to the cost of an appropriate jump action in the appriopriatly defined instance of the \Frog game.  Before we present the key idea for the analysis, we need to develop some tools to relate the \Frog game with the dynamically changing instance. We start by making (without loss of generality) some assumptions on the final instance $\Intervals_n$. First of all, observe that if the final instance $\Intervals_n$ is not connected as a graph, then the \Recolor algorithm behaves as if it was applied to each connected component separately. Thus, the analysis can be applied to each connected component of $\Intervals_n$. Hence, without loss of generality, we assume that $\Intervals_n$ is connected as a graph.

For every instance $\Intervals_j$, visible to the algorithm after interval $\interval_j$ is presented, we define a set of intervals $\Spanset_j$, which contains spans of all cliques in $\Intervals_j$. To be more precise $\Spanset_j= \{ \spann{\IntervalsTwo} : \IntervalsTwo \subseteq \Intervals_j \text{ and } \IntervalsTwo \text{ is a $k$-clique } \}$. In Figure~\ref{fig:4} the intervals of $\Intervals_j$ are solid black, and their spans $\Spanset_j$ are pictured above as solid black segments. Observe that spans in $\Spanset_j$ are pairwise disjoint so they are linearly ordered by $<$. 
For any two intervals $\interval$ and $\intervalTwo$ we define the distance between them as
$\distInt{\interval}{\intervalTwo}= 
0 \text{ if } \interval \cap \intervalTwo \neq \emptyset \text{ and }
\distInt{\interval}{\intervalTwo}=\min \bra{ \Abs{\opy{\interval}-\opx{\intervalTwo}},\Abs{\opy{\intervalTwo}-\opx{\interval}} } \text{ otherwise.}
$

At this point we would like to make one more assumption on the final instance $\Intervals_n$, mainly that $\Spanset_n$ is such that $\distInt{\Spanel}{\Spanel'} < 1$ for any two consecutive spans $\Spanel,\Spanel' \in \Spanset_n$. If this is not the case, we insert $n'-n$ dummy intervals for some $n'>n$ in the following way. Let $i=1$. As long as there are two consecutive spans $\Spanel,\Spanel' \in \Spanset_{n+i-1}$ with $\distInt{\Spanel}{\Spanel'}\geq 1$, we insert a dummy unit interval $\interval_{n+i}$ such that $\Spanel < \interval_{n+i} < \Spanel'$ and we increase $i$ by one.  Observe, that dummy intervals do not increase the clique number above $k$. Since $\Intervals_n$ is connected as a graph, it holds that $n' \in \fO{kn}$. By Observation~\ref{obs:unitCli} (that follows next), $\size{ \Spanset_{n'} } \in \fO{kn}$. From now on we denote the size of the final instance as $n'$.

\begin{observation}\label{obs:unitCli}
For any $x \in \RealNumbers$, at most $k$ spans of $\Spanset_{n'}$ intersect unit interval $\intDef{x}{x+1}$. 
\end{observation}
\begin{proof}
There is at most $k$ intervals in $\Intervals_{n'}$ with begin coordinate in $\intDef{x}{x+1}$. Each span intersecting $\intDef{x}{x+1}$ can be uniquely associated with a begin coordinate inside $\intDef{x}{x+1}$. \end{proof}

Before we move on to the analysis, we introduce a partition of $\Spanset_{n'}$ for each insertion step $j$. The sets of the partition are referred to as blocks. First of all, let us call the spans in $\Spanset_j$ active in step $j$, while the spans in $\Spanset_{n'} \setminus \Spanset_j$ will be called passive in step $j$. So the active spans are the ones present in the instance $\Intervals_j$, and the other spans of $\Spanset_{n'}$ are passive. In Figure~\ref{fig:4}, active spans of $\Spanset_j$ and intervals of $\Intervals_j$ are solid black, while the passive spans of $\Spanset_{n'} \setminus \Spanset_j$ and future intervals of $\Intervals_{n'} \setminus \Intervals_j$ are dashed blue. In our partition at step $j$, we place all passive spans in singleton blocks, while the consecutive (in $<$ order) active spans are placed in the same block of the partition. The partition is marked green in Figure~\ref{fig:4}. 
\begin{figure}[hbpt!]
\begin{small}

\centering

\begin{tikzpicture}[xscale={2\textwidth/52.75}, yscale=15]
	\begin{pgfonlayer}{nodelayer}
		\node [style=none] (0) at (0, 0) {};
		\node [style=none] (1) at (5.25, 0) {};
		\node [style=none] (2) at (2, 2) {};
		\node [style=none] (3) at (7.25, 2) {};
		\node [style=none] (4) at (3.5, 4) {};
		\node [style=none] (5) at (8.75, 4) {};
		\node [style=none] (6) at (5.75, 0) {};
		\node [style=none] (7) at (11, 0) {};
		\node [style=none] (8) at (7.75, 2) {};
		\node [style=none] (9) at (13, 2) {};
		\node [style=none] (10) at (9.5, 4) {};
		\node [style=none] (11) at (14.75, 4) {};
		\node [style=none] (12) at (11.75, 0) {};
		\node [style=none] (13) at (17, 0) {};
		\node [style=none] (14) at (15.25, 4) {};
		\node [style=none] (15) at (20.5, 4) {};
		\node [style=none] (16) at (15.75, 2) {};
		\node [style=none] (17) at (21, 2) {};
		\node [style=none] (18) at (19.5, 0) {};
		\node [style=none] (19) at (24.75, 0) {};
		\node [style=none] (20) at (21.5, 4) {};
		\node [style=none] (21) at (26.75, 4) {};
		\node [style=none] (22) at (23.75, 2) {};
		\node [style=none] (23) at (28.5, 2) {};
		\node [style=none] (24) at (25.5, 0) {};
		\node [style=none] (25) at (30.75, 0) {};
		\node [style=none] (26) at (27.5, 4) {};
		\node [style=none] (27) at (32.75, 4) {};
		\node [style=none] (28) at (29.25, 2) {};
		\node [style=none] (29) at (34.5, 2) {};
		\node [style=none] (30) at (31.5, 0) {};
		\node [style=none] (31) at (36.75, 0) {};
		\node [style=none] (32) at (33.25, 4) {};
		\node [style=none] (33) at (38.5, 4) {};
		\node [style=none] (34) at (35, 2) {};
		\node [style=none] (35) at (40.25, 2) {};
		\node [style=none] (36) at (37.5, 0) {};
		\node [style=none] (37) at (42.75, 0) {};
		\node [style=none] (38) at (39.25, 4) {};
		\node [style=none] (39) at (44.5, 4) {};
		\node [style=none] (40) at (41, 2) {};
		\node [style=none] (41) at (46.25, 2) {};
		\node [style=none] (42) at (43.5, 0) {};
		\node [style=none] (43) at (48.75, 0) {};
		\node [style=none] (44) at (46.75, 4) {};
		\node [style=none] (45) at (52.25, 4) {};
		\node [style=none] (46) at (47.25, 2) {};
		\node [style=none] (47) at (52.75, 2) {};
		\node [style=none] (48) at (3.5, 9.25) {};
		\node [style=none] (49) at (5.25, 9) {};
		\node [style=none] (50) at (5.25, -1.25) {};
		\node [style=none] (51) at (3.5, -1) {};
		\node [style=none] (52) at (5.75, -1) {};
		\node [style=none] (53) at (7.25, -1.25) {};
		\node [style=none] (54) at (5.75, 9.25) {};
		\node [style=none] (55) at (7.25, 9) {};
		\node [style=none] (56) at (7.75, 9.25) {};
		\node [style=none] (57) at (8.75, 9) {};
		\node [style=none] (58) at (8.75, -1.25) {};
		\node [style=none] (59) at (7.75, -1) {};
		\node [style=none] (60) at (9.5, 9.25) {};
		\node [style=none] (61) at (11, 9) {};
		\node [style=none] (62) at (11, -1.25) {};
		\node [style=none] (63) at (9.5, -1) {};
		\node [style=none] (64) at (11.75, -1) {};
		\node [style=none] (65) at (13, -1.25) {};
		\node [style=none] (66) at (13, 9) {};
		\node [style=none] (67) at (11.75, 9.25) {};
		\node [style=none] (68) at (15.75, 9.25) {};
		\node [style=none] (69) at (17, 9) {};
		\node [style=none] (70) at (17, -1.25) {};
		\node [style=none] (71) at (15.75, -1) {};
		\node [style=none] (72) at (19.5, -1) {};
		\node [style=none] (73) at (20.5, -1.25) {};
		\node [style=none] (74) at (20.5, 9) {};
		\node [style=none] (75) at (19.5, 9.25) {};
		\node [style=none] (76) at (23.75, -1) {};
		\node [style=none] (77) at (24.75, -1.25) {};
		\node [style=none] (78) at (24.75, 9) {};
		\node [style=none] (79) at (23.75, 9.25) {};
		\node [style=none] (80) at (25.5, -1) {};
		\node [style=none] (81) at (26.75, -1.25) {};
		\node [style=none] (82) at (26.75, 9) {};
		\node [style=none] (83) at (25.5, 9.25) {};
		\node [style=none] (84) at (27.5, -1) {};
		\node [style=none] (85) at (28.5, -1.25) {};
		\node [style=none] (86) at (27.5, 9.25) {};
		\node [style=none] (87) at (28.5, 9) {};
		\node [style=none] (88) at (29.25, -1) {};
		\node [style=none] (89) at (30.75, -1.25) {};
		\node [style=none] (90) at (30.75, 9) {};
		\node [style=none] (91) at (29.25, 9.25) {};
		\node [style=none] (92) at (31.5, -1) {};
		\node [style=none] (93) at (32.75, -1.25) {};
		\node [style=none] (94) at (32.75, 9) {};
		\node [style=none] (95) at (31.5, 9.25) {};
		\node [style=none] (96) at (33.25, 9.25) {};
		\node [style=none] (97) at (34.5, 9) {};
		\node [style=none] (98) at (33.25, -1) {};
		\node [style=none] (99) at (34.5, -1.25) {};
		\node [style=none] (100) at (35, -1) {};
		\node [style=none] (101) at (36.75, -1.25) {};
		\node [style=none] (102) at (36.75, 9) {};
		\node [style=none] (103) at (35, 9.25) {};
		\node [style=none] (104) at (37.5, -1) {};
		\node [style=none] (105) at (38.5, -1.25) {};
		\node [style=none] (106) at (37.5, 9.25) {};
		\node [style=none] (107) at (38.5, 9) {};
		\node [style=none] (108) at (39.25, -1) {};
		\node [style=none] (109) at (40.25, -1.25) {};
		\node [style=none] (110) at (39.25, 9.25) {};
		\node [style=none] (111) at (40.25, 9) {};
		\node [style=none] (112) at (41, -1) {};
		\node [style=none] (113) at (42.75, -1.25) {};
		\node [style=none] (114) at (41, 9.25) {};
		\node [style=none] (115) at (42.75, 9) {};
		\node [style=none] (116) at (43.5, -1) {};
		\node [style=none] (117) at (44.5, -1.25) {};
		\node [style=none] (118) at (44.5, 9) {};
		\node [style=none] (119) at (43.5, 9.25) {};
		\node [style=none] (120) at (47.25, -1) {};
		\node [style=none] (121) at (48.75, -1.25) {};
		\node [style=none] (122) at (47.25, 9.25) {};
		\node [style=none] (123) at (48.75, 9) {};
		\node [style=none] (124) at (3.5, 7) {};
		\node [style=none] (125) at (5.25, 7) {};
		\node [style=none] (126) at (5.75, 7) {};
		\node [style=none] (127) at (7.25, 7) {};
		\node [style=none] (128) at (7.75, 7) {};
		\node [style=none] (129) at (8.75, 7) {};
		\node [style=none] (130) at (9.5, 7) {};
		\node [style=none] (131) at (11, 7) {};
		\node [style=none] (132) at (11.75, 7) {};
		\node [style=none] (133) at (13, 7) {};
		\node [style=none] (134) at (15.75, 7) {};
		\node [style=none] (135) at (17, 7) {};
		\node [style=none] (136) at (3.25, 7) {};
		\node [style=none] (137) at (3.75, 8) {};
		\node [style=none] (138) at (8.5, 8) {};
		\node [style=none] (139) at (9, 7) {};
		\node [style=none] (140) at (8.5, 6) {};
		\node [style=none] (141) at (3.75, 6) {};
		\node [style=none] (142) at (9.25, 7) {};
		\node [style=none] (143) at (9.75, 6) {};
		\node [style=none] (144) at (10.75, 6) {};
		\node [style=none] (145) at (11.25, 7) {};
		\node [style=none] (146) at (9.75, 8) {};
		\node [style=none] (147) at (10.75, 8) {};
		\node [style=none] (150) at (11.5, 7) {};
		\node [style=none] (151) at (12, 8) {};
		\node [style=none] (152) at (12.75, 8) {};
		\node [style=none] (153) at (13.25, 7) {};
		\node [style=none] (154) at (12.75, 6) {};
		\node [style=none] (155) at (12, 6) {};
		\node [style=none] (156) at (15.5, 7) {};
		\node [style=none] (157) at (16, 8) {};
		\node [style=none] (158) at (16.75, 8) {};
		\node [style=none] (159) at (17.25, 7) {};
		\node [style=none] (160) at (16.75, 6) {};
		\node [style=none] (161) at (16, 6) {};
		\node [style=none] (162) at (19.5, 7) {};
		\node [style=none] (163) at (20.5, 7) {};
		\node [style=none] (164) at (19.75, 8) {};
		\node [style=none] (165) at (20.25, 8) {};
		\node [style=none] (166) at (20.75, 7) {};
		\node [style=none] (167) at (20.25, 6) {};
		\node [style=none] (168) at (19.75, 6) {};
		\node [style=none] (169) at (19.25, 7) {};
		\node [style=none] (170) at (23.75, 7) {};
		\node [style=none] (171) at (24.75, 7) {};
		\node [style=none] (172) at (25.5, 7) {};
		\node [style=none] (173) at (26.75, 7) {};
		\node [style=none] (174) at (27.5, 7) {};
		\node [style=none] (175) at (28.5, 7) {};
		\node [style=none] (176) at (29.25, 7) {};
		\node [style=none] (177) at (30.75, 7) {};
		\node [style=none] (178) at (31.5, 7) {};
		\node [style=none] (179) at (32.75, 7) {};
		\node [style=none] (180) at (33.25, 7) {};
		\node [style=none] (181) at (34.5, 7) {};
		\node [style=none] (182) at (35, 7) {};
		\node [style=none] (183) at (36.75, 7) {};
		\node [style=none] (184) at (37.5, 7) {};
		\node [style=none] (185) at (38.5, 7) {};
		\node [style=none] (186) at (39.25, 7) {};
		\node [style=none] (187) at (40.25, 7) {};
		\node [style=none] (188) at (41, 7) {};
		\node [style=none] (189) at (42.75, 7) {};
		\node [style=none] (190) at (44.5, 7) {};
		\node [style=none] (191) at (47.25, 7) {};
		\node [style=none] (192) at (48.75, 7) {};
		\node [style=none] (193) at (43.5, 7) {};
		\node [style=none] (194) at (23.5, 7) {};
		\node [style=none] (195) at (24, 8) {};
		\node [style=none] (196) at (24.5, 8) {};
		\node [style=none] (197) at (25, 7) {};
		\node [style=none] (198) at (24.5, 6) {};
		\node [style=none] (199) at (24, 6) {};
		\node [style=none] (200) at (25.25, 7) {};
		\node [style=none] (201) at (25.75, 8) {};
		\node [style=none] (202) at (26.5, 8) {};
		\node [style=none] (203) at (27, 7) {};
		\node [style=none] (204) at (26.5, 6) {};
		\node [style=none] (205) at (25.75, 6) {};
		\node [style=none] (206) at (27.25, 7) {};
		\node [style=none] (207) at (27.75, 8) {};
		\node [style=none] (208) at (28.25, 8) {};
		\node [style=none] (209) at (28.75, 7) {};
		\node [style=none] (210) at (28.25, 6) {};
		\node [style=none] (211) at (27.75, 6) {};
		\node [style=none] (212) at (29, 7) {};
		\node [style=none] (213) at (29.5, 8) {};
		\node [style=none] (214) at (30.5, 8) {};
		\node [style=none] (215) at (31, 7) {};
		\node [style=none] (216) at (30.5, 6) {};
		\node [style=none] (217) at (29.5, 6) {};
		\node [style=none] (218) at (31.25, 7) {};
		\node [style=none] (219) at (31.75, 8) {};
		\node [style=none] (220) at (36.5, 8) {};
		\node [style=none] (221) at (37, 7) {};
		\node [style=none] (222) at (36.5, 6) {};
		\node [style=none] (223) at (31.75, 6) {};
		\node [style=none] (224) at (37.25, 7) {};
		\node [style=none] (225) at (37.75, 8) {};
		\node [style=none] (226) at (38.25, 8) {};
		\node [style=none] (227) at (38.75, 7) {};
		\node [style=none] (228) at (38.25, 6) {};
		\node [style=none] (229) at (37.75, 6) {};
		\node [style=none] (230) at (39, 7) {};
		\node [style=none] (231) at (39.5, 8) {};
		\node [style=none] (232) at (40, 8) {};
		\node [style=none] (233) at (40.5, 7) {};
		\node [style=none] (234) at (40, 6) {};
		\node [style=none] (235) at (39.5, 6) {};
		\node [style=none] (236) at (40.75, 7) {};
		\node [style=none] (237) at (41.25, 8) {};
		\node [style=none] (238) at (42.5, 8) {};
		\node [style=none] (239) at (43, 7) {};
		\node [style=none] (240) at (42.5, 6) {};
		\node [style=none] (241) at (41.25, 6) {};
		\node [style=none] (242) at (43.25, 7) {};
		\node [style=none] (243) at (43.75, 8) {};
		\node [style=none] (244) at (44.25, 8) {};
		\node [style=none] (245) at (44.75, 7) {};
		\node [style=none] (246) at (44.25, 6) {};
		\node [style=none] (247) at (43.75, 6) {};
		\node [style=none] (248) at (47, 7) {};
		\node [style=none] (249) at (47.5, 8) {};
		\node [style=none] (250) at (48.5, 8) {};
		\node [style=none] (251) at (49, 7) {};
		\node [style=none] (252) at (48.5, 6) {};
		\node [style=none] (253) at (47.5, 6) {};
		\node [style=none] (254) at (1, 4) {$\Intervals_{n'}$:};
		\node [style=none] (255) at (1, 7) {$\Spanset_{n'}$:};
		\node [style=none] (256) at (50.75, 7) {\textcolor{mygreen}{$\PartB_j$}};
	\end{pgfonlayer}
	\begin{pgfonlayer}{edgelayer}
		\draw [style=fillGray] (50.center)
			 to (49.center)
			 to [in=30, out=-150, looseness=1.50] (48.center)
			 to (51.center)
			 to [in=-150, out=30, looseness=1.50] cycle;
		\draw [style=fillGray] (55.center)
			 to [in=30, out=-150, looseness=1.50] (54.center)
			 to (52.center)
			 to [in=-150, out=30, looseness=1.50] (53.center)
			 to cycle;
		\draw [style=fillGray] (56.center)
			 to [in=-150, out=30, looseness=1.50] (57.center)
			 to (58.center)
			 to [in=30, out=-150, looseness=1.50] (59.center)
			 to cycle;
		\draw [style=fillBlue] (63.center)
			 to (60.center)
			 to [in=-150, out=30, looseness=1.50] (61.center)
			 to (62.center)
			 to [in=30, out=-150, looseness=1.50] cycle;
		\draw [style=fillBlue] (67.center)
			 to (64.center)
			 to [in=-150, out=30, looseness=1.25] (65.center)
			 to (66.center)
			 to [in=30, out=-150, looseness=1.25] cycle;
		\draw [style=fillBlue] (68.center)
			 to [in=-150, out=30, looseness=1.25] (69.center)
			 to (70.center)
			 to [in=30, out=-150, looseness=1.25] (71.center)
			 to cycle;
		\draw [style=fillBlue] (74.center)
			 to [in=30, out=-150, looseness=1.25] (75.center)
			 to (72.center)
			 to [in=-150, out=30, looseness=1.25] (73.center)
			 to cycle;
		\draw [style=fillBlue] (78.center)
			 to [in=30, out=-150, looseness=1.25] (79.center)
			 to (76.center)
			 to [in=-150, out=30, looseness=1.25] (77.center)
			 to cycle;
		\draw [style=fillBlue] (80.center)
			 to [in=-150, out=30, looseness=1.25] (81.center)
			 to (82.center)
			 to [in=30, out=-150, looseness=1.25] (83.center)
			 to cycle;
		\draw [style=fillBlue] (85.center)
			 to (87.center)
			 to [in=30, out=-150, looseness=1.50] (86.center)
			 to (84.center)
			 to [in=-150, out=30, looseness=1.50] cycle;
		\draw [style=fillBlue] (91.center)
			 to (88.center)
			 to [in=-150, out=30, looseness=1.25] (89.center)
			 to (90.center)
			 to [in=30, out=-150, looseness=1.25] cycle;
		\draw [style=fillGray] (94.center)
			 to [in=30, out=-150, looseness=1.25] (95.center)
			 to (92.center)
			 to [in=-150, out=30, looseness=1.25] (93.center)
			 to cycle;
		\draw [style=fillGray] (98.center)
			 to [in=-150, out=30, looseness=1.25] (99.center)
			 to (97.center)
			 to [in=30, out=-150, looseness=1.25] (96.center)
			 to cycle;
		\draw [style=fillGray] (103.center)
			 to (100.center)
			 to [in=-150, out=30, looseness=1.25] (101.center)
			 to (102.center)
			 to [in=30, out=-150, looseness=1.25] cycle;
		\draw [style=fillBlue] (107.center)
			 to [in=30, out=-150, looseness=1.25] (106.center)
			 to (104.center)
			 to [in=-150, out=30, looseness=1.25] (105.center)
			 to cycle;
		\draw [style=fillBlue] (109.center)
			 to (111.center)
			 to [in=30, out=-150, looseness=1.25] (110.center)
			 to (108.center)
			 to [in=-150, out=30, looseness=1.25] cycle;
		\draw [style=fillBlue] (113.center)
			 to (115.center)
			 to [in=30, out=-150, looseness=1.25] (114.center)
			 to (112.center)
			 to [in=-150, out=30, looseness=1.50] cycle;
		\draw [style=fillBlue] (119.center)
			 to (116.center)
			 to [in=-150, out=30, looseness=1.25] (117.center)
			 to (118.center)
			 to [in=30, out=-150, looseness=1.25] cycle;
		\draw [style=fillGray] (123.center)
			 to (121.center)
			 to [in=30, out=-150, looseness=1.25] (120.center)
			 to (122.center)
			 to [in=-150, out=30, looseness=1.25] cycle;
		\draw [style=interval] (0.center) to (1.center);
		\draw [style=interval] (2.center) to (3.center);
		\draw [style=interval] (4.center) to (5.center);
		\draw [style=interval] (6.center) to (7.center);
		\draw [style=interval] (8.center) to (9.center);
		\draw [style=intDashedBlue] (10.center) to (11.center);
		\draw [style=interval] (12.center) to (13.center);
		\draw [style=interval] (14.center) to (15.center);
		\draw [style=intDashedBlue] (16.center) to (17.center);
		\draw [style=interval] (18.center) to (19.center);
		\draw [style=intDashedBlue] (20.center) to (21.center);
		\draw [style=interval] (22.center) to (23.center);
		\draw [style=intDashedBlue] (24.center) to (25.center);
		\draw [style=interval] (26.center) to (27.center);
		\draw [style=interval] (28.center) to (29.center);
		\draw [style=interval] (30.center) to (31.center);
		\draw [style=interval] (32.center) to (33.center);
		\draw [style=interval] (34.center) to (35.center);
		\draw [style=intDashedBlue] (36.center) to (37.center);
		\draw [style=interval] (38.center) to (39.center);
		\draw [style=intDashedBlue] (40.center) to (41.center);
		\draw [style=interval] (42.center) to (43.center);
		\draw [style=interval] (44.center) to (45.center);
		\draw [style=interval] (46.center) to (47.center);
		\draw [style=intThick] (124.center) to (125.center);
		\draw [style=intThick] (126.center) to (127.center);
		\draw [style=intThick] (128.center) to (129.center);
		\draw [style=interval] (124.center) to (125.center);
		\draw [style=interval] (126.center) to (127.center);
		\draw [style=interval] (128.center) to (129.center);
		\draw [style=intThickDashedBlue] (130.center) to (131.center);
		\draw [style=intDashedBlue] (130.center) to (131.center);
		\draw [style=areaGreen, in=-90, out=180, looseness=1.25] (141.center) to (136.center);
		\draw [style=areaGreen, in=-180, out=90, looseness=1.25] (136.center) to (137.center);
		\draw [style=areaGreen] (137.center) to (138.center);
		\draw [style=areaGreen, in=90, out=0, looseness=1.25] (138.center) to (139.center);
		\draw [style=areaGreen, in=0, out=-90, looseness=1.25] (139.center) to (140.center);
		\draw [style=areaGreen] (140.center) to (141.center);
		\draw [style=areaGreen] (146.center) to (147.center);
		\draw [style=areaGreen, in=90, out=0, looseness=1.25] (147.center) to (145.center);
		\draw [style=areaGreen, in=0, out=-90, looseness=1.25] (145.center) to (144.center);
		\draw [style=areaGreen] (144.center) to (143.center);
		\draw [style=areaGreen, in=-90, out=-180, looseness=1.25] (143.center) to (142.center);
		\draw [style=areaGreen, in=-180, out=90, looseness=1.25] (142.center) to (146.center);
		\draw [style=areaGreen] (155.center) to (154.center);
		\draw [style=areaGreen, in=-90, out=0, looseness=1.25] (154.center) to (153.center);
		\draw [style=areaGreen, in=0, out=90, looseness=1.25] (153.center) to (152.center);
		\draw [style=areaGreen] (152.center) to (151.center);
		\draw [style=areaGreen, in=90, out=-180, looseness=1.25] (151.center) to (150.center);
		\draw [style=areaGreen, in=180, out=-90, looseness=1.25] (150.center) to (155.center);
		\draw [style=intDashedBlue] (132.center) to (133.center);
		\draw [style=intThickDashedBlue] (132.center) to (133.center);
		\draw [style=intThickDashedBlue] (134.center) to (135.center);
		\draw [style=intDashedBlue] (134.center) to (135.center);
		\draw [style=areaGreen, in=180, out=90, looseness=1.25] (156.center) to (157.center);
		\draw [style=areaGreen] (157.center) to (158.center);
		\draw [style=areaGreen, in=90, out=0, looseness=1.25] (158.center) to (159.center);
		\draw [style=areaGreen, in=0, out=-90, looseness=1.25] (159.center) to (160.center);
		\draw [style=areaGreen] (160.center) to (161.center);
		\draw [style=areaGreen, in=-90, out=-180, looseness=1.25] (161.center) to (156.center);
		\draw [style=intDashedBlue] (162.center) to (163.center);
		\draw [style=intThickDashedBlue] (162.center) to (163.center);
		\draw [style=areaGreen] (164.center) to (165.center);
		\draw [style=areaGreen, in=90, out=0, looseness=1.25] (165.center) to (166.center);
		\draw [style=areaGreen, in=0, out=-90, looseness=1.25] (166.center) to (167.center);
		\draw [style=areaGreen] (167.center) to (168.center);
		\draw [style=areaGreen, in=-90, out=-180, looseness=1.25] (168.center) to (169.center);
		\draw [style=areaGreen, in=-180, out=90, looseness=1.25] (169.center) to (164.center);
		\draw [style=intDashedBlue] (170.center) to (171.center);
		\draw [style=intDashedBlue] (172.center) to (173.center);
		\draw [style=intDashedBlue] (174.center) to (175.center);
		\draw [style=intDashedBlue] (176.center) to (177.center);
		\draw [style=intDashedBlue] (184.center) to (185.center);
		\draw [style=intDashedBlue] (186.center) to (187.center);
		\draw [style=intDashedBlue] (188.center) to (189.center);
		\draw [style=intDashedBlue] (193.center) to (190.center);
		\draw [style=intThickDashedBlue] (170.center) to (171.center);
		\draw [style=intThickDashedBlue] (172.center) to (173.center);
		\draw [style=intThickDashedBlue] (174.center) to (175.center);
		\draw [style=intThickDashedBlue] (176.center) to (177.center);
		\draw [style=intThickDashedBlue] (184.center) to (185.center);
		\draw [style=intThickDashedBlue] (186.center) to (187.center);
		\draw [style=intThickDashedBlue] (188.center) to (189.center);
		\draw [style=intThickDashedBlue] (193.center) to (190.center);
		\draw [style=interval] (191.center) to (192.center);
		\draw [style=interval] (178.center) to (179.center);
		\draw [style=interval] (180.center) to (181.center);
		\draw [style=interval] (182.center) to (183.center);
		\draw [style=intThick] (178.center) to (179.center);
		\draw [style=intThick] (180.center) to (181.center);
		\draw [style=intThick] (182.center) to (183.center);
		\draw [style=intThick] (191.center) to (192.center);
		\draw [style=areaGreen, in=-90, out=-180, looseness=1.25] (199.center) to (194.center);
		\draw [style=areaGreen, in=180, out=90, looseness=1.25] (194.center) to (195.center);
		\draw [style=areaGreen] (195.center) to (196.center);
		\draw [style=areaGreen, in=90, out=0, looseness=1.25] (196.center) to (197.center);
		\draw [style=areaGreen, in=0, out=-90, looseness=1.25] (197.center) to (198.center);
		\draw [style=areaGreen] (198.center) to (199.center);
		\draw [style=areaGreen, in=-90, out=-180, looseness=1.25] (205.center) to (200.center);
		\draw [style=areaGreen, in=180, out=90, looseness=1.25] (200.center) to (201.center);
		\draw [style=areaGreen] (201.center) to (202.center);
		\draw [style=areaGreen, in=90, out=0, looseness=1.25] (202.center) to (203.center);
		\draw [style=areaGreen, in=0, out=-90, looseness=1.25] (203.center) to (204.center);
		\draw [style=areaGreen] (204.center) to (205.center);
		\draw [style=areaGreen, in=-90, out=-180, looseness=1.25] (211.center) to (206.center);
		\draw [style=areaGreen, in=180, out=90, looseness=1.25] (206.center) to (207.center);
		\draw [style=areaGreen] (207.center) to (208.center);
		\draw [style=areaGreen, in=90, out=0, looseness=1.25] (208.center) to (209.center);
		\draw [style=areaGreen, in=0, out=-90, looseness=1.25] (209.center) to (210.center);
		\draw [style=areaGreen] (210.center) to (211.center);
		\draw [style=areaGreen, in=-90, out=-180, looseness=1.25] (217.center) to (212.center);
		\draw [style=areaGreen, in=180, out=90, looseness=1.25] (212.center) to (213.center);
		\draw [style=areaGreen] (213.center) to (214.center);
		\draw [style=areaGreen, in=90, out=0, looseness=1.25] (214.center) to (215.center);
		\draw [style=areaGreen, in=0, out=-90, looseness=1.25] (215.center) to (216.center);
		\draw [style=areaGreen] (216.center) to (217.center);
		\draw [style=areaGreen, in=-90, out=180, looseness=1.25] (223.center) to (218.center);
		\draw [style=areaGreen, in=-180, out=90, looseness=1.25] (218.center) to (219.center);
		\draw [style=areaGreen] (219.center) to (220.center);
		\draw [style=areaGreen, in=90, out=0, looseness=1.25] (220.center) to (221.center);
		\draw [style=areaGreen, in=0, out=-90, looseness=1.25] (221.center) to (222.center);
		\draw [style=areaGreen] (222.center) to (223.center);
		\draw [style=areaGreen, in=-90, out=-180, looseness=1.25] (229.center) to (224.center);
		\draw [style=areaGreen, in=-180, out=90, looseness=1.25] (224.center) to (225.center);
		\draw [style=areaGreen] (225.center) to (226.center);
		\draw [style=areaGreen, in=90, out=0, looseness=1.25] (226.center) to (227.center);
		\draw [style=areaGreen, in=0, out=-90, looseness=1.25] (227.center) to (228.center);
		\draw [style=areaGreen] (228.center) to (229.center);
		\draw [style=areaGreen, in=-90, out=-180, looseness=1.25] (235.center) to (230.center);
		\draw [style=areaGreen, in=-180, out=90, looseness=1.25] (230.center) to (231.center);
		\draw [style=areaGreen] (231.center) to (232.center);
		\draw [style=areaGreen, in=90, out=0, looseness=1.25] (232.center) to (233.center);
		\draw [style=areaGreen, in=0, out=-90, looseness=1.25] (233.center) to (234.center);
		\draw [style=areaGreen] (234.center) to (235.center);
		\draw [style=areaGreen, in=-90, out=-180, looseness=1.25] (241.center) to (236.center);
		\draw [style=areaGreen, in=-180, out=90, looseness=1.25] (236.center) to (237.center);
		\draw [style=areaGreen] (237.center) to (238.center);
		\draw [style=areaGreen, in=90, out=0, looseness=1.25] (238.center) to (239.center);
		\draw [style=areaGreen, in=0, out=-90, looseness=1.25] (239.center) to (240.center);
		\draw [style=areaGreen] (240.center) to (241.center);
		\draw [style=areaGreen, in=-90, out=-180, looseness=1.25] (247.center) to (242.center);
		\draw [style=areaGreen, in=180, out=90, looseness=1.25] (242.center) to (243.center);
		\draw [style=areaGreen] (243.center) to (244.center);
		\draw [style=areaGreen, in=90, out=0, looseness=1.25] (244.center) to (245.center);
		\draw [style=areaGreen, in=0, out=-90, looseness=1.25] (245.center) to (246.center);
		\draw [style=areaGreen] (246.center) to (247.center);
		\draw [style=areaGreen, in=-90, out=-180, looseness=1.25] (253.center) to (248.center);
		\draw [style=areaGreen, in=-180, out=90, looseness=1.25] (248.center) to (249.center);
		\draw [style=areaGreen] (249.center) to (250.center);
		\draw [style=areaGreen, in=90, out=0, looseness=1.25] (250.center) to (251.center);
		\draw [style=areaGreen, in=0, out=-90, looseness=1.25] (251.center) to (252.center);
		\draw [style=areaGreen] (252.center) to (253.center);
	\end{pgfonlayer}
\end{tikzpicture}
 
\end{small}
\caption{Example span partition : {\protect\tikzInterval} -- intervals of $\Intervals_{j}$, {\protect\tikzIntDashedBlue} -- intervals of $\Intervals_{n'} \setminus \Intervals_{j}$ (future intervals), {\protect\tikzIntThick} -- spans of $\Spanset_{j}$, {\protect\tikzIntThickDashedBlue} -- spans of $\Spanset_{n'} \setminus \Spanset_{j}$ (future spans), {\protect\tikzAreaGreen} -- blocks of partition $\PartB_j$.}
\label{fig:4}
\end{figure}
To be more formal, let us first enumerate all passive spans, i.e., let $\Spanset_{n'} \setminus \Spanset_j = \set{ P_1, P_2, \ldots, P_{p_j} }$, where $P_i < P_{i+1}$ for $i \in \num{p_j-1}$.
Let $\mathcal{A}_1= \set{ \Spanel \in \Spanset_{n'}: \Spanel < P_1 }$, $\mathcal{A}_{p_j+1}=\set{ \Spanel \in \Spanset_{n'}: P_{p_j} < \Spanel }$ and $\mathcal{A}_{i+1}=\set{\Spanel \in \Spanset_{n'}: P_i < \Spanel < P_{i+1}}$ for $i \in \num{p_j-1}$. We define the partition $\PartB_j \eqdef \bigcup_{i \in \num{p_j+1}: \mathcal{A}_i \neq \emptyset} \set{ \mathcal{A}_i } \cup \bigcup_{i \in \num{p_j}} \set{ \set{P_i} }$. The non empty blocks  $\mathcal{A}_i$ of the partition are called active, while the blocks $\set{ P_i }$ are called passive. Let us now enumerate blocks of $\PartB_j$ according to $<$ order, i.e., let $\PartB_j=\set{ \Block^j_1, \ldots, \Block^j_{b_j} }$, where $\Block^j_i < \Block^j_{i+1}$ for $i \in \num{b_j-1}$. Note that as long as there is no $k$-clique in $\Intervals_j$, all blocks of $\PartB_j$ are passive singleton blocks. We now plan to relate the cost of \Recolor algorithm in step $j$ with the span partition that we introduced. For that, let us first make a few more definitions. For any set $\Block$ of consecutive spans in $\Spanset_{n'}$ (according to $<$ order) we denote by $\intervals{j}{\Block}$ the set of all intervals in $\Intervals_j$ intersecting a span in $\Block$.  Observe that $\intervals{j}{\Block}$ are consecutive (according to $\linearOrder$) intervals of $\Intervals_j$. Observe also, that since two consecutive spans of $\Spanset_{n'}$ are at distance less then $1$, there is a span $\Spanel^{(j)} \in \Spanset_{n'}$ such that $\Spanel^{(j)} \subseteq \interval_j$ (if there are more than one such spans, take as $\Spanel^{(j)}$ any of them). Then, $\Spanel^{(j)}$ is passive in step $j-1$, and hence it constitutes a passive singleton block in $\PartB_{j-1}$. Let $s(j)$ be such that $\Block^{j-1}_{s(j)}=\set{\Spanel^{(j)}}$, i.e., $\set{\Spanel^{(j)}}$ is the $s(j)$'th block of $\PartB_{j-1}$ in $<$ order. The observations that we make next state, that the intervals recolored by \Recolor algorithm in step $j$ span over only a small number $\gamma=\fO{k^3}$ of consecutive blocks of partition $\PartB_j$. Due to space limitations, the proof of Observation~\ref{obs:algToBlR} is deferred to Appendix~\ref{app:connection}, while the proof of Observation~\ref{obs:algToBlL} is analogous to the proof of Observation~\ref{obs:algToBlR}. 
\begin{observation}\label{obs:algToBlR}Let $\gamma=4k^3+2k(k+3)-1$ and $\Block^{j-1}_{s(j)+i}\eqdef \emptyset$ for $s(j)+i > b_j$. Then,\\
$\infix{\Intervals_j}{\interval_j}{\intervalTwo^{(R)}} \subseteq \set{\interval_j} \cup  \intervals{j-1}{\Block^{j-1}_{s(j)} \cup \Block^{j-1}_{s(j)+1} \cup \ldots \cup \Block^{j-1}_{s(j)+\gamma}}$. 
\end{observation}

\begin{observation}\label{obs:algToBlL}Let $\gamma=4k^3+2k(k+3)-1$ and $\Block^{j-1}_{s(j)-i}\eqdef \emptyset$ if $s(j)-i < 1$. Then,\\
$\infix{\Intervals_j}{\intervalTwo^{(L)}}{\interval_j} \subseteq \set{\interval_j} \cup \intervals{j-1}{\Block^{j-1}_{s(j)-\gamma} \cup \Block^{j-1}_{s(j)-\gamma+1} \cup \ldots \cup \Block^{j-1}_{s(j)}}$, 
\end{observation}

We are now ready to introduce the appropriate \Frog game, that will allow us to charge for the recolorings applied by $\Recolor$ algorithm. We run the \Frog game alongside the interval insertions. We maintain a number of invariants that relate the blocks of $\PartB_j$ with the ranks in the maintained \Frog game. 

We now move on to describing the \Frog game. We start by setting the parameter $\jn \eqdef (k+1) \gamma$ for $\gamma=4k^3+2k(k+3)+1$ as in Observation~\ref{obs:algToBlL} and Observation~\ref{obs:algToBlR}.
The idea is that for each insertion step $j$ we execute some jumps in the \Frog game. Thus, by $\ti(j)$ we denote a turn in \Frog game, where all the jumps corresponding to insertion of $\interval_j$ were already performed.
From now on we will denote $|\Spanset_{n'}|$ by $\sizesp$. 
We initialize the \Frog game with the initial rank sequence $\RRanks{1}=\bra{\rrank{1}{1},\ldots,\rrank{1}{(k+1)\sizesp}}$, where $\rrank{1}{i}=\Rini \eqdef k$ for $i \in \num{(k+1)\sizesp}$. We let $\ti(0)=1$, i.e., the initial empty instance $\Intervals_0=\emptyset$ is associated with rank sequence $\RRanks{1}$. Note that there is no insertion in step $0$. Recall that $\RRanks{\ti(j)}=\bra{\rrank{\ti(j)}{1},\ldots,\rrank{\ti(j)}{(k+1)\sizesp-\ti(j)+1}}$ is the rank sequence in the Frogs game in turn $\ti(j)$. The invariants maintained in step $j$ are:
 \begin{enumerate}
  \item There are at least as many ranks as there are blocks in the partition: $(k+1)\sizesp-\ti(j)+1 \geq b_j$
  \item There is a function $f: \num{(k+1)\sizesp-\ti(j)+1} \mapsto \num{b_j}$ assigning each rank in $\RRanks{\ti(j)}$ to a block in $\PartB_j=\set{\Block^j_1,\ldots,\Block^j_{b_j}}$.
\item The ranks assigned to the same block are consecutive, or equvalently $f$ is non-decreasing, i.e., $i' < i''$ implies $f(i') \leq f(i'')$ 
  \item The rank assigned to a block bounds the number of intervals intersecting this block, formally: $\size{ \intervals{j}{\Block^j_i} } \leq \rrank{\ti(j)}{i'}$ if $i =f(i')$
  \item Each passive block $\Block \in \PartB_j$ is assigned at least $k+1-\size{\intervals{j}{\Block}}$ (consecutive) ranks, while each active block is assigned precisely one rank
 \end{enumerate}
The example rank assignment satisfying the invariants is shown in Figure~\ref{fig:6}. Note that active blocks are asigned one rank, while the passive ones are assigned as many ranks as Invariant~5 requires.

For $j = 0$, i.e., when $\ti(j)=1$, there is precisely $(k+1)\sizesp$ ranks in $\RRanks{\ti(j)}$ and precisely $\sizesp=b_0$ singleton passive blocks in $\PartB_j$ (and no active blocks). Note that $\size{\intervals{0}{\Block}}=0$ for all $\Block \in \PartB_j$. We then set $f(i)=(i-1) \; \divi \; (k+1)+1$. In other words each block in $\PartB_j$ is assigned $(k+1)$ consecutive ranks. With so defined function $f$, all invariants are clearly satisfied for $j=0$.
\begin{figure}[hbpt]
\begin{small}

\centering

\begin{tikzpicture}[xscale={2\textwidth/53.5}, yscale=15]
	\begin{pgfonlayer}{nodelayer}
		\node [style=none] (0) at (0, 0) {};
		\node [style=none] (1) at (6.5, 0) {};
		\node [style=none] (2) at (2, 2) {};
		\node [style=none] (3) at (8.5, 2) {};
		\node [style=none] (4) at (3.5, 4) {};
		\node [style=none] (5) at (10, 4) {};
		\node [style=none] (6) at (7, 0) {};
		\node [style=none] (7) at (13.75, 0) {};
		\node [style=none] (8) at (9, 2) {};
		\node [style=none] (9) at (15.75, 2) {};
		\node [style=none] (10) at (10.75, 4) {};
		\node [style=none] (11) at (17.5, 4) {};
		\node [style=none] (12) at (14.5, 0) {};
		\node [style=none] (13) at (21.25, 0) {};
		\node [style=none] (14) at (20, 4) {};
		\node [style=none] (15) at (26.5, 4) {};
		\node [style=none] (16) at (18, 2) {};
		\node [style=none] (17) at (24.5, 2) {};
		\node [style=none] (18) at (23, 0) {};
		\node [style=none] (19) at (29.5, 0) {};
		\node [style=none] (20) at (43, 4) {};
		\node [style=none] (21) at (49.75, 4) {};
		\node [style=none] (22) at (25.25, 2) {};
		\node [style=none] (23) at (31.5, 2) {};
		\node [style=none] (24) at (32, 2) {};
		\node [style=none] (25) at (38.75, 2) {};
		\node [style=none] (26) at (28.25, 4) {};
		\node [style=none] (27) at (34.75, 4) {};
		\node [style=none] (30) at (33.75, 0) {};
		\node [style=none] (31) at (40.5, 0) {};
		\node [style=none] (32) at (35.5, 4) {};
		\node [style=none] (33) at (42.25, 4) {};
		\node [style=none] (34) at (39.5, 2) {};
		\node [style=none] (35) at (46.25, 2) {};
		\node [style=none] (36) at (41.25, 0) {};
		\node [style=none] (37) at (48, 0) {};
		\node [style=none] (40) at (47, 2) {};
		\node [style=none] (41) at (53.5, 2) {};
		\node [style=none] (48) at (3.5, 9.25) {};
		\node [style=none] (49) at (6.5, 9) {};
		\node [style=none] (50) at (6.5, -1.25) {};
		\node [style=none] (51) at (3.5, -1) {};
		\node [style=none] (52) at (7, -1) {};
		\node [style=none] (53) at (8.5, -1.25) {};
		\node [style=none] (54) at (7, 9.25) {};
		\node [style=none] (55) at (8.5, 9) {};
		\node [style=none] (56) at (9, 9.25) {};
		\node [style=none] (57) at (10, 9) {};
		\node [style=none] (58) at (10, -1.25) {};
		\node [style=none] (59) at (9, -1) {};
		\node [style=none] (60) at (10.75, 9.25) {};
		\node [style=none] (61) at (13.75, 9) {};
		\node [style=none] (62) at (13.75, -1.25) {};
		\node [style=none] (63) at (10.75, -1) {};
		\node [style=none] (64) at (14.5, -1) {};
		\node [style=none] (65) at (15.75, -1.25) {};
		\node [style=none] (66) at (15.75, 9) {};
		\node [style=none] (67) at (14.5, 9.25) {};
		\node [style=none] (68) at (20, 9.25) {};
		\node [style=none] (69) at (21.25, 9) {};
		\node [style=none] (70) at (21.25, -1.25) {};
		\node [style=none] (71) at (20, -1) {};
		\node [style=none] (72) at (23, -1) {};
		\node [style=none] (73) at (24.5, -1.25) {};
		\node [style=none] (74) at (24.5, 9) {};
		\node [style=none] (75) at (23, 9.25) {};
		\node [style=none] (84) at (33.75, -1) {};
		\node [style=none] (85) at (34.75, -1.25) {};
		\node [style=none] (86) at (33.75, 9.25) {};
		\node [style=none] (87) at (34.75, 9) {};
		\node [style=none] (88) at (35.5, -1) {};
		\node [style=none] (89) at (38.75, -1.25) {};
		\node [style=none] (90) at (38.75, 9) {};
		\node [style=none] (91) at (35.5, 9.25) {};
		\node [style=none] (92) at (25.25, -1) {};
		\node [style=none] (93) at (26.5, -1.25) {};
		\node [style=none] (94) at (26.5, 9) {};
		\node [style=none] (95) at (25.25, 9.25) {};
		\node [style=none] (96) at (28.25, 9.25) {};
		\node [style=none] (97) at (29.5, 9) {};
		\node [style=none] (98) at (28.25, -1) {};
		\node [style=none] (99) at (29.5, -1.25) {};
		\node [style=none] (100) at (39.5, -1) {};
		\node [style=none] (101) at (40.5, -1.25) {};
		\node [style=none] (102) at (40.5, 9) {};
		\node [style=none] (103) at (39.5, 9.25) {};
		\node [style=none] (104) at (41.25, -1) {};
		\node [style=none] (105) at (42.25, -1.25) {};
		\node [style=none] (106) at (41.25, 9.25) {};
		\node [style=none] (107) at (42.25, 9) {};
		\node [style=none] (108) at (47, -1) {};
		\node [style=none] (109) at (48, -1.25) {};
		\node [style=none] (110) at (47, 9.25) {};
		\node [style=none] (111) at (48, 9) {};
		\node [style=none] (112) at (43, -1) {};
		\node [style=none] (113) at (46.25, -1.25) {};
		\node [style=none] (114) at (43, 9.25) {};
		\node [style=none] (115) at (46.25, 9) {};
		\node [style=none] (124) at (3.5, 7) {};
		\node [style=none] (125) at (6.5, 7) {};
		\node [style=none] (126) at (7, 7) {};
		\node [style=none] (127) at (8.5, 7) {};
		\node [style=none] (128) at (9, 7) {};
		\node [style=none] (129) at (10, 7) {};
		\node [style=none] (130) at (10.75, 7) {};
		\node [style=none] (131) at (13.75, 7) {};
		\node [style=none] (132) at (14.5, 7) {};
		\node [style=none] (133) at (15.75, 7) {};
		\node [style=none] (134) at (20, 7) {};
		\node [style=none] (135) at (21.25, 7) {};
		\node [style=none] (136) at (3.25, 7) {};
		\node [style=none] (137) at (3.75, 8) {};
		\node [style=none] (138) at (9.75, 8) {};
		\node [style=none] (139) at (10.25, 7) {};
		\node [style=none] (140) at (9.75, 6) {};
		\node [style=none] (141) at (3.75, 6) {};
		\node [style=none] (142) at (10.5, 7) {};
		\node [style=none] (143) at (11, 6) {};
		\node [style=none] (144) at (13.5, 6) {};
		\node [style=none] (145) at (14, 7) {};
		\node [style=none] (146) at (11, 8) {};
		\node [style=none] (147) at (13.5, 8) {};
		\node [style=none] (150) at (14.25, 7) {};
		\node [style=none] (151) at (14.75, 8) {};
		\node [style=none] (152) at (15.5, 8) {};
		\node [style=none] (153) at (16, 7) {};
		\node [style=none] (154) at (15.5, 6) {};
		\node [style=none] (155) at (14.75, 6) {};
		\node [style=none] (156) at (19.75, 7) {};
		\node [style=none] (157) at (20.25, 8) {};
		\node [style=none] (158) at (21, 8) {};
		\node [style=none] (159) at (21.5, 7) {};
		\node [style=none] (160) at (21, 6) {};
		\node [style=none] (161) at (20.25, 6) {};
		\node [style=none] (162) at (23, 7) {};
		\node [style=none] (163) at (24.5, 7) {};
		\node [style=none] (164) at (23.25, 8) {};
		\node [style=none] (165) at (24.25, 8) {};
		\node [style=none] (166) at (24.75, 7) {};
		\node [style=none] (167) at (24.25, 6) {};
		\node [style=none] (168) at (23.25, 6) {};
		\node [style=none] (169) at (22.75, 7) {};
		\node [style=none] (174) at (33.75, 7) {};
		\node [style=none] (175) at (34.75, 7) {};
		\node [style=none] (176) at (35.5, 7) {};
		\node [style=none] (177) at (38.75, 7) {};
		\node [style=none] (178) at (25.25, 7) {};
		\node [style=none] (179) at (26.5, 7) {};
		\node [style=none] (180) at (28.25, 7) {};
		\node [style=none] (181) at (29.5, 7) {};
		\node [style=none] (182) at (39.5, 7) {};
		\node [style=none] (183) at (40.5, 7) {};
		\node [style=none] (184) at (41.25, 7) {};
		\node [style=none] (185) at (42.25, 7) {};
		\node [style=none] (186) at (47, 7) {};
		\node [style=none] (187) at (48, 7) {};
		\node [style=none] (188) at (43, 7) {};
		\node [style=none] (189) at (46.25, 7) {};
		\node [style=none] (206) at (33.5, 7) {};
		\node [style=none] (207) at (34, 8) {};
		\node [style=none] (208) at (34.5, 8) {};
		\node [style=none] (209) at (35, 7) {};
		\node [style=none] (210) at (34.5, 6) {};
		\node [style=none] (211) at (34, 6) {};
		\node [style=none] (212) at (35.25, 7) {};
		\node [style=none] (213) at (35.75, 8) {};
		\node [style=none] (214) at (38.5, 8) {};
		\node [style=none] (215) at (39, 7) {};
		\node [style=none] (216) at (38.5, 6) {};
		\node [style=none] (217) at (35.75, 6) {};
		\node [style=none] (218) at (39.25, 7) {};
		\node [style=none] (219) at (39.75, 8) {};
		\node [style=none] (220) at (40.25, 8) {};
		\node [style=none] (221) at (40.75, 7) {};
		\node [style=none] (222) at (40.25, 6) {};
		\node [style=none] (223) at (39.75, 6) {};
		\node [style=none] (224) at (41, 7) {};
		\node [style=none] (225) at (41.5, 8) {};
		\node [style=none] (226) at (42, 8) {};
		\node [style=none] (227) at (42.5, 7) {};
		\node [style=none] (228) at (42, 6) {};
		\node [style=none] (229) at (41.5, 6) {};
		\node [style=none] (230) at (46.75, 7) {};
		\node [style=none] (231) at (47.25, 8) {};
		\node [style=none] (232) at (47.75, 8) {};
		\node [style=none] (233) at (48.25, 7) {};
		\node [style=none] (234) at (47.75, 6) {};
		\node [style=none] (235) at (47.25, 6) {};
		\node [style=none] (236) at (42.75, 7) {};
		\node [style=none] (237) at (43.25, 8) {};
		\node [style=none] (238) at (46, 8) {};
		\node [style=none] (239) at (46.5, 7) {};
		\node [style=none] (240) at (46, 6) {};
		\node [style=none] (241) at (43.25, 6) {};
		\node [style=none] (254) at (1, 4) {$\Intervals_{n'}$:};
		\node [style=none] (255) at (1, 7) {$\Spanset_{n'}$:};
		\node [style=none] (256) at (50.25, 7) {\textcolor{mygreen}{$\PartB_j$}};
		\node [style=none] (257) at (25, 7) {};
		\node [style=none] (258) at (25.5, 8) {};
		\node [style=none] (259) at (29.25, 8) {};
		\node [style=none] (260) at (29.75, 7) {};
		\node [style=none] (261) at (29.25, 6) {};
		\node [style=none] (262) at (25.5, 6) {};
		\node [style=textRotated] (263) at (7.25, 12) {$r_1$};
		\node [style=none] (264) at (3.25, 9.5) {};
		\node [style=none] (265) at (10.25, 9.5) {};
		\node [style=none] (266) at (14, 9.5) {};
		\node [style=none] (267) at (10.5, 9.5) {};
		\node [style=textRotated] (268) at (13.25, 12.5) {$r_2, r_3$};
		\node [style=none] (269) at (16, 9.5) {};
		\node [style=none] (270) at (14.25, 9.5) {};
		\node [style=textRotated] (271) at (16.25, 12.5) {$r_4, r_5$};
		\node [style=none] (272) at (21.5, 9.5) {};
		\node [style=none] (273) at (19.75, 9.5) {};
		\node [style=none] (274) at (24.75, 9.5) {};
		\node [style=none] (275) at (22.75, 9.5) {};
		\node [style=none] (276) at (29.75, 9.5) {};
		\node [style=none] (277) at (25, 9.5) {};
		\node [style=none] (278) at (35, 9.5) {};
		\node [style=none] (279) at (33.5, 9.5) {};
		\node [style=none] (280) at (39, 9.5) {};
		\node [style=none] (281) at (35.25, 9.5) {};
		\node [style=none] (282) at (40.75, 9.5) {};
		\node [style=none] (283) at (39.25, 9.5) {};
		\node [style=none] (284) at (42.5, 9.5) {};
		\node [style=none] (285) at (41, 9.5) {};
		\node [style=none] (286) at (46.5, 9.5) {};
		\node [style=none] (287) at (42.75, 9.5) {};
		\node [style=none] (288) at (48.25, 9.5) {};
		\node [style=none] (289) at (46.75, 9.5) {};
		\node [style=textRotated] (290) at (21.75, 12.5) {$r_6, r_7$};
		\node [style=textRotated] (291) at (24.75, 12.5) {$r_8, r_9$};
		\node [style=textRotated] (292) at (28, 12.25) {$r_{10}$};
		\node [style=textRotated] (293) at (35.5, 13) {$r_{11}, r_{12}$};
		\node [style=textRotated] (294) at (38.5, 13) {$r_{13}, r_{14}$};
		\node [style=textRotated] (295) at (40.5, 12.25) {$r_{15}$};
		\node [style=textRotated] (296) at (43.25, 13) {$r_{16}, r_{17}$};
		\node [style=textRotated] (297) at (46.75, 13.5) {$r_{18}, r_{19}, r_{20}$};
		\node [style=textRotated] (298) at (48.75, 13) {$r_{21}, r_{22}$};
	\end{pgfonlayer}
	\begin{pgfonlayer}{edgelayer}
		\draw [style=fillGray] (50.center)
			 to (49.center)
			 to [in=30, out=-150, looseness=1.50] (48.center)
			 to (51.center)
			 to [in=-150, out=30, looseness=1.50] cycle;
		\draw [style=fillGray] (55.center)
			 to [in=30, out=-150, looseness=1.50] (54.center)
			 to (52.center)
			 to [in=-150, out=30, looseness=1.50] (53.center)
			 to cycle;
		\draw [style=fillGray] (56.center)
			 to [in=-150, out=30, looseness=1.50] (57.center)
			 to (58.center)
			 to [in=30, out=-150, looseness=1.50] (59.center)
			 to cycle;
		\draw [style=fillBlue] (63.center)
			 to (60.center)
			 to [in=-150, out=30, looseness=1.50] (61.center)
			 to (62.center)
			 to [in=30, out=-150, looseness=1.50] cycle;
		\draw [style=fillBlue] (67.center)
			 to (64.center)
			 to [in=-150, out=30, looseness=1.25] (65.center)
			 to (66.center)
			 to [in=30, out=-150, looseness=1.25] cycle;
		\draw [style=fillBlue] (68.center)
			 to [in=-150, out=30, looseness=1.25] (69.center)
			 to (70.center)
			 to [in=30, out=-150, looseness=1.25] (71.center)
			 to cycle;
		\draw [style=fillBlue] (74.center)
			 to [in=30, out=-150, looseness=1.25] (75.center)
			 to (72.center)
			 to [in=-150, out=30, looseness=1.25] (73.center)
			 to cycle;
		\draw [style=fillBlue] (85.center)
			 to (87.center)
			 to [in=30, out=-150, looseness=1.50] (86.center)
			 to (84.center)
			 to [in=-150, out=30, looseness=1.50] cycle;
		\draw [style=fillBlue] (91.center)
			 to (88.center)
			 to [in=-150, out=30, looseness=1.25] (89.center)
			 to (90.center)
			 to [in=30, out=-150, looseness=1.25] cycle;
		\draw [style=fillGray] (94.center)
			 to [in=30, out=-150, looseness=1.25] (95.center)
			 to [in=90, out=-90] (92.center)
			 to [in=-150, out=30, looseness=1.25] (93.center)
			 to cycle;
		\draw [style=fillGray] (98.center)
			 to [in=-150, out=30, looseness=1.25] (99.center)
			 to [in=270, out=90] (97.center)
			 to [in=30, out=-150, looseness=1.25] (96.center)
			 to cycle;
		\draw [style=fillGray] (103.center)
			 to (100.center)
			 to [in=-150, out=30, looseness=1.25] (101.center)
			 to (102.center)
			 to [in=30, out=-150, looseness=1.25] cycle;
		\draw [style=fillBlue] (107.center)
			 to [in=30, out=-150, looseness=1.25] (106.center)
			 to (104.center)
			 to [in=-150, out=30, looseness=1.25] (105.center)
			 to cycle;
		\draw [style=fillBlue] (109.center)
			 to (111.center)
			 to [in=30, out=-150, looseness=1.25] (110.center)
			 to (108.center)
			 to [in=-150, out=30, looseness=1.25] cycle;
		\draw [style=fillBlue] (113.center)
			 to (115.center)
			 to [in=30, out=-150, looseness=1.25] (114.center)
			 to (112.center)
			 to [in=-150, out=30, looseness=1.50] cycle;
		\draw [style=interval] (0.center) to (1.center);
		\draw [style=interval] (2.center) to (3.center);
		\draw [style=interval] (4.center) to (5.center);
		\draw [style=interval] (6.center) to (7.center);
		\draw [style=interval] (8.center) to (9.center);
		\draw [style=intDashedBlue] (10.center) to (11.center);
		\draw [style=interval] (12.center) to (13.center);
		\draw [style=interval] (14.center) to (15.center);
		\draw [style=intDashedBlue] (16.center) to (17.center);
		\draw [style=interval] (18.center) to (19.center);
		\draw [style=intDashedBlue] (20.center) to (21.center);
		\draw [style=interval] (22.center) to (23.center);
		\draw [style=intDashedBlue] (24.center) to (25.center);
		\draw [style=interval] (26.center) to (27.center);
		\draw [style=interval] (30.center) to (31.center);
		\draw [style=interval] (32.center) to (33.center);
		\draw [style=interval] (34.center) to (35.center);
		\draw [style=intDashedBlue] (36.center) to (37.center);
		\draw [style=intDashedBlue] (40.center) to (41.center);
		\draw [style=intThick] (124.center) to (125.center);
		\draw [style=intThick] (126.center) to (127.center);
		\draw [style=intThick] (128.center) to (129.center);
		\draw [style=interval] (124.center) to (125.center);
		\draw [style=interval] (126.center) to (127.center);
		\draw [style=interval] (128.center) to (129.center);
		\draw [style=intThickDashedBlue] (130.center) to (131.center);
		\draw [style=intDashedBlue] (130.center) to (131.center);
		\draw [style=areaGreen, in=-90, out=180, looseness=1.25] (141.center) to (136.center);
		\draw [style=areaGreen, in=-180, out=90, looseness=1.25] (136.center) to (137.center);
		\draw [style=areaGreen] (137.center) to (138.center);
		\draw [style=areaGreen, in=90, out=0, looseness=1.25] (138.center) to (139.center);
		\draw [style=areaGreen, in=0, out=-90, looseness=1.25] (139.center) to (140.center);
		\draw [style=areaGreen] (140.center) to (141.center);
		\draw [style=areaGreen] (146.center) to (147.center);
		\draw [style=areaGreen, in=90, out=0, looseness=1.25] (147.center) to (145.center);
		\draw [style=areaGreen, in=0, out=-90, looseness=1.25] (145.center) to (144.center);
		\draw [style=areaGreen] (144.center) to (143.center);
		\draw [style=areaGreen, in=-90, out=-180, looseness=1.25] (143.center) to (142.center);
		\draw [style=areaGreen, in=-180, out=90, looseness=1.25] (142.center) to (146.center);
		\draw [style=areaGreen] (155.center) to (154.center);
		\draw [style=areaGreen, in=-90, out=0, looseness=1.25] (154.center) to (153.center);
		\draw [style=areaGreen, in=0, out=90, looseness=1.25] (153.center) to (152.center);
		\draw [style=areaGreen] (152.center) to (151.center);
		\draw [style=areaGreen, in=90, out=-180, looseness=1.25] (151.center) to (150.center);
		\draw [style=areaGreen, in=180, out=-90, looseness=1.25] (150.center) to (155.center);
		\draw [style=intDashedBlue] (132.center) to (133.center);
		\draw [style=intThickDashedBlue] (132.center) to (133.center);
		\draw [style=intThickDashedBlue] (134.center) to (135.center);
		\draw [style=intDashedBlue] (134.center) to (135.center);
		\draw [style=areaGreen, in=180, out=90, looseness=1.25] (156.center) to (157.center);
		\draw [style=areaGreen] (157.center) to (158.center);
		\draw [style=areaGreen, in=90, out=0, looseness=1.25] (158.center) to (159.center);
		\draw [style=areaGreen, in=0, out=-90, looseness=1.25] (159.center) to (160.center);
		\draw [style=areaGreen] (160.center) to (161.center);
		\draw [style=areaGreen, in=-90, out=-180, looseness=1.25] (161.center) to (156.center);
		\draw [style=intDashedBlue] (162.center) to (163.center);
		\draw [style=intThickDashedBlue] (162.center) to (163.center);
		\draw [style=areaGreen] (164.center) to (165.center);
		\draw [style=areaGreen, in=90, out=0, looseness=1.25] (165.center) to (166.center);
		\draw [style=areaGreen, in=0, out=-90, looseness=1.25] (166.center) to (167.center);
		\draw [style=areaGreen] (167.center) to (168.center);
		\draw [style=areaGreen, in=-90, out=-180, looseness=1.25] (168.center) to (169.center);
		\draw [style=areaGreen, in=-180, out=90, looseness=1.25] (169.center) to (164.center);
		\draw [style=intDashedBlue] (174.center) to (175.center);
		\draw [style=intDashedBlue] (176.center) to (177.center);
		\draw [style=intDashedBlue] (184.center) to (185.center);
		\draw [style=intDashedBlue] (186.center) to (187.center);
		\draw [style=intDashedBlue] (188.center) to (189.center);
		\draw [style=intThickDashedBlue] (174.center) to (175.center);
		\draw [style=intThickDashedBlue] (176.center) to (177.center);
		\draw [style=intThickDashedBlue] (184.center) to (185.center);
		\draw [style=intThickDashedBlue] (186.center) to (187.center);
		\draw [style=intThickDashedBlue] (188.center) to (189.center);
		\draw [style=interval] (178.center) to (179.center);
		\draw [style=interval] (180.center) to (181.center);
		\draw [style=interval] (182.center) to (183.center);
		\draw [style=intThick] (178.center) to (179.center);
		\draw [style=intThick] (180.center) to (181.center);
		\draw [style=intThick] (182.center) to (183.center);
		\draw [style=areaGreen, in=-90, out=-180, looseness=1.25] (211.center) to (206.center);
		\draw [style=areaGreen, in=180, out=90, looseness=1.25] (206.center) to (207.center);
		\draw [style=areaGreen] (207.center) to (208.center);
		\draw [style=areaGreen, in=90, out=0, looseness=1.25] (208.center) to (209.center);
		\draw [style=areaGreen, in=0, out=-90, looseness=1.25] (209.center) to (210.center);
		\draw [style=areaGreen] (210.center) to (211.center);
		\draw [style=areaGreen, in=-90, out=-180, looseness=1.25] (217.center) to (212.center);
		\draw [style=areaGreen, in=180, out=90, looseness=1.25] (212.center) to (213.center);
		\draw [style=areaGreen] (213.center) to (214.center);
		\draw [style=areaGreen, in=90, out=0, looseness=1.25] (214.center) to (215.center);
		\draw [style=areaGreen, in=0, out=-90, looseness=1.25] (215.center) to (216.center);
		\draw [style=areaGreen] (216.center) to (217.center);
		\draw [style=areaGreen, in=-90, out=180, looseness=1.25] (223.center) to (218.center);
		\draw [style=areaGreen, in=-180, out=90, looseness=1.25] (218.center) to (219.center);
		\draw [style=areaGreen] (219.center) to (220.center);
		\draw [style=areaGreen, in=90, out=0, looseness=1.25] (220.center) to (221.center);
		\draw [style=areaGreen, in=0, out=-90, looseness=1.25] (221.center) to (222.center);
		\draw [style=areaGreen] (222.center) to (223.center);
		\draw [style=areaGreen, in=-90, out=-180, looseness=1.25] (229.center) to (224.center);
		\draw [style=areaGreen, in=-180, out=90, looseness=1.25] (224.center) to (225.center);
		\draw [style=areaGreen] (225.center) to (226.center);
		\draw [style=areaGreen, in=90, out=0, looseness=1.25] (226.center) to (227.center);
		\draw [style=areaGreen, in=0, out=-90, looseness=1.25] (227.center) to (228.center);
		\draw [style=areaGreen] (228.center) to (229.center);
		\draw [style=areaGreen, in=-90, out=-180, looseness=1.25] (235.center) to (230.center);
		\draw [style=areaGreen, in=-180, out=90, looseness=1.25] (230.center) to (231.center);
		\draw [style=areaGreen] (231.center) to (232.center);
		\draw [style=areaGreen, in=90, out=0, looseness=1.25] (232.center) to (233.center);
		\draw [style=areaGreen, in=0, out=-90, looseness=1.25] (233.center) to (234.center);
		\draw [style=areaGreen] (234.center) to (235.center);
		\draw [style=areaGreen, in=-90, out=-180, looseness=1.25] (241.center) to (236.center);
		\draw [style=areaGreen, in=-180, out=90, looseness=1.25] (236.center) to (237.center);
		\draw [style=areaGreen] (237.center) to (238.center);
		\draw [style=areaGreen, in=90, out=0, looseness=1.25] (238.center) to (239.center);
		\draw [style=areaGreen, in=0, out=-90, looseness=1.25] (239.center) to (240.center);
		\draw [style=areaGreen] (240.center) to (241.center);
		\draw [style=areaGreen] (261.center) to (262.center);
		\draw [style=areaGreen, in=-90, out=-180, looseness=1.25] (262.center) to (257.center);
		\draw [style=areaGreen, in=180, out=90, looseness=1.25] (257.center) to (258.center);
		\draw [style=areaGreen] (258.center) to (259.center);
		\draw [style=areaGreen, in=90, out=0, looseness=1.25] (259.center) to (260.center);
		\draw [style=areaGreen, in=0, out=-90, looseness=1.25] (260.center) to (261.center);
		\draw [style=intBrace] (265.center) to (264.center);
		\draw [style=intBrace] (266.center) to (267.center);
		\draw [style=intBrace] (269.center) to (270.center);
		\draw [style=intBrace] (272.center) to (273.center);
		\draw [style=intBrace] (274.center) to (275.center);
		\draw [style=intBrace] (276.center) to (277.center);
		\draw [style=intBrace] (278.center) to (279.center);
		\draw [style=intBrace] (280.center) to (281.center);
		\draw [style=intBrace] (282.center) to (283.center);
		\draw [style=intBrace] (284.center) to (285.center);
		\draw [style=intBrace] (286.center) to (287.center);
		\draw [style=intBrace] (288.center) to (289.center);
	\end{pgfonlayer}
\end{tikzpicture}
 
\end{small}
\caption{An example rank assignment satisfying the invariants : {\protect\tikzInterval} -- intervals of $\Intervals_{j}$, {\protect\tikzIntDashedBlue} -- intervals of $\Intervals_{n'} \setminus \Intervals_{j}$ (future intervals), {\protect\tikzIntThick} -- spans of $\Spanset_{j}$, {\protect\tikzIntThickDashedBlue} -- spans of $\Spanset_{n'} \setminus \Spanset_{j}$ (future spans), {\protect\tikzAreaGreen} -- blocks of partition $\PartB_j$.}
\label{fig:6}
\end{figure}
Let us assume now that the invariants are satisfied for step $j-1$ and turn $\ti(j-1)$, and that $j \geq 1$.
We first show how we charge for the costly recolorings of the $\Recolor$ algorithm in step $j$. Recall that we need to charge for $\min ( |\infix{\Intervals_j}{\interval_j}{\intervalTwo^{(R)}}|, |\infix{\Intervals_j}{\intervalTwo^{(L)}}{\interval_j}| )$. Recall also, that there  is a span $\Spanel^{(j)} \in \Spanset_{n'}$ such that $\Spanel^{(j)} \subseteq \interval_j$ and there is a passive singleton block $\Block^{j-1}_{s(j)} \in \PartB_{j-1}$ containing $\Spanel^{(j)}$. Since each passive block, by Invariant 5, is assigned more than one rank, there are two consecutive ranks $\rrank{\ti(j-1)}{s}$ and $\rrank{\ti(j-1)}{s+1}$ assigned to $\Block^{j-1}_{s(j)} \in \PartB_{j-1}$. We use one jump to merge $\rrank{\ti(j-1)}{s}$ and $\rrank{\ti(j-1)}{s+1}$. In other words, we set the $\ti(j-1)$'th jump in the \Frog game to position $s$, i.e., $j_{\ti(j-1)}\eqdef s$. We modify the rank assignment $f$ as to assign the rank $\rrank{\ti(j-1)}{s}+\rrank{\ti(j-1)}{s+1}$ to block $\Block^{j-1}_{s(j)}$. The Frogs game incurs the cost of $\cost_{\ti(j-1)} = \min\bra{\rrank{\ti(j-1)}{s-\jn+1}+\ldots+\rrank{\ti(j-1)}{s},\ \rrank{\ti(j-1)}{s+1}+\ldots+\rrank{\ti(j-1)}{s+\jn} }$. By Observation~\ref{obs:algToBlR}, Observation~\ref{obs:algToBlL}  and the invariants we get for $\gamma=\fO{k^3}$ and $\jn=(k+1)\gamma$: $$\begin{cases}|\infix{\Intervals_j}{\interval_j}{\intervalTwo^{(R)}}| \leq 1+ \sum_{i=0}^{\gamma} \size{ \intervals{j-1}{ \Block^{j-1}_{s(j)+i} } } \leq 1+ \rrank{\ti(j-1)}{s+1}+\ldots+\rrank{\ti(j-1)}{s+\jn}\\
|\infix{\Intervals_j}{\intervalTwo^{(L)}}{\interval_j}| \leq 1+ \sum_{i=0}^{\gamma} \size{\intervals{j-1}{ \Block^{j-1}_{s(j)-\gamma+i}} } \leq 1+ \rrank{\ti(j-1)}{s-\jn+1}+\ldots+\rrank{\ti(j-1)}{s}\end{cases}$$
Thus, $\min ( |\infix{\Intervals_j}{\interval_j}{\intervalTwo^{(R)}}|, |\infix{\Intervals_j}{\intervalTwo^{(L)}}{\interval_j}| ) \leq \cost_{\ti(j-1)}+1$. By Theorem~\ref{thm:main} the total cost incurred by the \Frog game is $\delta \bra{2 \jn-1}\bra{(k+1)\sizesp+2 \jn-2}\log_2\frac{(k+1)\sizesp+2 \jn -2}{2 \jn-1}$, where $\delta~=~k$, $\sizesp \in \fO{kn}$ and $\jn \in \fO{k^4}$. Thus, the total cost incurred by the Frogs game is $\fO{k^7 n \log n}$, and it is also an upperbound on the total number of recolorings applied by the \Recolor algorithm.

It remains to show how the invariants are maintained. Unfortunatelly, due to space limitations we are forced to defer the details to Appendix~\ref{app:invariants}. The main idea there is the following. Upon arrival of $\interval_j$ a set $\widehat{\Spanset}_j$ of consecutive spans that were passive in step $j-1$  becomes active in step $j$. Hence, in partition $\PartB_j$, $\widehat{\Spanset}_j$ is contained in one active block. Thus, we need to merge the ranks of the corresponding passive blocks into one rank (by adding them). Then, if the block $\Block$ of $\PartB_{j-1}$ preceeding $\widehat{\Spanset}_j$ was active, then $\Block$ and $\widehat{\Spanset}_j$ merge into one block in step $j$, and so we need to merge their ranks into one. Similar thing happens if the block of $\PartB_{j-1}$ succeding $\widehat{\Spanset}_j$ was active. Note that before the merge the ranks bound the number of intervals intersecting the blocks. It is not hard to see that after the merge the sum of the ranks bounds the number of intervals of the newly merged block. It turns out that the invariants give the room to account for the newly added interval $\interval_j$ as well.
\end{proof}
 
\section{Set of unit circular-arcs}
Now we define unit circular-arc on the circle with a fixed circumference $\circumference \in \RealNumbers$ with $\circumference > 1$.
The \emph{point on the circle with a circumference $\circumference$} is an equivalence class $\eqClass{\pointOne}$ (where $\pointOne\in\RealNumbers$) for an equivalence relation $\eqArc$ defined on the real numbers in this way that $\pointTwo \eqArc \pointThree$ iff there is $k\in\IntegerNumbers$ such that $\pointTwo-\pointThree = \circumference k$ for any $\pointTwo, \pointThree \in \RealNumbers$.
As a consequence, $\eqClass{\pointOne} = \set{\pointOne+k\circumference : k\in\IntegerNumbers}$.
The set of points on the circle with a circumference $\circumference$ we describe as $\CirclePoints$.
A \emph{circular-arc (half-open)} on the circle with a circumference $\circumference$ is
$$
\arcDef{\eqClass{\pointOne}}{\eqClass{\pointTwo}} \eqdef \set{\eqClass{\pointThree} \in \CirclePoints : \pointOne \leqslant \pointThree < \pointTwo' \ \text{for any} \ \pointTwo'\in\eqClass{\pointTwo} \ \text{such that} \ \pointOne < \pointTwo' }.
$$
An \emph{unit circular-arc (half-open)} on the circle with a circumference $\circumference$ is
$$
\arcDef{\eqClass{\pointOne}}{\eqClass{\pointOne}+1} \eqdef \arcDef{\eqClass{\pointOne}}{\eqClass{\pointOne+1}} = \set{\eqClass{\pointTwo} \in \CirclePoints : \pointOne \leqslant \pointTwo < \pointOne+1 }.
$$
A \emph{set of unit circular-arcs on the circle with a circumference $\circumference>1$} is $\Arcs \eqdef \set{\arc_1,\arc_2,\ldots,\arc_n}$, where $\arc_i\eqdef\arcDef{\alpha_i}{\alpha_i+1}$ and $\alpha_i\in \RealNumbers_{/\circumference\IntegerNumbers}$ for each $i\in\num{n}$.
A graph $\Graph$ is an \emph{unit circular-arc} iff there is a set of unit circular-arcs $\Arcs$ and bijection $\arcRepr:\VerticesOf{\Graph} \rightarrow \Arcs$ such that $\vertexOne \sim \vertexTwo$ iff $\fArcRepr{\vertexOne} \cap \fArcRepr{\vertexTwo} = \emptyset$ for any $\vertexOne,\vertexTwo \in \VerticesOf{\Graph}$.
Then we say that the mapping $\arcRepr$ is an \emph{unit circular-arc representation} of $\Graph$.
A set of unit circular-arcs $\ArcsTwo=\set{\arcTwo_1,\arcTwo_2,\ldots,\arcTwo_k}$ is an \emph{overlapping set} iff $$\bigcap\ArcsTwo \eqdef \arcTwo_1 \cap \arcTwo_2 \cap \ldots \cap \arcTwo_k \neq \emptyset.$$
If $\circumference \geqslant 2$ and $\arcTwo_i=\intDef{\circlePointTwo_i}{\circlePointTwo_i+1}$ for $i\in\num{k}$ then the intersection $\bigcap_{i=1}^{k} \intDef{\circlePointTwo_i}{\circlePointTwo_i+1} = \intDef{\circlePointThree}{\circlePointFour}$.
If circumference $\circumference < 3$ then there is a unit circular-arc graph $\Graph$ and a unit circular-arc representation $\arcRepr$ of $\Graph$ such that $\VerticesOf{\Graph}$ is a clique in $\Graph$ and $\fArcRepr{\VerticesTwo}$ is not an overlapping set.
\begin{figure}[hbpt]
\begin{center}

\begin{tikzpicture}[xscale=15, yscale=15]

\def\x{0}
\def\y{0}

\def\r{4}
\def\from{-60}
\def\to{120}

\def\rr{5}
\def\fromm{60}
\def\too{240}

\def\rrr{6}
\def\frommm{181}
\def\tooo{359}

\coordinate (c1) at ({ - \x - \r * cos(\from + 180)},{ - \y  - \r * sin(\from + 180)});
\coordinate (c2) at ({ - \x - \rr * cos(\fromm + 180)},{ - \y  - \rr * sin(\fromm + 180)});
\coordinate (c3) at ({ - \x - \rrr * cos(\frommm + 180)},{ - \y  - \rrr * sin(\frommm + 180)});

\draw[style=interval] (c1) arc (\from:\to:\r);
\draw[style=interval] (c2) arc (\fromm:\too:\rr);
\draw[style=interval] (c3) arc (\frommm:\tooo:\rrr);

\end{tikzpicture}
 
\caption{The set of unit circular-arcs which is not an overlap set but forms clique $K_3$.}
\end{center}
\end{figure}
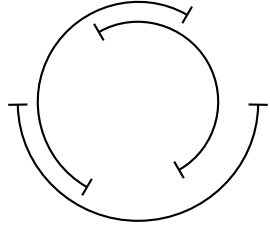
If circumference $\circumference \geqslant 3$ then for a unit circular-arc graph $\Graph$ and a unit circular-arc representation $\arcRepr$ of $\Graph$ the subset of vertices $\VerticesTwo \subseteq \VerticesOf{\Graph}$ is a clique in $\Graph$ iff $\fArcRepr{\VerticesTwo}$ is an overlapping set.
A \emph{load} of a set of unit circular-arcs $\Arcs$ in a point $\circlePointTwo \in \CirclePoints$ is the number of intervals from $\Arcs$ which contains $\circlePointTwo$ and it is denoted as $\fLoad{\Arcs}{\circlePointTwo}$, i.e., $\fLoad{\Arcs}{\circlePointTwo} \eqdef \size{\set{\arcOne\in\Arcs:\circlePointTwo\in\arcOne}}$.
A \emph{maximum load} of a set of unit circular-arcs $\Arcs$ is maximum load of $\Arcs$ in a point $\circlePointTwo \in \CirclePoints$ over all $\circlePointTwo \in \CirclePoints$ and it is denoted as $\fMaxLoad{\Arcs}$, i.e., $\fMaxLoad{\Arcs} \eqdef \max_{\circlePointTwo \in \CirclePoints} \fLoad{\Arcs}{\circlePointTwo}$.
If circumference $\circumference \geqslant 3$ then for a unit circular-arc graph $\Graph=\tuple{\Vertices,\Edges}$ and a unit circular-arc representation $\arcRepr:\Vertices\rightarrow\Arcs$ of $\Graph$, the size of the maximum overlapping set equals $\fMaxLoad{\Arcs} = \fClique{\Graph}$.
A function $\coloring:\Arcs\rightarrow\Gamma$, where $\Gamma$ is an arbitrarily set of \emph{colors} of size $\size{\Gamma}=l$, is a~\emph{(proper) \mbox{$l$-coloring}} of a set of unit circular-arcs $\Arcs$ iff for any two intervals $\arc_i,\arc_j\in\Arcs$ if $\arc_i\cap\arc_j\neq\emptyset$ we have that $\fColoring{\arc_i}\neq\fColoring{\arc_j}$.
For a unit circular-arc graph $\Graph=\tuple{\Vertices,\Edges}$ and a unit circular-arc representation $\arcRepr:\Vertices\rightarrow\Arcs$ of $\Graph$, a mapping $\Arcs\ni\arc_i \mapsto \fColoring{\arc_i}\in\Gamma$ is a proper coloring of the set of unit circular-arcs $\Arcs$ iff a mapping $\Vertices\ni\vertex_i \mapsto \fColoring{\fArcRepr{\vertex_i}}\in\Gamma$ is a proper coloring of the graph $\Graph$.
The \emph{chromatic number} of the set of unit circular-arcs $\Arcs$ is the smallest number of colors needed to proper color of $\Arcs$ and is is denoted as $\fChromatic{\Arcs}$, i.e.,
$$
\fChromatic{\Arcs}\eqdef\min\set{l\in\NaturalNumbers:\text{exists an \mbox{$l$-coloring} of} \ \Arcs}.
$$
For a unit circular-arc graph $\Graph=\tuple{\Vertices,\Edges}$ and a unit circular-arc representation $\arcRepr:\Vertices\rightarrow\Arcs$ of $\Graph$, we have that $\fChromatic{\Graph}=\fChromatic{\Arcs}$.
For each circumference $\circumference>1$ there is a set of unit circular-arcs $\Arcs$ such that $\fMaxLoad{\Arcs} < \fChromatic{\Arcs}$
\begin{figure}[hbpt]
\begin{center}

\begin{tikzpicture}[xscale=15, yscale=15]

\def\x{0}
\def\y{0}

\def\r{4}
\def\from{0}
\def\to{108}

\def\rr{5}
\def\fromm{72}
\def\too{180}

\def\rrr{4}
\def\frommm{144}
\def\tooo{252}

\def\rrrr{5}
\def\frommmm{216}
\def\toooo{324}

\def\rrrrr{6}
\def\frommmmm{288}
\def\tooooo{396}

\coordinate (c1) at ({ - \x - \r * cos(\from + 180)},{ - \y  - \r * sin(\from + 180)});
\coordinate (c2) at ({ - \x - \rr * cos(\fromm + 180)},{ - \y  - \rr * sin(\fromm + 180)});
\coordinate (c3) at ({ - \x - \rrr * cos(\frommm + 180)},{ - \y  - \rrr * sin(\frommm + 180)});
\coordinate (c4) at ({ - \x - \rrrr * cos(\frommmm + 180)},{ - \y  - \rrrr * sin(\frommmm + 180)});
\coordinate (c5) at ({ - \x - \rrrrr * cos(\frommmmm + 180)},{ - \y  - \rrrrr * sin(\frommmmm + 180)});

\draw[style=interval] (c1) arc (\from:\to:\r);
\draw[style=interval] (c2) arc (\fromm:\too:\rr);
\draw[style=interval] (c3) arc (\frommm:\tooo:\rrr);
\draw[style=interval] (c4) arc (\frommmm:\toooo:\rrrr);
\draw[style=interval] (c5) arc (\frommmmm:\tooooo:\rrrrr);

\end{tikzpicture}
 
\caption{An unit circular-arc representation of $C_5$.} 
\end{center}
\end{figure}
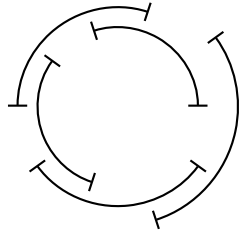
For a unit circular-arc graph $\Graph=\tuple{\Vertices,\Edges}$ and a unit circular-arc representation $\arcRepr:\Vertices\rightarrow\Arcs$ of $\Graph$, we have that
$
\fMaxLoad{\Arcs} \leqslant \fClique{\Graph} \leqslant \fChromatic{\Graph} = \fChromatic{\Arcs},
$
and if $\circumference>3$ then $\fMaxLoad{\Arcs} = \fClique{\Graph}$.
Belkale and Chandran \cite{BCh09} proved that unit circular-arc graphs satisfy Hadwiger's conjecture, i.e., if a unit circular-arc graph $\Graph$ not contains a clique $K_k$ as a minor then $\fChromatic{\Graph}\leqslant k-1$.
Tucker Theorem \cite{Tuc75} implies the following statement
\begin{theorem}[Tucker 1975]
For a set of unit circular-arcs $\Arcs$ on the circle with a circumference $\circumference>3$, there is $\fChromatic{\Arcs} \leqslant \frac{3}{2}\fMaxLoad{\Arcs}$.
\end{theorem}
Later, Valencia-Pabon \cite{Val03} strengthened the above theorem and in the discussed context, this made the following theorem.
\begin{theorem}[Valencia-Pabon 2003]
For a set of unit circular-arcs $\Arcs$ on the circle with a circumference $\circumference>4$, there is that $\fChromatic{\Arcs} \leqslant \ceil{\frac{\circumference}{\circumference-1}\fMaxLoad{\Arcs}}$.
\end{theorem}
This, in turn, for a sufficiently large circumference gives an estimate of the chromatic number by the clique number with an accuracy of $1$.
\begin{corollary}
\label{cor:Valencia-Pabon}
For a set of unit circular-arcs $\Arcs$ on the circle with a circumference $\circumference \geqslant \fMaxLoad{\Arcs}+1$, we have that $\fChromatic{\Arcs} \leqslant \fMaxLoad{\Arcs}+1$.
\end{corollary}
We are going further.
We show that if the circuit is even larger, the constraint is accurate as long as there is an appropriate number of places in which the load is one smaller than max-load.
The following theorem shows this, and it is a simple consequence of Lemma \ref{lem:partColor}.
\begin{theorem}
\label{thm:main-circ}
Let $\Arcs$ be a set of unit circular-arcs on the circle with a circumference $\circumference \geqslant 2$ such that $\Arcs$ can be extended with $r=\fMaxLoad{\Arcs}^2-1$ not intersecting unit circular-arcs $\arcTwo_1, \ldots,\arcTwo_r$ which does not increase a max load, i.e., $\fMaxLoad{\Arcs}=\fMaxLoad{\Arcs \cup \set{\arcTwo_1, \ldots,\arcTwo_r}}$. Then $\fChromatic{\Arcs} \leqslant \fMaxLoad{\Arcs}$.
\end{theorem}
\begin{proof}[Idea of the proof]
Without losing generality, we can assume that the unit circular-arcs containing the point $\eqClass{0}$ form the maximum overlapping set in $\Arcs$.
It means that we can represent the set $\Arcs = \set{\arcDef{\eqClass{\point_1}}{\eqClass{\point_1+1}},\ldots,\arcDef{\eqClass{\point_n}}{\eqClass{\point_n+1}}}$ with the order at starting points:
$$
-1 < \point_1 < \point_2 < \ldots < \point_l \leqslant 0 < \point_{l+1} < \point_{l+2} < \ldots < \point_n \leqslant \circumference-1
$$
where $l\eqdef \fMaxLoad{\Arcs}$ and $n \eqdef \size{\Arcs}$.
Now we define the set of $n+l$ unit intervals 
\begin{multline*}
\Intervals\eqdef\{\intDef{\point_1}{\point_1+1}, \intDef{\point_2}{\point_2+1}, \ldots, \intDef{\point_n}{\point_n+1},
\\
\intDef{\point_{1}+\circumference}{\point_{1}+\circumference+1},\intDef{\point_{2}+\circumference}{\point_{2}+\circumference+1},\ldots\intDef{\point_{l}+\circumference}{\point_{l}+\circumference+1}\}
\end{multline*}
and coloring $\coloringPrime$ of intervals $\intDef{\point_1}{\point_1+1}, \ldots, \intDef{\point_l}{\point_l+1}$
and coloring $\coloringDouble$ of intervals $\intDef{\point_1+\circumference}{\point_1+\circumference+1}, \ldots, \intDef{\point_l+\circumference}{\point_l+\circumference+1}$ such that for each $i=1,2,\ldots,l$ there is
$
\fColoringPrime{\intDef{\point_i}{\point_i+1}}\eqdef\fColoringDouble{\intDef{\point_i+\circumference}{\point_i+\circumference+1}}\eqdef i.
$
By the assumptions of the theorem and Lemma~\ref{lem:partColor} we obtain the proper coloring $\coloring$ of the whole $\Intervals$ by colors from $\set{1,2,\ldots,l}$ which is the extension of the colorings $\coloringPrime \cup \coloringDouble$, i.e.,
$
\fColoring{\intDef{\point_i}{\point_i+1}}=\fColoring{\intDef{\point_i+\circumference}{\point_i+\circumference+1}}= i
$
for each $i=1,2,\ldots,l$.
The final coloring $\coloringArc:\Arcs\rightarrow\set{1,2,\ldots,l}$ is obtained by appropriately copying the coloring $\coloring$ in the following way
$
\fColoringArc{\arcDef{\eqClass{\point_i}}{\eqClass{\point_i+1}}}\eqdef\fColoring{\intDef{\point_i}{\point_i+1}}
$
for each $i=1,2,\ldots,n$.
\end{proof}
In order to understand where the limits of the above approach are, it is important to understand what the Theorem \ref{thm:main-circ} says in the language of Corollary \ref{cor:Valencia-Pabon}.
We write it down in the following note.
\begin{corollary}
For a set of unit circular-arcs $\Arcs$ on the circle with a circumference $\circumference \in \fOmega{\fMaxLoad{\Arcs}^2}$, we have that $\fChromatic{\Arcs} \leqslant \fMaxLoad{\Arcs}+1$.
\end{corollary}
This ultimately suggests the following reinforcement of Theorem \ref{thm:main-circ}.
\begin{conjecture}
A set of unit circular-arcs $\Arcs$ on the circle with a circumference $\circumference \geqslant 2$ such that adding $k=\fMaxLoad{\Arcs}+1$ not intersecting unit circular-arcs $\arcTwo_1, \ldots,\arcTwo_k$ no increasing a max load of $\Arcs$, i.e., $\fMaxLoad{\Arcs}=\fMaxLoad{\Arcs \cup \set{\arcTwo_1, \ldots,\arcTwo_k}}$. Then $\fChromatic{\Arcs} \leqslant \fMaxLoad{\Arcs}$.
\end{conjecture}

\appendix
\section{Lower bound for the \PruirFD problem}\label{app:lb}
\begin{proof}[Proof of Observation~\ref{obs:FD}]
We first construct an instance $\IntervalsFive=\Intervals \cup \IntervalsTwo \cup \IntervalsThree \cup \IntervalsFour$, where $\Intervals=\set{\intervalOne_1,\ldots,\intervalOne_n}$, $\IntervalsTwo=\set{\intervalTwo_1,\ldots,\intervalTwo_n}$, $\IntervalsThree=\set{\intervalThree_1,\ldots,\intervalThree_n}$ and $\IntervalsFour=\set{\intervalFour_1,\ldots,\intervalFour_n}$. All four sets contain pairwise disjoint intervals. Both $\IntervalsThree \cup \IntervalsFour$ and $\Intervals \cup \IntervalsTwo$ are paths when interpreted as graphs. Additionally, set $\IntervalsThree \cup \IntervalsFour$ is placed to the left of $\Intervals \cup \IntervalsTwo$. This is shown in Figure~\ref{fig:fd12}.
\begin{figure}[hbpt!]
\begin{small}

\centering

\begin{tikzpicture}[xscale={2\textwidth/34.5}, yscale=18]
	\begin{pgfonlayer}{nodelayer}
		\node [style=none] (0) at (3, 2) {};
		\node [style=none] (1) at (6.5, 2) {};
		\node [style=none] (2) at (6.75, 2) {};
		\node [style=none] (3) at (10.25, 2) {};
		\node [style=none] (4) at (10.5, 2) {};
		\node [style=none] (5) at (14, 2) {};
		\node [style=none] (6) at (5.5, 0) {};
		\node [style=none] (7) at (9, 0) {};
		\node [style=none] (8) at (9.25, 0) {};
		\node [style=none] (9) at (12.75, 0) {};
		\node [style=none] (10) at (13, 0) {};
		\node [style=none] (11) at (16.5, 0) {};
		\node [style=none] (12) at (14.25, 6) {};
		\node [style=none] (13) at (17.75, 6) {};
		\node [style=none] (14) at (16.75, 4) {};
		\node [style=none] (15) at (20.25, 4) {};
		\node [style=none] (16) at (18, 0) {};
		\node [style=none] (17) at (21.5, 0) {};
		\node [style=none] (18) at (20.5, 2) {};
		\node [style=none] (19) at (24, 2) {};
		\node [style=none] (20) at (21.75, 0) {};
		\node [style=none] (21) at (25.25, 0) {};
		\node [style=none] (22) at (25.5, 0) {};
		\node [style=none] (23) at (29, 0) {};
		\node [style=none] (24) at (24.25, 2) {};
		\node [style=none] (25) at (27.75, 2) {};
		\node [style=none] (26) at (28, 2) {};
		\node [style=none] (27) at (31.5, 2) {};
		\node [style=none] (28) at (3.75, 2.75) {$\intervalFour_3$};
		\node [style=none] (29) at (7.5, 2.75) {$\intervalFour_2$};
		\node [style=none] (30) at (11.25, 2.75) {$\intervalFour_1$};
		\node [style=none] (31) at (6.25, 0.75) {$\intervalThree_3$};
		\node [style=none] (32) at (10, 0.75) {$\intervalThree_2$};
		\node [style=none] (33) at (13.75, 0.75) {$\intervalThree_1$};
		\node [style=none] (34) at (21.25, 2.75) {$\intervalTwo_1$};
		\node [style=none] (35) at (25, 2.75) {$\intervalTwo_2$};
		\node [style=none] (36) at (28.75, 2.75) {$\intervalTwo_3$};
		\node [style=none] (37) at (18.75, 0.75) {$\intervalOne_1$};
		\node [style=none] (38) at (22.5, 0.75) {$\intervalOne_2$};
		\node [style=none] (39) at (26.25, 0.75) {$\intervalOne_3$};
		\node [style=none] (40) at (17.375, 4.75) {$\intervalPrim$};
		\node [style=none] (41) at (14.75, 6.75) {$\interval$};
		\node [style=none] (42) at (1, 1) {\begin{huge}$\cdots$\end{huge}};
		\node [style=none] (43) at (33.5, 1) {\begin{huge}$\cdots$\end{huge}};
		\node [style=none] (44) at (0, 2) {};
		\node [style=none] (45) at (0, 0) {};
		\node [style=none] (46) at (34.5, 2) {};
		\node [style=none] (47) at (34.5, 0) {};
	\end{pgfonlayer}
	\begin{pgfonlayer}{edgelayer}
		\draw [style=intGreen] (0.center) to (1.center);
		\draw [style=intGreen] (2.center) to (3.center);
		\draw [style=intGreen] (4.center) to (5.center);
		\draw [style=intRed] (6.center) to (7.center);
		\draw [style=intRed] (8.center) to (9.center);
		\draw [style=intRed] (10.center) to (11.center);
		\draw [style=interval] (12.center) to (13.center);
		\draw [style=interval] (14.center) to (15.center);
		\draw [style=intRed] (16.center) to (17.center);
		\draw [style=intGreen] (18.center) to (19.center);
		\draw [style=intRed] (20.center) to (21.center);
		\draw [style=intRed] (22.center) to (23.center);
		\draw [style=intGreen] (24.center) to (25.center);
		\draw [style=intGreen] (26.center) to (27.center);
	\end{pgfonlayer}
\end{tikzpicture}
 
\end{small}

\caption{The update when $\intervalThree_1$ has the same color as $\intervalOne_1$.}
\label{fig:fd12}
\end{figure}
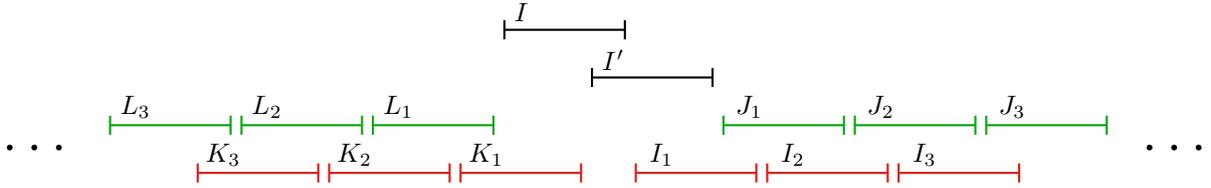
Observe that since $\Intervals$ is $2$-colorable, all sets $\IntervalsThree$, $ \IntervalsFour$, $\Intervals$ and $\IntervalsTwo$ are necessarily monochromatic, moreover set $\IntervalsThree$ has different color than $\IntervalsFour$, and $\Intervals$ has different color than $\IntervalsTwo$. We construct the instance so, that $\intervalThree_1$ is the rightmost interval of $\IntervalsThree \cup \IntervalsFour$ and interval $\intervalOne_1$ is the leftmost interval of $\Intervals \cup \IntervalsTwo$. Moreover, the distance between the end coordinate of $\intervalThree_1$ and the begin coordinate of $\intervalOne_1$ is less then one (see Figure~\ref{fig:fd12}). Consider the case when $\intervalThree_1$ has the same color as $\intervalOne_1$. Then we insert to $\IntervalsFive$ two intersecting intervals $\interval$ and $\intervalPrim$, such that $\interval$ intersects $\intervalThree_1$ and $\intervalPrim$ intersects $\intervalOne_1$ (see Figure~\ref{fig:fd12}). This results in $\IntervalsFive \cup \set{ \interval,\intervalPrim}$ being a path when interpreted as a graph. The instance is still $2$-colorable, but now $\Intervals$ needs to have the same color as $\IntervalsFour$ and $\IntervalsTwo$ needs to have the same color as $\IntervalsThree$. This requires recoloring $2n$ intervals, since either $\Intervals \cup \IntervalsTwo$ or $\IntervalsThree \cup \IntervalsFour$ has to be recolored. Next, we remove $\interval$ and $\intervalPrim$. Consider now the case when $\intervalThree_1$ has a different color than $\intervalOne_1$ (see Figure~\ref{fig:fd34}).
\begin{figure}[hbpt!]
\begin{small}

\centering

\begin{tikzpicture}[xscale={2\textwidth/34.5}, yscale=18]
	\begin{pgfonlayer}{nodelayer}
		\node [style=none] (0) at (3, 2) {};
		\node [style=none] (1) at (6.5, 2) {};
		\node [style=none] (2) at (6.75, 2) {};
		\node [style=none] (3) at (10.25, 2) {};
		\node [style=none] (4) at (10.5, 2) {};
		\node [style=none] (5) at (14, 2) {};
		\node [style=none] (6) at (5.5, 0) {};
		\node [style=none] (7) at (9, 0) {};
		\node [style=none] (8) at (9.25, 0) {};
		\node [style=none] (9) at (12.75, 0) {};
		\node [style=none] (10) at (13, 0) {};
		\node [style=none] (11) at (16.5, 0) {};
		\node [style=none] (14) at (15.5, 4) {};
		\node [style=none] (15) at (19, 4) {};
		\node [style=none] (16) at (18, 0) {};
		\node [style=none] (17) at (21.5, 0) {};
		\node [style=none] (18) at (20.5, 2) {};
		\node [style=none] (19) at (24, 2) {};
		\node [style=none] (20) at (21.75, 0) {};
		\node [style=none] (21) at (25.25, 0) {};
		\node [style=none] (22) at (25.5, 0) {};
		\node [style=none] (23) at (29, 0) {};
		\node [style=none] (24) at (24.25, 2) {};
		\node [style=none] (25) at (27.75, 2) {};
		\node [style=none] (26) at (28, 2) {};
		\node [style=none] (27) at (31.5, 2) {};
		\node [style=none] (28) at (3.75, 2.75) {$\intervalFour_3$};
		\node [style=none] (29) at (7.5, 2.75) {$\intervalFour_2$};
		\node [style=none] (30) at (11.25, 2.75) {$\intervalFour_1$};
		\node [style=none] (31) at (6.25, 0.75) {$\intervalThree_3$};
		\node [style=none] (32) at (10, 0.75) {$\intervalThree_2$};
		\node [style=none] (33) at (13.75, 0.75) {$\intervalThree_1$};
		\node [style=none] (34) at (21.25, 2.75) {$\intervalTwo_1$};
		\node [style=none] (35) at (25, 2.75) {$\intervalTwo_2$};
		\node [style=none] (36) at (28.75, 2.75) {$\intervalTwo_3$};
		\node [style=none] (37) at (18.75, 0.75) {$\intervalOne_1$};
		\node [style=none] (38) at (22.5, 0.75) {$\intervalOne_2$};
		\node [style=none] (39) at (26.25, 0.75) {$\intervalOne_3$};
		\node [style=none] (41) at (16, 4.75) {$\interval$};
		\node [style=none] (42) at (1, 1) {\begin{huge}$\cdots$\end{huge}};
		\node [style=none] (43) at (33.5, 1) {\begin{huge}$\cdots$\end{huge}};
		\node [style=none] (44) at (0, 2) {};
		\node [style=none] (45) at (0, 0) {};
		\node [style=none] (46) at (34.5, 2) {};
		\node [style=none] (47) at (34.5, 0) {};
	\end{pgfonlayer}
	\begin{pgfonlayer}{edgelayer}
		\draw [style=intGreen] (0.center) to (1.center);
		\draw [style=intGreen] (2.center) to (3.center);
		\draw [style=intGreen] (4.center) to (5.center);
		\draw [style=intRed] (6.center) to (7.center);
		\draw [style=intRed] (8.center) to (9.center);
		\draw [style=intRed] (10.center) to (11.center);
		\draw [style=interval] (14.center) to (15.center);
		\draw [style=intGreen] (16.center) to (17.center);
		\draw [style=intRed] (18.center) to (19.center);
		\draw [style=intGreen] (20.center) to (21.center);
		\draw [style=intGreen] (22.center) to (23.center);
		\draw [style=intRed] (24.center) to (25.center);
		\draw [style=intRed] (26.center) to (27.center);
	\end{pgfonlayer}
\end{tikzpicture}
 
\end{small}

\caption{The update when $\intervalThree_1$ has different color than $\intervalOne_1$.}
\label{fig:fd34}
\end{figure}
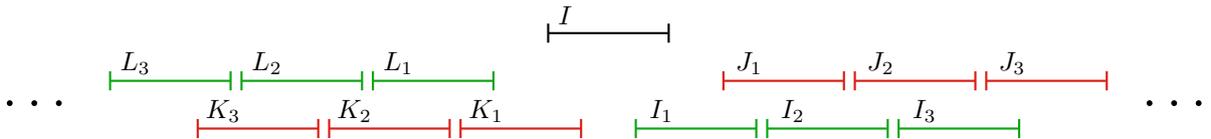
In that case we insert an interval $\interval$ intersecting both $\intervalThree_1$ and $\intervalOne_1$. This causes that $\intervalThree_1$ and $\intervalOne_1$ have to now be assigned different colors, and again either $\Intervals \cup \IntervalsTwo$ or $\IntervalsThree \cup \IntervalsFour$ has to be recolored. 
\end{proof}
 \section{The Frogs game - proof}\label{app:zaba}

\begin{proof}[Proof of Theorem \ref{thm:main}.]
For the ranks $\RRanks{\ti}$ we define the potential function
$$
\Phi\bra{\RRanks{\ti}}\eqdef\sum_{j=-2k+3}^{n-\ti+1}
H\bra{\rrank{\ti}{j}+\ldots+\rrank{\ti}{j+2 \jn-2}}
$$
where $H\bra{x}\eqdef x\log_2 x$ for $x>0$ and $H\bra{0}\eqdef 0$, and all $\rrank{\ti}{i}$ which do not occurs in $\RRanks{\ti}$ we define as $\rrank{\ti}{i}\eqdef \Rini$.\footnote{The assumption that $\rrank{\ti}{i}\eqdef 0$ probably implies the bound $c\leqslant \bra{2k-1}n\log_2\frac{n}{2k-1}$ but with more complicated proof.}
Because $\RRanks{1}=\bra{\Rini,\Rini,\ldots,\Rini}$ and $\RRanks{n}=\bra{\Rini n}$ we have that $\Phi\bra{\RRanks{1}} = \bra{n+2 \jn -2} \bra{2 \jn-1} \Rini\log_2\bra{\bra{2 \jn-1} \Rini}$ and $\Phi\bra{\RRanks{n}} = \bra{2 \jn-1}\bra{n+2 \jn-2} \Rini \log_2 \bra{\bra{n+2 \jn-2} \Rini}$, respectively.
If we prove that
\begin{equation}
\label{eq:potential}
\cost_{\ti} \leqslant \Phi\bra{\RRanks{\ti+1}}-\Phi\bra{\RRanks{\ti}}
\end{equation}
for $\ti=1,\ldots,n-1$, we will finally get that the total cost
\begin{multline*}
\cost=\sum_{\ti=1}^{n-1}\cost_{\ti} \leqslant \sum_{\ti=1}^{n-1}\Bra{\Phi\bra{\RRanks{\ti+1}}-\Phi\bra{\RRanks{\ti}}}=\Phi\bra{\RRanks{n}}-\Phi\bra{\RRanks{1}} = \\
 =\Rini \bra{2 \jn -1}\bra{n+2 \jn-2}\log_2\frac{n+2 \jn-2}{2 \jn-1}
\end{multline*}
which would finish the proof.

We split the proof of Inequality (\ref{eq:potential}) into the two parts.
First, let us fix integer $\ti=1,\ldots,n-1$ and two sequences of ranks $\RanksTwo=\bra{\rankTwo{1},\ldots,\rankTwo{n-\ti+1}}$ and $\RanksThree=\bra{\rankThree{1},\ldots,\rankThree{n-\ti}}$ defined as follows.
For any $i=1,\ldots,n-\ti+1$
$$
\rankTwo{i} \eqdef
\begin{cases}
\rrank{\ti}{i} & \text{for } i=j_{\ti}-\jn+1,\ldots,j_{\ti}+\jn,\\
0 & \text{otherwise;}
\end{cases}
$$
and for any $i=1,\ldots,n-\ti$
$$
\rankThree{i} \eqdef
\begin{cases}
\rrank{\ti+1}{i} & \text{for } i=j_{\ti}-\jn+1,\ldots,j_{\ti}+\jn-1,\\
0 & \text{otherwise.}
\end{cases}
$$
The following two Claims hold.

\begin{claim}
\label{clm:first}
$
\Phi\bra{\RanksThree}-\Phi\bra{\RanksTwo} \leqslant \Phi\bra{\RRanks{\ti+1}}-\Phi\bra{\RRanks{\ti}}.
$
\end{claim}

\begin{proof}[Proof of Claim \ref{clm:first}]
For $l\leqslant j_{\ti}-\jn+1$ let us define two sequences of ranks $\RRanksTwo{l}=\bra{\rrankTwo{l}{1},\ldots,\rrankTwo{l}{n-\ti+1}}$ and $\RRanksThree{l}=\bra{\rrankThree{l}{1},\ldots,\rrankThree{l}{n-\ti}}$ as follows.
For each $i=1,\ldots,n-\ti+1$
$$
\rrankTwo{l}{i} \eqdef
\begin{cases}
\rrank{\ti}{i} & \text{for } i\geqslant l,\\
0 & \text{otherwise;}
\end{cases}
$$
and for each $i=1,\ldots,n-\ti$
$$
\rrankThree{l}{i} \eqdef
\begin{cases}
\rrank{\ti+1}{i} & \text{for } i\geqslant l,\\
0 & \text{otherwise.}
\end{cases}
$$
We will prove in an inductive way that for $l\leqslant j_{\ti}-\jn+1$ we have that
\begin{equation}
\label{eq:ineq-ind}
\Phi\bra{\RRanksThree{l}}-\Phi\bra{\RRanksTwo{l}} \leqslant \Phi\bra{\RRanks{\ti+1}}-\Phi\bra{\RRanks{\ti}}.
\end{equation}
First, we note that for $j\leqslant j_{\ti}-2 \jn+1$ we have that
$$
H\bra{\rrank{\ti}{j}+\ldots+\rrank{\ti}{j+2 \jn-2}} = H\bra{\rrank{\ti+1}{j}+\ldots+\rrank{\ti+1}{j+2 \jn-2}}
$$
and
$$
H\bra{\rrankTwo{l}{j}+\ldots+\rrankTwo{l}{j+2 \jn-2}} = H\bra{\rrankThree{l}{j}+\ldots+\rrankThree{l}{j+2 \jn-2}}
$$
and as a consequence for $l\leqslant j_{\ti}-2 \jn+1$ we have that $\RRanksTwo{l}=\RRanksThree{l}$ and $\RRanks{\ti}=\RRanks{\ti+1}$ and so Inequality~(\ref{eq:ineq-ind}) trivially holds.
Let us assume that Inequality~(\ref{eq:ineq-ind}) holds for $j_{\ti}-2 \jn+1<l<j_{\ti}-\jn+1$ and we will prove it for $l+1$.
We divide the proof into the following three statements.

\begin{claim}
\label{clm:increasing}
For any real constant $c\geqslant 0$ the function $f\bra{x}=\bra{x+c}\log_2\bra{x+c}-x\log_2 x$ is non-decreasing.
\end{claim}

\begin{proof}[Proof of Claim \ref{clm:increasing}]
The derivative of the function $f(x)$ equals $f'\bra{x}=\log_2\bra{x+c}-\log_2 x$.
Because $c\geqslant0$ and the binary logarithm is the increasing function we have that $\log_2\bra{x+c}\geqslant\log_2 x$ for each $x>0$ and in consequence $f'\bra{x}=\log_2\bra{x+c}-\log_2 x\geqslant 0$ for any $x>0$.
This implies a desired thesis.
\end{proof}

\begin{claim}
\label{clm:subadd-H}
For the arbitrary chosen non-negative real numbers $a,b,c\geqslant 0$ there is $$H\bra{b+c}-H\bra{b}\leqslant H\bra{a+b+c}-H\bra{a+b}.$$
\end{claim}

\begin{proof}[Proof of Subclaim \ref{clm:subadd-H}]
From the fact that for any real constant $c\geqslant 0$ the function $f\bra{x}=\bra{x+c}\log_2\bra{x+c}-x\log_2 x$ is non-decreasing (see Subclaim \ref{clm:increasing}) and $b\leqslant a + b$ we obtain that
\begin{multline*}
H\bra{b+c}-H\bra{b} = \bra{b+c}\log_2\bra{b+c} - b \log_2 b = f\bra{b}\leqslant f\bra{a+b} =\\
= \bra{a+b+c}\log_2\bra{a+b+c} - \bra{a+b} \log_2 \bra{a+b} = H\bra{a+b+c}-H\bra{a+b}.
\end{multline*}
\end{proof}

\begin{claim}
\label{clm:ind-step}
$\Phi\bra{\RRanksThree{l+1}}-\Phi\bra{\RRanksTwo{l+1}} \leqslant\Phi\bra{\RRanksThree{l}}-\Phi\bra{\RRanksTwo{l}}$.
\end{claim}

\begin{proof}[Proof of Subclaim \ref{clm:ind-step}]
Because $H\bra{0}=0$ and $j_{\ti}-2 \jn+1<l<j_{\ti}-\jn+1$ we have that
\begin{multline}
\label{eq:Pl}
\Phi\bra{\RRanksTwo{l}}
=\sum_{j=-2 \jn+3}^{n-\ti}H\bra{\rrankTwo{l}{j}+\ldots+\rrankTwo{l}{j+2 \jn-2}}=\\
=\sum_{j=j_{\ti}-2 \jn+2}^{j_{\ti}}H\bra{\rrank{\ti}{l}+\ldots+\rrank{\ti}{j+2 \jn-2}}
+\sum_{j=j_{\ti}+1}^{n-\ti+1}H\bra{\rrank{\ti}{j}+\ldots+\rrank{\ti}{j+2 \jn-2}}.
\end{multline}
Because
\begin{multline*}
\Phi\bra{\RRanks{\ti+1}}=\sum_{j=-2 \jn+3}^{n-\ti}H\bra{\rrank{\ti+1}{j},\ldots,\rrank{\ti+1}{j+2 \jn-2}}
=\sum_{j=-2 \jn+3}^{j_{\ti}-2 \jn+1}H\bra{\rrank{\ti}{j},\ldots,\rrank{\ti}{j+2 \jn-2}}+\\
+\sum_{j=j_{\ti}-2 \jn+2}^{j_{\ti}}H\bra{\rrank{\ti}{j},\ldots,\rrank{\ti}{j+2 \jn-1}}
+\sum_{j=j_{\ti}+2}^{n-\ti+1}H\bra{\rrank{\ti}{j},\ldots,\rrank{\ti}{j+2 \jn-2}}
\end{multline*}
we have
\begin{multline}
\label{eq:Ql}
\Phi\bra{\RRanksThree{l}}
=\sum_{j=-2 \jn+3}^{n-\ti}H\bra{\rrankThree{l}{j}+\ldots+\rrankThree{l}{j+2 \jn-2}}=\\
=\sum_{j=j_{\ti}-2 \jn+2}^{j_{\ti}}H\bra{\rrank{\ti}{l}+\ldots+\rrank{\ti}{j+2 \jn-1}}
+\sum_{j=j_{\ti}+1}^{n-\ti+1}H\bra{\rrank{\ti}{j}+\ldots+\rrank{\ti}{j+2 \jn-2}}.
\end{multline}
Similarly
\begin{equation}
\label{eq:Pl+1}
\Phi\bra{\RRanksTwo{l+1}}=\sum_{j=j_{\ti}-2 \jn+2}^{j_{\ti}}H\bra{\rrank{\ti}{l+1}+\ldots+\rrank{\ti}{j+2 \jn-2}}
+\sum_{j=j_{\ti}+1}^{n-\ti+1}H\bra{\rrank{\ti}{j}+\ldots+\rrank{\ti}{j+2 \jn-2}}
\end{equation}
and
\begin{equation}
\label{eq:Ql+1}
\Phi\bra{\RRanksThree{l+1}}
=\sum_{j=j_{\ti}-2 \jn+2}^{j_{\ti}}H\bra{\rrank{\ti}{l+1}+\ldots+\rrank{\ti}{j+2 \jn-1}}
+\sum_{j=j_{\ti}+1}^{n-\ti+1}H\bra{\rrank{\ti}{j}+\ldots+\rrank{\ti}{j+2 \jn-2}}.
\end{equation}
Taking together Equations (\ref{eq:Pl})-(\ref{eq:Ql+1}) and Subclaim \ref{clm:subadd-H} (where $a=\rrank{\ti}{l}$, $b=\rrank{\ti}{l+1}+\ldots+\rrank{\ti}{j+2 \jn-2}$, and $c=\rrank{\ti}{j+2 \jn-1}$) we obtain that
\begin{multline*}
\Phi\bra{\RRanksThree{l+1}}-\Phi\bra{\RRanksTwo{l+1}} = \sum_{j=j_{\ti}-2 \jn+2}^{j_{\ti}}\Bra{H\bra{\rrank{\ti}{l+1}+\ldots+\rrank{\ti}{j+2 \jn-1}} - H\bra{\rrank{\ti}{l+1}+\ldots+\rrank{\ti}{j+2 \jn-2}}} \leqslant\\
\leqslant \sum_{j=j_{\ti}-2 \jn+2}^{j_{\ti}}\Bra{H\bra{\rrank{\ti}{l}+\ldots+\rrank{\ti}{j+2 \jn-1}} - H\bra{\rrank{\ti}{l}+\ldots+\rrank{\ti}{j+2 \jn-2}}} = \Phi\bra{\RRanksThree{l}}-\Phi\bra{\RRanksTwo{l}},
\end{multline*}
which finishes the proof of Subclaim \ref{clm:ind-step}.
\end{proof}
By Subclaim \ref{clm:ind-step} we obtain that $\Phi\bra{\RRanksThree{l+1}}-\Phi\bra{\RRanksTwo{l+1}} \leqslant\Phi\bra{\RRanksThree{l}}-\Phi\bra{\RRanksTwo{l}}\leqslant \Phi\bra{\RRanks{\ti+1}}-\Phi\bra{\RRanks{\ti}}$ which finishes the inductive step in the proof of Inequality (\ref{eq:ineq-ind}).
In particular, we have that $\Phi\bra{\RRanksThree{j_{\ti}-\jn+1}}-\Phi\bra{\RRanksTwo{j_{\ti}-\jn+1}}\leqslant \Phi\bra{\RRanks{\ti+1}}-\Phi\bra{\RRanks{\ti}}$.
Colloquially speaking, we have shown that you can replace a certain prefix with zeros.
The second part of the proof of Claim \ref{clm:first} would be to show that you can additionally reset the suffix.
We will omit this part of the proof, however, because it is fully symmetrical with the first part.
\end{proof}
\begin{claim}
\label{clm:c-alpha-P-Q}
$\cost_{\ti} \leqslant \Phi\bra{\RanksThree}-\Phi\bra{\RanksTwo}$.
\end{claim}
\begin{proof}[Proof of Claim \ref{clm:c-alpha-P-Q}]
To simplify the notation, let us define $\HR{t}{u}\eqdef H\bra{\rrank{\ti}{t}+\ldots+\rrank{\ti}{u}}$ and $\HRP{t}{u}\eqdef H\bra{\rrank{\ti+1}{t}+\ldots+\rrank{\ti+1}{u}}$.
Then
\begin{multline}
\label{eq:P}
\Phi\bra{\RanksTwo}
=\sum_{j=-2 \jn+3}^{n-\ti+1}H\bra{\rankTwo{j}+\ldots+\rankTwo{j+2 \jn-2}}=\\
=\HR{j_{\ti}-\jn+1}{j_{\ti}-\jn+1}+\HR{j_{\ti}-\jn+1}{j_{\ti}-\jn+2}+\ldots+\HR{j_{\ti}-\jn+1}{j_{\ti}+\jn-1}
+\HR{j_{\ti}-\jn+2}{j_{\ti}+\jn}+\HR{j_{\ti}-\jn+3}{j_{\ti}+\jn}+\ldots+\HR{j_{\ti}+\jn}{j_{\ti}+\jn}
\end{multline}
We will divide the sum $\Phi\bra{\RanksTwo}$ into those elements that do not contain $\rrank{\ti}{j_0}+\rrank{\ti}{j_0+1}$ and are earlier than (colored red), those that contain $\rrank{\ti}{j_0}+\rrank{\ti}{j_0+1}$ (colored blue), and those that do not contain $\rrank{\ti}{j_0}+\rrank{\ti}{j_0+1}$ and lie later than (colored blue).
\begin{multline}
\label{eq:PH}
\Phi\bra{\RanksTwo}=\HRred{j_{\ti}-\jn+1}{j_{\ti}-\jn+1}+\HRred{j_{\ti}-\jn+1}{j_{\ti}-\jn+2}+\ldots+\HRred{j_{\ti}-\jn+1}{j_{\ti}-1}+\HR{j_{\ti}-\jn+1}{j_{\ti}}+\\
+\HRblue{j_{\ti}-\jn+1}{j_{\ti}+1}+\HRblue{j_{\ti}-\jn+1}{j_{\ti}+2}+\ldots+\HRblue{j_{\ti}-\jn+1}{j_{\ti}+\jn-1}+\HRblue{j_{\ti}-\jn+1}{j_{\ti}+\jn}
+\HRblue{j_{\ti}-\jn+2}{j_{\ti}+\jn}+\ldots+\HRblue{j_{\ti}}{j_{\ti}+\jn}+\\
+\HR{j_{\ti}+1}{j_{\ti}+\jn}+\HRgreen{j_{\ti}+2}{j_{\ti}+\jn}+\HRgreen{j_{\ti}+3}{j_{\ti}+\jn}+\ldots+\HRgreen{j_{\ti}+\jn}{j_{\ti}+\jn}.
\end{multline}
We will write analogously $\Phi\bra{\RanksThree}$.
\begin{multline}
\label{eq:Q}
\Phi\bra{\RanksThree}
= \sum_{j=-2 \jn+3}^{n-\ti}H\bra{\rankThree{j}+\ldots+\rankThree{j+2 \jn-2}}=\\
=\HRP{j_{\ti}-\jn+1}{j_{\ti}-\jn+1}+\HRP{j_{\ti}-\jn+1}{j_{\ti}-\jn+2}+\ldots+\HRP{j_{\ti}-\jn+1}{j_{\ti}+\jn-1}
+\HRP{j_{\ti}-\jn+2}{j_{\ti}+\jn-1}+\ldots+\HRP{j_{\ti}+\jn-1}{j_{\ti}+\jn-1}
\end{multline}
As previously, we will divide the sum $\Phi\bra{\RanksThree}$ into those elements that do not contain $\rrank{\ti}{j_0}+\rrank{\ti}{j_0+1}$ and are earlier than (colored red), those that contain $\rrank{\ti}{j_0}+\rrank{\ti}{j_0+1}$ (colored blue), and those that do not contain $\rrank{\ti}{j_0}+\rrank{\ti}{j_0+1}$ and lie later than (colored blue).
\begin{multline}
\label{eq:QH}
\Phi\bra{\RanksThree}=\HRred{j_{\ti}-\jn+1}{j_{\ti}-\jn+1}+\HRred{j_{\ti}-\jn+1}{j_{\ti}-\jn+2}+\ldots+\HRred{j_{\ti}-\jn+1}{j_{\ti}-1}+\\
+\HRblue{j_{\ti}-\jn+1}{j_{\ti}+1}+\HRblue{j_{\ti}-\jn+1}{j_{\ti}+2}+\ldots+\HRblue{j_{\ti}-\jn+1}{j_{\ti}+\jn-1}+\HRblue{j_{\ti}-\jn+1}{j_{\ti}+\jn}
+\HRblue{j_{\ti}-\jn+2}{j_{\ti}+\jn}+\ldots+\HRblue{j_{\ti}}{j_{\ti}+\jn}+\\
+\HRgreen{j_{\ti}+2}{j_{\ti}+\jn}+\HRgreen{j_{\ti}+3}{j_{\ti}+\jn}+\ldots+\HRgreen{j_{\ti}+\jn}{j_{\ti}+\jn}.
\end{multline}
If we compare formulas (\ref{eq:PH}) and (\ref{eq:QH}), it can be seen that the difference equals
\begin{multline}
\label{eq:QH-PH}
\Phi\bra{\RanksThree}-\Phi\bra{\RanksTwo}=\HR{j_{\ti}-\jn+1}{j_{\ti}+\jn}-\HR{j_{\ti}-\jn+1}{j_{\ti}}-\HR{j_{\ti}+1}{j_{\ti}+\jn}=\\
= \bra{R_L + R_R} \log_2 \bra{R_L + R_R} - R_L \log_2 R_L - R_R \log_2 R_R,
\end{multline}
where $R_L\eqdef\rrank{\ti}{j_{\ti}-\jn+1}+\ldots+\rrank{\ti}{j_{\ti}}$ and $R_R\eqdef\rrank{\ti}{j_{\ti}+1}+\ldots+\rrank{\ti}{j_{\ti}+\jn}$.
The rest of the proof comes down to the following statement.
\begin{claim}
\label{clm:logab}
For the arbitrary chosen non-negative real numbers $a\leqslant b$ there is
$$
2a \leqslant \bra{a+b} \log_2\bra{a+b} - a \log_2 a - b \log_2 b.
$$
\end{claim}
\begin{proof}[Proof of Subclaim \ref{clm:logab}]
From the fact that for any real constant $a \geqslant 0$ the function $f\bra{x} = \bra{x + a} \log_2 \bra{x + a} - x \log_2 x$ is non-decreasing (see Subclaim \ref{clm:increasing}) and $a \leqslant b$ we obtain that
$$
\bra{a+a}\log_2 \bra{a+a} - a \log_2 a \leqslant \bra{a+b} \log_2 \bra{a+b} - b \log_2 b
$$
and in consequence
$$
2a = 2a \log_2 2a - a \log_2 a - a \log_2 a \leqslant \bra{a+b} \log_2 \bra{a+b} - a \log_2 a - b \log_2 b.
$$
\end{proof}
Without loss of generality, we can assume that $R_L\leqslant R_R$.
Then, taking together definition of~$\cost_{\ti}$, Subclaim \ref{clm:logab} and Equation (\ref{eq:QH-PH}) we obtain that
\begin{multline*}
\cost_{\ti} = \min\set{R_L,R_R}=R_L\leqslant \\
\leqslant \bra{R_L + R_R} \log_2 \bra{R_L + R_R} - R_L \log_2 R_L - R_R \log_2 R_R = \Phi\bra{\RanksThree}-\Phi\bra{\RanksTwo}.
\end{multline*}
\end{proof}

Taking together Claims \ref{clm:first} and \ref{clm:c-alpha-P-Q} we obtain that
$$
\cost_{\ti} \leqslant \Phi\bra{\RanksThree}-\Phi\bra{\RanksTwo} \leqslant \Phi\bra{\RRanks{\ti+1}}-\Phi\bra{\RRanks{\ti}}
$$
which ends the proof of Inequality (\ref{eq:ineq-ind}) and finally Theorem \ref{thm:main}.
\end{proof}

 \section{The connection between \Recolor algorithm and the blocks of $\PartB_j$}\label{app:connection}
\begin{proof}[Proof of Observation~\ref{obs:algToBlR}]
 First observe that $\set{\interval_j} \cup \intervals{j-1}{\Block^{j-1}_{s(j)} \cup \Block^{j-1}_{s(j)+1},\ldots \Block^{j-1}_{s(j)+\gamma}}=\intervals{j}{\Block^{j-1}_{s(j)} \cup \Block^{j-1}_{s(j)+1},\ldots \Block^{j-1}_{s(j)+\gamma}}.$
Recall that $\intervals{j}{\Block^{j-1}_{s(j)}  \cup \ldots \cup \Block^{j-1}_{s(j)+\gamma}}$ is a set of consecutive intervals in $\Intervals_j$ according to $\linearOrder$ order. If $s(j)+\gamma > b_j$, then $\intervals{j}{\Block^{j-1}_{s(j)}  \cup \ldots \cup \Block^{j-1}_{s(j)+\gamma}}$ is a suffix of $\Intervals_j$ containing $\interval_j$. Thus, $\infix{\Intervals_j}{\interval_j}{\intervalTwo^{(R)}} \subseteq \suffix{\Intervals_j}{\interval_j} \subseteq \intervals{j}{\Block^{j-1}_{s(j)}  \cup \ldots \cup \Block^{j-1}_{s(j)+\gamma}}$, what completes the proof in this case. 

So let us assume $s(j)+\gamma \leq b_j$. Let  $\interval^{(last)}$ be the last interval of $\intervals{j}{\Block^{j-1}_{s(j)}  \cup \ldots \cup \Block^{j-1}_{s(j)+\gamma}}$
in $\linearOrder$ order.
  Recall that $\infix{\Intervals_j}{\interval_j}{\intervalTwo^{(R)}}$ is the shortest infix of $\Intervals_j$ starting in $\interval_j$, such that one of the two following conditions hold:
\begin{itemize}

 \item $\infix{\Intervals_j}{\interval_j}{\intervalTwo^{(R)}}$ is disjoint from the successor of $\intervalTwo^{(R)}$
 \item for intervals $\interval,\intervalTwo$
 such that $\infix{\Intervals_j}{\interval_j}{\interval}=k+1$ and $\infix{\Intervals_j}{\intervalTwo}{\intervalTwo^{(R)}}=k+1$ there is a set $\IntervalsFour$ of $k^2-1$ disjoint unit intervals such that $\set{\interval} < \IntervalsFour < \set{ \intervalTwo }$ and $\fClique{\Intervals_j \cup \IntervalsFour} \leq k$. 

\end{itemize}
Thus, to prove the observation, it is sufficient to show the following claim.
\begin{claim}\label{clm:inBetween}
There is a set $\IntervalsFour'$ of $k^2-1$ disjoint unit intervals such that
\begin{itemize}
\item $\overline{\interval} < \IntervalsFour' < \overline{\intervalTwo}$ for $\overline{\interval},\overline{\intervalTwo} \in \Intervals_j$ such that $\infix{\Intervals_j}{\interval_j}{\overline{\interval}}=k+1$ and $\infix{\Intervals_j}{\overline{\intervalTwo}}{\interval^{(last)}}=k+1$ 
\item $\fClique{\Intervals_j \cup \IntervalsFour'} \leq k$.
\end{itemize}
\end{claim}
\begin{proof}
By Observation~\ref{obs:unitCli} it holds that $\infix{\Intervals_j}{\overline{\intervalTwo}}{\interval^{(last)}}$ intersects a suffix of $\Block^{j-1}_{s(j)} \cup \Block^{j-1}_{s(j)+1} \cup \ldots \cup \Block^{j-1}_{s(j)+\gamma}$ containing at most $(k+1)k$ spans (this is beacause each of $(k+1)$ intervals of $\infix{\Intervals_j}{\overline{\intervalTwo}}{\interval^{(last)}}$ intersects at most $k$ spans). Similarily, $\infix{\Intervals_j}{\interval_j}{\overline{\interval}}$ intersects a prefix of $\Block^{j-1}_{s(j)} \cup \Block^{j-1}_{s(j)+1} \cup \ldots \cup \Block^{j-1}_{s(j)+\gamma}$ consisting of at most $(k+1)k$ spans. Let $\lambda=s(j)+(k+1)k$ and $\mu=\gamma+1-2k(k+1)=4k^3+4k$. Let $\Spanset=\Block^{j-1}_{\lambda} \cup \Block^{j-1}_{\lambda+1} \cup\ldots \cup \Block^{j-1}_{\lambda+\mu-1}$ be the union of $\mu$ consecutive blocks starting in $\Block^{j-1}_{\lambda}$. Then both $\infix{\Intervals_j}{\overline{\intervalTwo}}{\interval^{(last)}}$ and $\infix{\Intervals_j}{\interval_j}{\overline{\interval}}$ do not intersect any span in $\Spanset$. Let $\Spanel^{(first)}$ and $\Spanel^{(last)}$ be the first and the last span in $\Spanset$ according to $<$ order. It is sufficient to show that there is a set $\IntervalsFour'$ of $k^2-1$ disjoint unit intervals such that $\set{\Spanel^{(first)}} \linearOrder \IntervalsFour' < \set{\Spanel^{(last)}}$ and $\fClique{\Intervals_j \cup \IntervalsFour'} \leq k$.  Since there are no two consecutive active blocks among $\Block^{j-1}_{\lambda}, \Block^{j-1}_{\lambda+1}, \ldots \Block^{j-1}_{\lambda+\mu-1}$, the total number of passive blocks in $\Spanset$ is at least $\floor{\mu/2}$. We now group together consecutive passive blocks. In other words on all passive blocks we define the equivalence relation of being consecutive, and the groups are the equivalence classes. Let us call these groups $G_1, \ldots G_l$, where $l$ is the number of groups. Let also $d_i=|G_i|$ for $i \in \num{l}$ be the number of passive blocks (elements) in the $i$'th group. 

First observe that for each group $G_i$ there exists a set $\IntervalsFour'_i$, such that
\begin{itemize}
 \item $\IntervalsFour'_i \subseteq \Intervals_{n'} \setminus \Intervals_j$ (intuitively, $\IntervalsFour'_i$ consists only of intervals presented in the future)
  \item every interval of $\IntervalsFour'_i$ intersects a span in $G_i$
 \item intervals of $\IntervalsFour'_i$ are pairwise disjoint
 \item $|\IntervalsFour'| \geq \ceil{d_i/(2k)}$
\end{itemize}
Such set $\IntervalsFour_i'$ can be taken greedily. Since $G_i$ consist only of passive spans, there must be a future interval $\interval_{j'}$, $j' > j$, intersecting the first span. We take $\interval_{j'}$ to $\IntervalsFour'$. The future intervals intersecting $\interval_{j'}$, based on Observarion~\ref{obs:unitCli}, intersect at most $2k$ first spans of $G_i$ in $<$ order. We remove the first $2k$ spans from $G_i$. We continue the above procedure until there are spans left in $G_i$. As a result, we have at least $\ceil{d_i/(2k)}$ intervals in $\IntervalsFour'_i$.   

Observe that sets $\IntervalsFour'_i$ assigned to different groups do not intersect, as there is an active block between each two groups. So the union $\bigcup_{i=1}^{l} \IntervalsFour'_i$ is an independent set. The first and the last of interval of $\bigcup_{i=1}^{l} \IntervalsFour'$ (in $<$ order) might, however, not fit between $\Spanel^{(first)}$ and $\Spanel^{(last)}$ as claimed. Thus we declare the set $\IntervalsFour'$ to be $\bigcup_{i=1}^{l} \IntervalsFour'$ without the first and last interval. Since all intervals of $\IntervalsFour'$ are future intervals, we get that $\fClique{\Intervals_j \cup \IntervalsFour'} \leq k$. 
It remains to bound the size of $\IntervalsFour'$ from below: $|\IntervalsFour'| \geq \sum_{i=1}^{l}\ceil{d_i/(2k)}-2  \geq \sum_{i=1}^{l} d_i/(2k) -2 \geq \floor{\mu/2}/(2k)-2=k^2-1$ disjoint unit intervals, and hence the claim holds.  
\end{proof}\end{proof}
 \section{Maintaining the invariants}
\label{app:invariants}
In this section we describe in detail how Invariants 1-5 from Section~\ref{sec:mainres} are maintained for dynamically changing partition $\PartB_j$.  
We assume that the invariants are satisfied in step $j-1$ and time $\ti(j-1)$. We have already defined (in Section~\ref{sec:mainres}) the first jump corresponding to step $j$. To account for insertion of $\interval_j$ we merged two consecutive ranks assigned to $\Block^{j-1}_{s(j)}$. We first show that the invariants hold after this merge. It is clear that the blocks are still assigned consecutive ranks, so Invariants 1-3 hold. Observe that the only block whose number of assigned ranks changed is $\Block^{j-1}_{s(j)}$. The number of assigned ranks decreased for $\Block^{j-1}_{(s(j))}$, but $\intervals{j}{\Block^{j-1}_{s(j)}}=\intervals{j-1}{\Block^{j-1}_{s(j)}}+1$, hence Invariant 5 remains valid (for now we temporarily assume that all passive blocks remain passive; if they turn active in turn $j$ we activate them with further jumps later on).  
As for the Invariant~4, it clearly remains valid for the blocks $\Block^{j-1}_i \in \PartB_{j-1}$ such that $\interval_j$ is disjoint with $\Block^{j-1}_i$. Let us consider a block $\Block^{j-1}_i \in \PartB_{j-1}$ intersecting $\interval_j$.
Observe, that $\interval_j$ only intersects the spans of $\Spanset_{n'}$ that are passive in step $j-1$. Thus, the block $\Block^{j-1}_i=\set{\Spanel}$ is a passive singleton block.
But then $\size{ \intervals{j}{\Block^{j-1}_{i}} }=\size{ \intervals{j}{\set{\Spanel}} } \leq k$, which is the lower bound for all ranks in $\RRanks{\ti(j-1)+1}$, so Invariant 4 holds for $\Block^{j-1}_i$. 

We move on to describing the process of merging the appriopriate blocks of $\PartB_{j-1}$ that become active in step $j$. We simulate this process in the \Frog game as well.
Let $\widehat{\Spanset}_j=\Spanset_j \setminus \Spanset_{j-1}$ be the set of spans of the newly added $k$-cliques. In other words, $\widehat{\Spanset}_j$ is the set of spans of those $k$-cliques in $\Intervals_j$ that contain the new interval $\interval_j$. 
The spans of $\widehat{\Spanset}_j$ correspond to passive singleton blocks in $\PartB_{j-1}$, which become active in step $j$, and thus have to be merged according to the definition of $\PartB_j$.
 Let $\Block^{j-1}_{pre}$ be the predecessor of $\widehat{\Spanset}_j$ in $\PartB_{j-1}$ in $<$ order and $\Block^{j-1}_{succ}$ be the successor of $\widehat{\Spanset}_j$.
The partition $\PartB_j$ is then defined by one of the following cases:
 \setcounter{equation}{0}
 \begin{subequations}
\begin{align}[left ={\PartB_j = \empheqlbrace}] 
\{ 
 \Block^{j-1}_1, \ldots, \Block^{j-1}_{pre-1},\Block^{j-1}_{pre},& \widehat{\Spanset}_j, \Block^{j-1}_{succ},\Block^{j-1}_{succ+1}, \ldots, \Block^{j-1}_{b_{j-1}}
  \} \text{ or } \label{case1} \\
 \{ \Block^{j-1}_1, \ldots, \Block^{j-1}_{pre-1},\Block^{j-1}_{pre} \cup &\widehat{\Spanset}_j, \Block^{j-1}_{succ},\Block^{j-1}_{succ+1}, \ldots, \Block^{j-1}_{b_{j-1}} \} \text{ or } \label{case2}\\
 \{ \Block^{j-1}_1, \ldots, \Block^{j-1}_{pre-1},\Block^{j-1}_{pre}, &\widehat{\Spanset}_j \cup \Block^{j-1}_{succ},\Block^{j-1}_{succ+1}, \ldots, \Block^{j-1}_{b_{j-1}} \} \text{ or  } \label{case3}\\
 \{ \Block^{j-1}_1, \ldots, \Block^{j-1}_{pre-1},\Block^{j-1}_{pre} \cup &\widehat{\Spanset}_j \cup \Block^{j-1}_{succ},\Block^{j-1}_{succ+1}, \ldots, \Block^{j-1}_{b_{j-1}} \}, \label{case4}
\end{align}
 \end{subequations}
 where the above cases are determined by the following conditions:
 \begin{itemize}
  \item Case~\ref{case1} holds iff both blocks $\Block^{j-1}_{pre}$ and $\Block^{j-1}_{succ}$ are passive
  \item Case~\ref{case2} holds iff block $\Block'^{j-1}_{pre}$ is active and $\Block'^{j-1}_{succ}$ is passive
  \item Case~\ref{case3} holds iff block $\Block'^{j-1}_{pre}$ is passive and $\Block'^{j-1}_{succ}$ is active
  \item Case~\ref{case4} holds iff both blocks $\Block'^{j-1}_{pre}$ and $\Block'^{j-1}_{succ}$ are active
 \end{itemize}

Our goal now is to define $\ti(j)$ and $\RRanks{\ti(j)}$ as to reflect the change in the block partition $\PartB_j$, depending on the case that occurs (Cases~(\ref{case1})-(\ref{case4})). First of all, all passive blocks in $\PartB_{j-1}$ containing spans of $\widehat{\Spanset}_j$ become active, and thus they are merged into one single block. Hence, we need to merge all ranks assigned to these blocks. In the jumping sequence for \Frog game we add at most $k\size{\widehat{\Spanset}_j}+\size{\widehat{\Spanset}_j} -1$ jumps to do so. We then use some jumps to merge with the rank for $\Block^{j-1}_{pre}$ and the rank for $\Block^{j-1}_{succ}$ if necessary. Let us describe this more formally. Let $\rrank{\ti(j-1)}{q}$ be the last rank assigned to $\Block^{j-1}_{pre}$ and $\rrank{\ti(j-1)}{q'}$ be the first rank assigned to $\Block^{j-1}_{succ}$. Let us consider the rank sequence in time $\ti(j-1)+1$, that is $$\RRanks{\ti(j-1)+1}=\set{ 
 \rrank{\ti(j-1)}{1},
 \ldots,\rrank{\ti(j-1)}{q},
 \rrank{\ti(j-1)}{q+1},
\ldots
 \rrank{\ti(j-1)}{q'-1},
 \rrank{\ti(j-1)}{q'},
 \ldots,
 \rrank{\ti(j-1)}{b_{j-1}} 
 },$$ where ranks $\rrank{\ti(j-1)}{q+1},
 \rrank{\ti(j-1)}{q+2},
 \ldots
 \rrank{\ti(j-1)}{q'-1}$ are assigned to passive singleton blocks corresponding to spans of $\widehat{\Spanset}_j$.
 We now merge the ranks assigned to blocks corresponding to $\widehat{\Spanset}_j$, so that we get one rank assigned to the temporary active block $\widehat{\Spanset}_j$. We modify the $f$ function accordingly.
 We add $q'-q-2$ jumps to the jumping sequence, each jump equal to position $q+1$, thus obtaining 
$$\RRanks{\ti(j-1)+q'-q-1}=\set{ 
 \rrank{\ti(j-1)}{1},
 \ldots,\rrank{\ti(j-1)}{q},
 \sum_{i=q+1}^{q'-1} \rrank{\ti(j-1)}{i},
 \rrank{\ti(j-1)}{q'},
 \ldots,
 \rrank{\ti(j-1)}{b'_{j-1}-q+1} 
 }.$$\\
 We observe the following.
\begin{equation}\label{eq:invSize}
 \size{\intervals{j}{\widehat{\Spanset}_j}}=1+\size{\intervals{j-1}{\widehat{\Spanset}_j}} \leq 1+\sum_{\Spanel \in \widehat{\Spanset}_j} \size{\intervals{j-1}{\Spanel}} \leq \sum_{i=q+1}^{q'-1} \rrank{\ti(j-1)}{i}.
\end{equation} The last inequality holds because each block corresponding to $\Spanel \in \widehat{\Spanset}_j$ is assigned at least two ranks, both bounding $\size{\intervals{j-1}{\Spanel}}$ from above, and both greater or equal to $k$.

Consired Case~(\ref{case1}). In this case we set $\ti(j)=\ti(j-1)+q'-q-1$ and we are done. The invariants hold because of the following. 
We modified the $f$ function by merging ranks $\rrank{\ti(j-1)}{i}$ for $i \in [ q+1,\ldots q'-1 ]$ and assigning the sum to $pre+1$, which is the index of the new active block $\widehat{\Spanset}_j$. Thus, the new active block is assigned precisely one rank and by Equation~\ref{eq:invSize} invariant 4 is satisfied for $\widehat{\Spanset}_j$.

Consider Case~(\ref{case2}). In this case the active block $\Block^{j-1}_{pre}$ is assigned precisely one rank $\rrank{\ti(j-1)}{q}$, i.e., $f(q)=pre$. Also the new active block $\widehat{\Spanset}_j$ is assigned precisely one rank, i.e., $\sum_{i=q+1}^{q'-1} \rrank{\ti(j-1)}{i}$.  We apply one more jump to merge ranks $\rrank{\ti(j-1)}{q}$ and $\sum_{i=q+1}^{q'-1} \rrank{\ti(j-1)}{i}$ in $\RRanks{\ti(j-1)+q'-q-1}$, thus obtaining 
$$\RRanks{\ti(j-1)+q'-q}=\set{ 
 \rrank{\ti(j-1)}{1},
 \ldots,
 \rrank{\ti(j-1)}{q-1},
 \sum_{i=q}^{q'-1} \rrank{\ti(j-1)}{i},
 \rrank{\ti(j-1)}{q'},
 \ldots,
 \rrank{\ti(j-1)}{b_{j-1}} 
 }.$$
We modify the assignment function so that the (precisely one) rank $\sum_{i=q}^{q'-1} \rrank{\ti(j-1)}{i}$ is assigned to $\Block^{j-1}_{pre} \cup \widehat{\Spanset}_j$. 
We set $\ti(j)=\ti(j-1)+q'-q$ and we are done. Since $\intervals{j-1}{\Block^{j-1}_{pre}}=\intervals{j}{\Block^{j-1}_{pre}}$, we get: $$\intervals{j}{\Block^{j-1}_{pre} \cup \widehat{\Spanset}_j} \leq \intervals{j}{\Block^{j-1}_{pre}} + \intervals{j}{\widehat{\Spanset}_j} \leq \sum_{i=q}^{q'-1} \rrank{\ti(j-1)}{i}.$$ Thus the invariants hold for the newly merged block $\Block^{j-1}_{pre} \cup \widehat{\Spanset}_j$ as well. Case~(\ref{case3}) is analogous to Case~(\ref{case2}). 
We are left with Case~(\ref{case4}). We then declare $\ti(j)=\ti(j-1)+q'-q+1$ and we execute one more jump to obtain the rank sequence $$\RRanks{\ti(j-1)+q'-q+1}=\set{ 
 \rrank{\ti(j-1)}{1},
 \ldots,\rrank{\ti(j-1)}{q-1},
 \sum_{i=q}^{q'} \rrank{\ti(j-1)}{i},
 \rrank{\ti(j-1)}{q'+1},
 \ldots,
 \rrank{\ti(j-1)}{b_{j-1}} 
 }.$$
The rank $\sum_{i=q}^{q'} \rrank{\ti(j-1)}{i}$ is assigned to the new active block $\Block^{j-1}_{pre} \cup \widehat{\Spanset}_j \cup \Block^{j-1}_{succ}$. 
The argument that the invariants hold is analogous to the previous case.
 
\bibliographystyle{plainurl}

\end{document}